\documentclass[aps,prx,reprint,preprintnumbers,showpacs,showkeys,superscriptaddress]{revtex4-2}
\usepackage{amsmath,amssymb,amsthm,mathtools,mathrsfs}
\usepackage{srcltx}
\usepackage[colorlinks=true,citecolor=blue,linkcolor=blue]{hyperref}
\usepackage[pdftex]{graphicx}
\usepackage{bm,braket,comment}
\usepackage{amsmath, amssymb,amsthm,mathtools}
\usepackage{xcolor}

\usepackage{longtable}
\usepackage{enumerate}

\newtheorem{theorem}{Theorem}
\newtheorem{definition}{Definition}
\newtheorem{lemma}{Lemma}
\newtheorem{corollary}{Corollary}
\newtheorem{proposition}{Proposition}

\makeatletter
\newcounter{savesection}
\newcounter{apdxsection}
\renewcommand\appendix{\par
  \setcounter{savesection}{\value{section}}%
  \setcounter{section}{\value{apdxsection}}%
  \setcounter{subsection}{0}%
  \gdef\thesection{\@Alph\c@section}}
\newcommand\unappendix{\par
  \setcounter{apdxsection}{\value{section}}%
  \setcounter{section}{\value{savesection}}%
  \setcounter{subsection}{0}%
  \gdef\thesection{\@arabic\c@section}}
\makeatother

\begin{document}


\title{Characterizing symmetry-protected thermal equilibrium by work extraction}


\author{Yosuke Mitsuhashi}
\affiliation{Department of Applied Physics, The University of Tokyo, Tokyo 113-8656, Japan}

\author{Kazuya Kaneko}
\affiliation{Department of Applied Physics, The University of Tokyo, Tokyo 113-8656, Japan}

\author{Takahiro Sagawa}
\affiliation{Department of Applied Physics, The University of Tokyo, Tokyo 113-8656, Japan}
\affiliation{Quantum-Phase Electronics Center (QPEC), The University of Tokyo, Tokyo 113-8656, Japan}




\begin{abstract}
The second law of thermodynamics states that work cannot be extracted from thermal equilibrium,
whose quantum formulation is known as complete passivity;
A state is called completely passive if work cannot be extracted from any number of copies of the state by any unitary operations.
It has been established that a quantum state is completely passive if and only if it is a Gibbs ensemble.
In physically plausible setups, however, the class of possible operations is often restricted by fundamental constraints such as symmetries imposed on the system. 
In the present work, we investigate the concept of complete passivity under symmetry constraints.
Specifically, we prove that a quantum state is completely passive under a symmetry constraint described by a connected compact Lie group, if and only if it is a generalized Gibbs ensemble (GGE) including conserved charges associated with the symmetry.
Remarkably, our result applies to non-commutative symmetry such as $SU(2)$ symmetry, suggesting an unconventional extension of the notion of GGE.
Furthermore, we consider the setup where a quantum work storage is explicitly included, and prove that the characterization of complete passivity remains unchanged.
Our result extends the notion of thermal equilibrium to systems protected by symmetries, and would lead to flexible design principles of quantum heat engines and batteries.
Moreover, our approach serves as a foundation of the resource theory of thermodynamics in the presence of symmetries.
\end{abstract}



\maketitle

\section{Introduction}
The second law of thermodynamics, also known as \textit{Kelvin's principle}, dictates that a positive amount of work can never be extracted by any cyclic operation from a single heat bath at a uniform temperature \cite{Callen1985}.
Ever since its establishment in the nineteenth century, the second law has served as the most fundamental constraint on our capability of energy harvesting, which prohibits the perpetual motion of the second kind. 
In recent years, the frontier of thermodynamics is extended to the quantum regime due to the development of quantum technologies.
Quantum heat engines have been experimentally realized by quantum technologies such as ion traps \cite{An2015, Lindenfels2019}, superconducting qubits \cite{Cottet2017, Masuyama2018, Naghiloo2020}, and NMR \cite{Batalhao2014, Camati2016}, which have opened new opportunities of power generation by utilizing quantum effects.

On the theory side, quantum information theory sheds new light on quantum thermodynamics.
In particular, an information-theoretic framework called resource theory has attracted much attention \cite{Horodecki2013, Brandao2013, Aberg2013, Brandao2015, Weilenmann2016, Faist2018, Faist2019, Gour2015, Chitambar2019, Lostaglio2019, Sagawa2020}, which identifies work with resources and thermal equilibrium states with resource-free states.
From this perspective, a concept called \textit{passive} state plays a key role \cite{Pusz1978, Lenard1978}, from which positive work cannot be extracted by any unitary operation.
It is known, however, that one can extract a positive amount of work from multiple copies of a certain passive state, and thus the concept of passivity is not sufficient to characterize thermal equilibrium from which energy harvesting should be strictly prohibited.
The full characterization of thermal equilibrium is given by \textit{complete passivity}: A positive amount of work cannot be extracted from any number of copies of a completely passive state.
It is known that a state is completely passive \textit{if and only if} it is a Gibbs ensemble, which suggests that complete passivity provides a physically meaningful, as well as information-theoretically accurate, definition of thermal equilibrium.

In the above approach to characterize thermal equilibrium, a central assumption is that all unitary operations are allowed for work extraction. 
In real physical situations, on the other hand, several constraints are often imposed on possible unitary operations, which often make the class of physically plausible unitary operations strictly smaller than all unitary operations.
Among such constraints, we here focus on the symmetry of quantum systems.
There are various kinds of symmetry and the corresponding conservation laws \cite{Sakurai1985}, such as $U(1)$ symmetry and particle number conservation, $SU(2)$ symmetry and spin (magnetization) conservation, $\mathbb{Z}_2$ symmetry and parity conservation.

Once the class of possible unitary operations is restricted by such symmetries, thermal equilibrium states are no longer necessarily Gibbs ensembles.
In fact, there are some non-Gibbs ensembles from which one cannot extract positive work by symmetry-respecting unitaries.
In other words, a broader class of states looks like effective thermal equilibrium, as long as the symmetry is respected (see Fig.~\ref{fig1}).
We call this extended notion of thermal equilibrium as \textit{symmetry-protected} thermal equilibrium. 
In terms of resource theory~\cite{Takagi2019}, the above observation implies that the class of free states is expanded if the class of free operations is restricted. 
Therefore, the conventional Gibbs ensemble would be insufficient to represent all symmetry-protected thermal equilibrium states.  
Then, a natural question raised is: What are concrete expressions of symmetry-protected thermal equilibrium states?  More specifically, how should the notion of complete passivity be extended if only symmetry-respecting unitaries are considered?

\begin{figure}
\begin{center}
\includegraphics[width=\columnwidth]{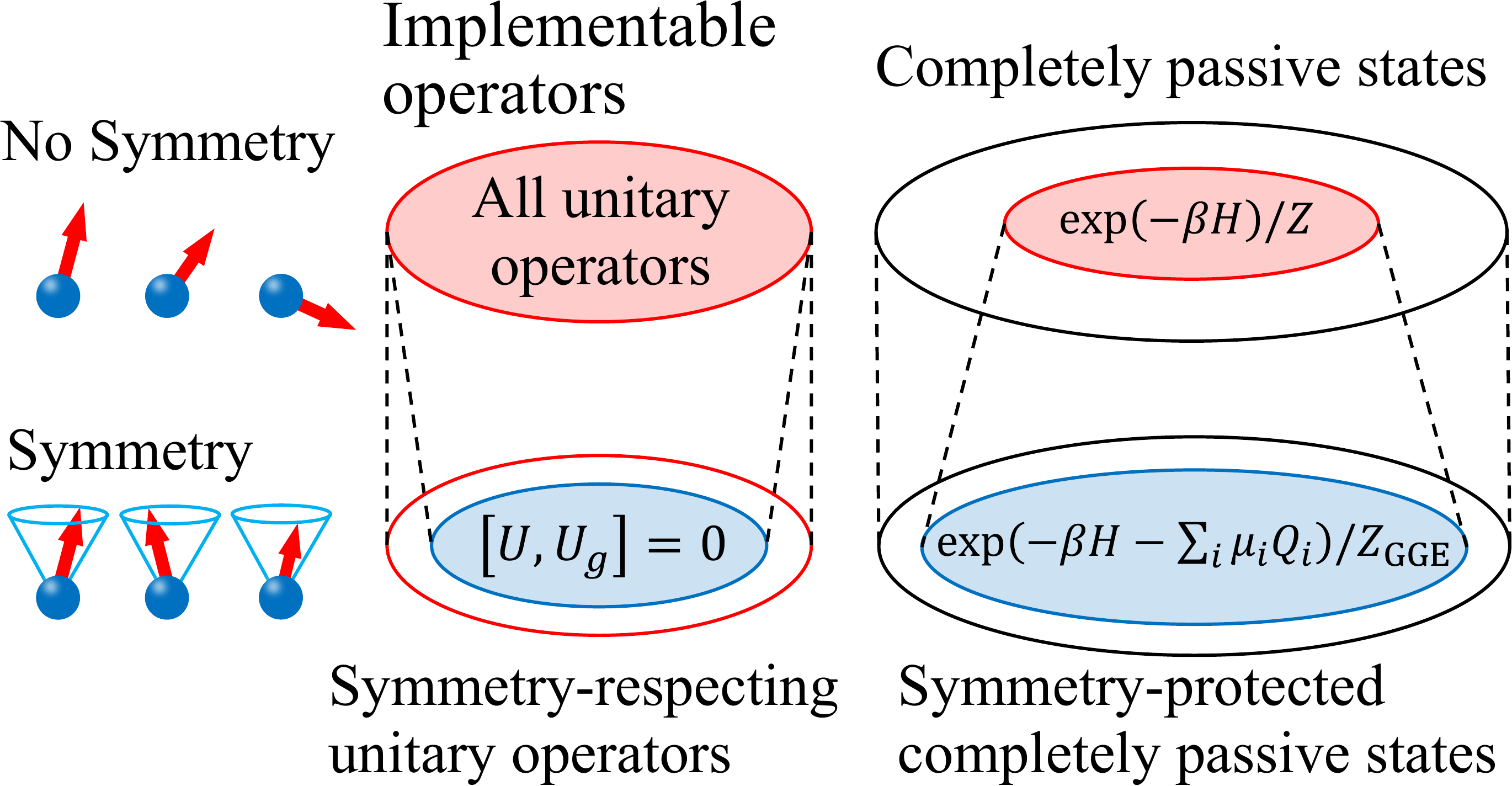}
\caption{Symmetry constraints on unitary operations and completely passive states. 
When the class of allowed operations is restricted, the class of quantum states from which positive work cannot be extracted is expanded.
This leads to the question to identify the states that behave as effective thermal equilibrium under symmetry constraints.
In this paper, we completely identify the class of such effective thermal equilibrium states, which we name as symmetry-protected thermal equilibrium states, by proving that those states are always given by GGEs of the form (\ref{GGE_main}), including the case that the conserved charges are non-commutative.  }
\label{fig1}
\end{center}
\end{figure}

In this paper, we answer this question by proving that 
a quantum state is completely passive under a symmetry constraint,
\textit{if and only if} it is a generalized Gibbs ensemble (GGE) that involves conserved charges associated with the symmetry.
This result is applicable to non-commutative symmetry such as $SU(2)$ symmetry, which leads to an unconventional extension of the notion of GGE.

More explicitly, consider a continuous symmetry  described by a connected compact Lie group, and let $Q_i$'s be its conserved charges such as the particle number operator or the spin operators.  
Then, what we prove is that any completely passive state is given in the form
\begin{equation}
    \rho_{\rm GGE} := \frac{1}{Z_{\rm GGE}} \exp\left( -\beta H - \sum_i \mu_i Q_i\right),
    \label{GGE_main}
\end{equation}
where $Z_\mathrm{GGE} :=\mathrm{tr}(\exp(-\beta H-\sum_i \mu_i Q_i))$ is the generalized partition function, $\beta \geq 0$ is the inverse temperature, and $\mu_i$'s are generalized ``chemical potentials.''
In the commutative $U(1)$ case, the GGE is reduced to the conventional grand canonical ensemble. 
Furthermore, we consider the setup where a quantum storage that stores work is explicitly introduced \cite{Aberg2014, Skrzypczyk2014}, and prove that the above characterization of symmetry-protected complete passivity remains unchanged from the setup where work is treated as a classical variable.

Our result establishes that any state other than the GGE is not completely passive and thus cannot be regarded as thermal equilibrium in terms of work extraction.
From the experimental point of view, this would lead to a more flexible design principle of quantum heat engines \cite{An2015, Lindenfels2019, Cottet2017, Masuyama2018, Naghiloo2020, Batalhao2014, Camati2016} and quantum batteries \cite{Alicki2013, Hovhannisyan2013, Campaioli2017, Andolina2019}.
For example, if one attempts to construct a heat engine operating under a symmetry constraint, the equilibrium state can be more flexibly chosen than in the conventional case, which is not necessarily the Gibbs ensemble but only needs to be a GGE.
From the theoretical point of view, our result would serve as a foundation of the resource theory of thermodynamics under symmetry constraints, as our result specifies the free states of such a resource theory.

In the context of thermalization, the GGE has been investigated as a state describing equilibration in integrable systems~\cite{Rigol2007, Rigol2006, Cazalilla2006, Vidmar2016}.
Most of previous works consider the case where conserved charges are commutative with each other (but see also Refs.~\cite{Halpern2020, Fukai2020}).
A non-commutative extension of the GGE has been proposed by Yunger Halpern \textit{et al.} in Ref.~\cite{Halpern2016}, and our result supports that the expression proposed by them is a proper form of the non-commutative GGE.
We emphasize, however, that our setup is fundamentally different from the setup of Ref.~\cite{Halpern2016};
In the present paper, we consider the purely energetic work extraction (instead of the chemical work extraction)
under a symmetry constraint that is imposed only on the system (instead of including charge storages).
Our setup is physically plausible given that heat engines and external systems are often very different (e.g., matter and light) \cite{An2015, Lindenfels2019, Cottet2017, Masuyama2018, Naghiloo2020, Batalhao2014, Camati2016}, where it would be natural to suppose that symmetries are imposed only on the system of interest.

This paper is organized as follows.
In Sec.~\ref{sec:setup}, we give the definition of symmetry-protected complete passivity and show our main theorem, stating that completely passive states under symmetry constraints are only GGEs.
Since the proof of this theorem is highly involved, we leave the full description of the proof to Supplemental Material.
Instead, we illustrate the physical implications of the theorem by some examples.
In Sec.~\ref{sec:work_storage_main}, we consider a setup including a quantum work storage.
In Sec.~\ref{sec:discussion}, we discuss the relation between the present study and other relevant previous studies. 
In Appendix~\ref{sec:passivity}, we show the condition for symmetry-protected passivity.
In Appendix~\ref{sec:proof}, we describe the full proof of the main theorem in Sec.~\ref{sec:setup} only for a simplest nontrivial example of a dimer model.
In Appendix~\ref{sec:discrete_symmetry}, we deal with the cases of some finite-group symmetries and time-reversal symmetry, where symmetry-protected completely passive states are only conventional Gibbs states.


\section{Setup and the main theorem} \label{sec:setup}

In this section, we discuss the setup and the main result of this paper.
In Sec.~\ref{subsec:symmetry}, we describe our definition of passivity and complete passivity under symmetry constraints. 
In Sec.~\ref{subsec:complete_passivity}, we present our main theorem, stating that any completely passive state under a symmetry constraint is a GGE.
In Sec.~\ref{subsec:examples}, we illustrate the setup and the theorem by some examples.

\subsection{Complete passivity under symmetry constraints}
\label{subsec:symmetry}

As a preliminary, we first formalize ordinary passivity without symmetry constraints~\cite{Pusz1978, Lenard1978}.
Let $H$ be the initial and final Hamiltonians of the system, which should be the same because we consider cyclic processes.
The time evolution of the system is represented by a unitary operator $U$.
Note that $U$ is an arbitrary unitary operator and is not necessarily given by $\exp(-\mathrm{i}tH)$, because the Hamiltonian can be time-dependent during the operation.
Then, the average work extracted from the system is defined as
\begin{align}
	W(\rho, U):=\mathrm{tr}(\rho H) -\mathrm{tr}\left(U\rho U^\dag H\right).
\label{work_def}
\end{align} 
Now, a state $\rho$ is called passive if $W(\rho, U)\leq0$ holds for all $U$. 
It is proved \cite{Pusz1978, Lenard1978} that a state is passive, if and only if the state $\rho$ is diagonal in the energy eigenbasis as $\rho=\sum_j p_j\ket{E_j}\bra{E_j}$, and the probabilities $\{p_j\}$ and the energy eigenvalues $\{E_j\}$ satisfy $p_1\geq p_2\geq\cdots$ and $E_1\leq E_2\leq\cdots$.

Even if a state $\rho$ is passive, there remains possibility to extract positive work from multiple copies of $\rho$. 
In such a case, $\rho$ cannot be regarded as truly thermal equilibrium, because one must not extract positive work from any number of copies of an equilibrium state.
We therefore define $\rho$ as completely passive, if $\rho^{\otimes N}$ is passive for all $N\in\mathbb{N}$.
It has been proved in Refs.~\cite{Pusz1978, Lenard1978} that $\rho$ is completely passive if and only if it is the Gibbs ensemble $\rho = e^{-\beta H} / Z$ for some $\beta \geq 0$.

In the foregoing conventional definition of passivity, any unitary operators are allowed to be implemented for work extraction.
In order to describe symmetry constraints, we will restrict the class of possible unitary operations in the following manner.
Consider a group $G$ that describes a symmetry, and fix a unitary representation of $G$. 
Let $U_g$ be the unitary labelled by $g \in G$, which should satisfy $U_g U_{g'}=U_{gg'}$ for any $g, g'\in G$. 
We do not assume that the unitary representation is irreducible, but technically, we assume that the representation is faithful.
(If a unitary representation is not faithful, the structure of $G$ is not fully represented by the representation. 
Physically, therefore, we can suppose that a unitary representation is faithful without loss of generality.)

When $G$ is a Lie group (a smooth continuous group), we can introduce the generators of the symmetry operators, which is the representation of the basis of the Lie algebra.
Physically, those generators are conserved charges $\{Q_i\}_{i=1}^n$, which are Hermitian operators linearly independent of each other.
If $G$ is connected and compact, they can generate all the symmetry operators as $U_g=\exp(\mathrm{i}\sum_{i=1}^n \alpha_i Q_i)$ with some $\alpha_1, \cdots, \alpha_n \in\mathbb{R}$ determined by $g$.
In the following, we assume that $G$ is a connected compact Lie group unless stated otherwise.

Now we say that a unitary $U$ respects the symmetry, if it commutes with all symmetry operators, that is, $[U, U_g]=0$ holds for all $g\in G$. Or equivalently, $U$ commutes with all conserved charges, that is, $[U, Q_i]=0$ holds for all $i$.
Moreover, we suppose that the Hamiltonian $H$ also respects the symmetry: $[H, U_g]=0$ for all $g\in G$, or equivalently $[H, Q_i]=0$ for all $i$.
We then define $\rho$ as \textit{symmetry-protected passive} if $W(\rho, U)\leq0$ for all symmetry-respecting unitaries $U$, where  $W(\rho, U)$ is defined by Eq.~\eqref{work_def}.

\begin{figure}
\begin{center}
\includegraphics[width=\columnwidth]{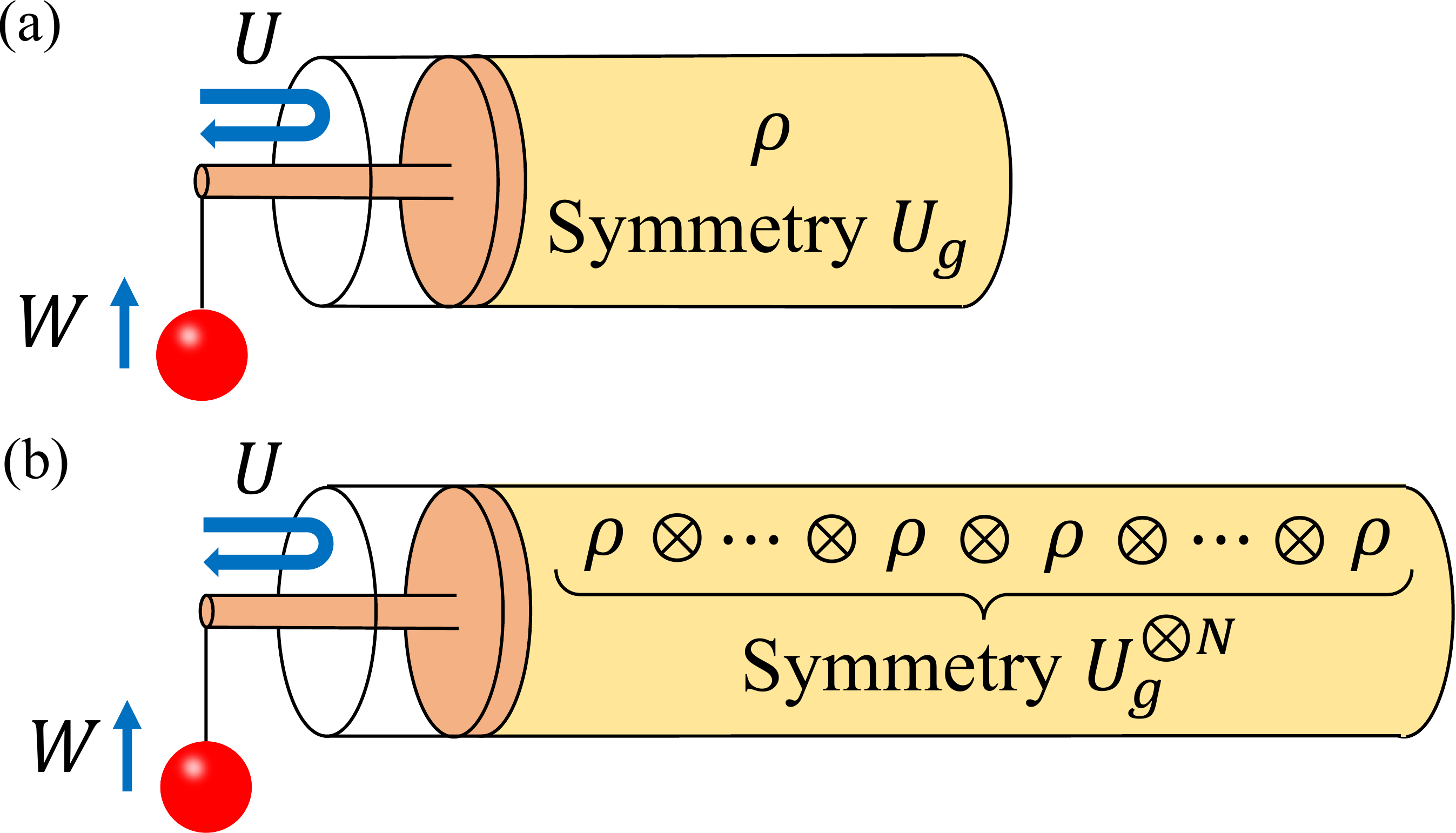}
\caption{Schematics of passivity and complete passivity under symmetry constraints.
The work $W$ is extracted by the unitary operation $U$.
(a) A state $\rho$ is called symmetry-protected passive, if one cannot extract positive work from a single state by any symmetry-respecting unitary $U$, which commutes with all $U_g$.
(b) A state $\rho$ is called symmetry-protected completely passive, if one cannot extract positive work from any number of copies of it by any symmetry-respecting unitary $U$, which commutes with all $U_g^{\otimes N}$ for all $N\in\mathbb{N}$.
Here, the symmetry is supposed to be global, i.e., the symmetry operations collectively act on all the copies.}
\label{fig2}
\end{center}
\end{figure}

It is more nontrivial to properly define complete passivity under symmetry constraints, because we need to specify the physically feasible class of symmetry operations on the multiple copies $\rho^{\otimes N}$. 
For this purpose, we here adopt the \textit{global} symmetry, which collectively acts on all the copies.
That is, we consider the unitary representation of the form $U_g^{\otimes N}$, by which all the copies are independently operated with the same symmetry operation $U_g$ (see also Fig.~\ref{fig2}).
Then, we say that an operator $U$ acting on $N$ copies respects the global symmetry, if it satisfies $[U , U_g^{\otimes N}] = 0$ for all $g\in G$.
We can also introduce the total charges 
\begin{align}
    Q_i^{(N)}:=\sum_{k=1}^N I^{\otimes k-1}\otimes Q_i \otimes I^{\otimes N-k} \label{eq:extensive_charge}
\end{align}
with $I$ being the identity operator.  
By using this notation, we can say that $U$ respects the global symmetry if the total charges are conserved: for all $i$,
\begin{equation}
    \left[U, Q_i^{(N)}\right] = 0.
    \label{total_charge_conservation}
\end{equation}

Meanwhile, the work extracted from the $N$ copies is defined by
\begin{equation}
    W^{(N)}(\rho^{\otimes N}, U)=\mathrm{tr}\left(\rho^{\otimes N} H^{(N)}\right)-\mathrm{tr}\left(U\rho^{\otimes N} U^\dag H^{(N)}\right),
    \label{multi_work}
\end{equation}
where 
\begin{align}
    H^{(N)} := \sum_{k=1}^N I^{\otimes k-1}\otimes H \otimes I^{\otimes N-k} \label{eq:extensive_Hamiltonian}
\end{align}
is the total Hamiltonian without interaction (as is the case for ordinary complete passivity).
We also suppose that the Hamiltonian $H$ respects the symmetry for individual copies, i.e., $[H , U_g ] = 0$ holds for all $g \in G$.
Finally, $\rho$ is called \textit{symmetry-protected completely passive} if $W^{(N)}(\rho^{\otimes N}, U)\leq 0$ holds for all $N \in \mathbb{N}$ and for all symmetry-respecting unitaries $U$ satisfying Eq.~\eqref{total_charge_conservation}.

\subsection{Main theorem} \label{subsec:complete_passivity}

We now state our main theorem of this paper, which gives the complete characterization of symmetry-protected complete passivity.
Note that the characterization of symmetry-protected (not complete) passivity is given in Appendix~\ref{sec:passivity}.

We first exclude the trivial situation where the Hamiltonian $H$ is of the form of $\alpha_0 I+ \sum_{i=1}^n \alpha_i Q_i$ with some $\alpha_0, \alpha_1, \cdots, \alpha_n \in\mathbb{R}$.
This is because in such a case, the energy is conserved under any symmetry-respecting unitary and thus all states are trivially completely passive.
We also suppose that  $\rho$ is positive-definite.
Then, the main theorem is stated as follows.

\begin{theorem} \label{thm:GFHCPmain}
	Let $G$ be a connected compact Lie group, $\{U_g\}_{g\in G}$ be its unitary (faithful) representation and $\{ Q_i \}$ be the corresponding conserved charges.
	A state $\rho$ is symmetry-protected completely passive with respect to a symmetry-respecting Hamiltonian $H$, if and only if $\rho$ is given by the GGE \eqref{GGE_main} with some $\beta \geq 0$ and $\mu_i \in \mathbb R$.
\end{theorem} 

The mathematically rigorous proof of this theorem is presented in Supplemental Material (Theorem~S3).
In particular, the proof of the \textit{only if} part is quite complicated and requires advanced tools from mathematical theory of Lie groups.
However, we will describe the proof for a special example in Appendix~\ref{sec:proof}.

At this stage, we only mention the proof of the \textit{if} part, which is much easier than the \textit{only if} part.
That is, we here show that the GGE \eqref{GGE_main} is symmetry-protected completely passive.
To see this, we remark that for given $\beta, \mu_i$, the GGE \eqref{GGE_main} can be seen as the Gibbs ensemble of the ``Hamiltonian'' $H' := H - \beta^{-1}\sum_i \mu_i Q_i$ and is completely passive with respect to $H'$.
If unitary $U$ respects the symmetry, it does not change the expectation value of the second term of $H'$, and thus the extracted work defined by $H'$ and $H$ are the same.  
This implies that the GGE \eqref{GGE_main} is symmetry-protected completely passive with respect to $H$.
Note that here we did not use the assumption that $H$ also respects the symmetry.
The above argument is essentially the same as a part of Ref.~\cite{Halpern2016}, while in our setup, this kind of argument cannot be applied to the converse part (i.e., the main part of this paper).

Meanwhile, we can determine the parameters $\beta$ and $\{ \mu_i \}$ in the GGE \eqref{GGE_main} in the following manner.
First, we consider the Hilbert-Schmidt inner product and orthonormalize the conserved charges $I, H, Q_1, \cdots Q_n$ by the Gram-Schmidt orthonormalization into $I/\sqrt{d}, Q'_0, Q'_1, \cdots Q'_n$ that satisfy $\mathrm{tr}(Q'_i)=0$ and $\mathrm{tr}(Q'_i Q'_j)=\delta_{ij}$ for $i, j=0, 1, \cdots, n$, where $d$ is the dimension of the Hilbert space of the system.
Then, the GGE is written as $\rho_\mathrm{GGE}=\exp(-\sum_{i=0}^n \mu'_i Q'_i) / Z'_\mathrm{GGE}$ with $\mu'_0, \mu'_1, \cdots\mu'_n\in\mathbb{R}$ and the normalization constant $Z'_\mathrm{GGE}$.
Here, $\mu'_0, \mu'_1, \cdots\mu'_n$ can be regarded as the coefficients of the orthonormalized basis $\{Q'_i\}$ in $-\log(\rho_\mathrm{GGE})$, and are given by $\mu'_i=-\mathrm{tr}(\log(\rho_\mathrm{GGE}) Q'_i)$.
By expressing $Q'_i$ by a linear combination of $I$, $H$, and $Q_i$, we obtain  $\beta$ and $\{\mu_i\}$.

\subsection{Examples}
\label{subsec:examples}

We show some illustrative examples.
In the case of $G=U(1)$, there is a single charge $Q$.
It often describes the particle number $N$, where Eq.~\eqref{total_charge_conservation} means the conservation of total particle number (see Fig.~\ref{fig3} (a)).
We note that the particle number of an individual $\rho$ is not necessarily conserved, but that of the multiple copies is globally conserved.
In this case, Theorem~\ref{thm:GFHCPmain} states that a state is symmetry-protected completely passive, if and only if it is the grand canonical ensemble 
\begin{align}
    \rho_\mathrm{GGE}=\frac{1}{Z_\mathrm{GGE}}\exp(-\beta H-\mu N)
\end{align}
with inverse temperature $\beta>0$ and chemical potential $\mu\in\mathbb{R}$.

In the case of $G=SU(2)$, there are three charges $Q^x, Q^y, Q^z$.
They often describe the spin operators in the $x, y, z$-directions, where Eq.~\eqref{total_charge_conservation} means the total spin (magnetization) conservation in all the directions.
We write $\bm{Q}:=(Q^x, Q^y, Q^z)$.

\begin{figure}
\begin{center}
\includegraphics[width=\columnwidth]{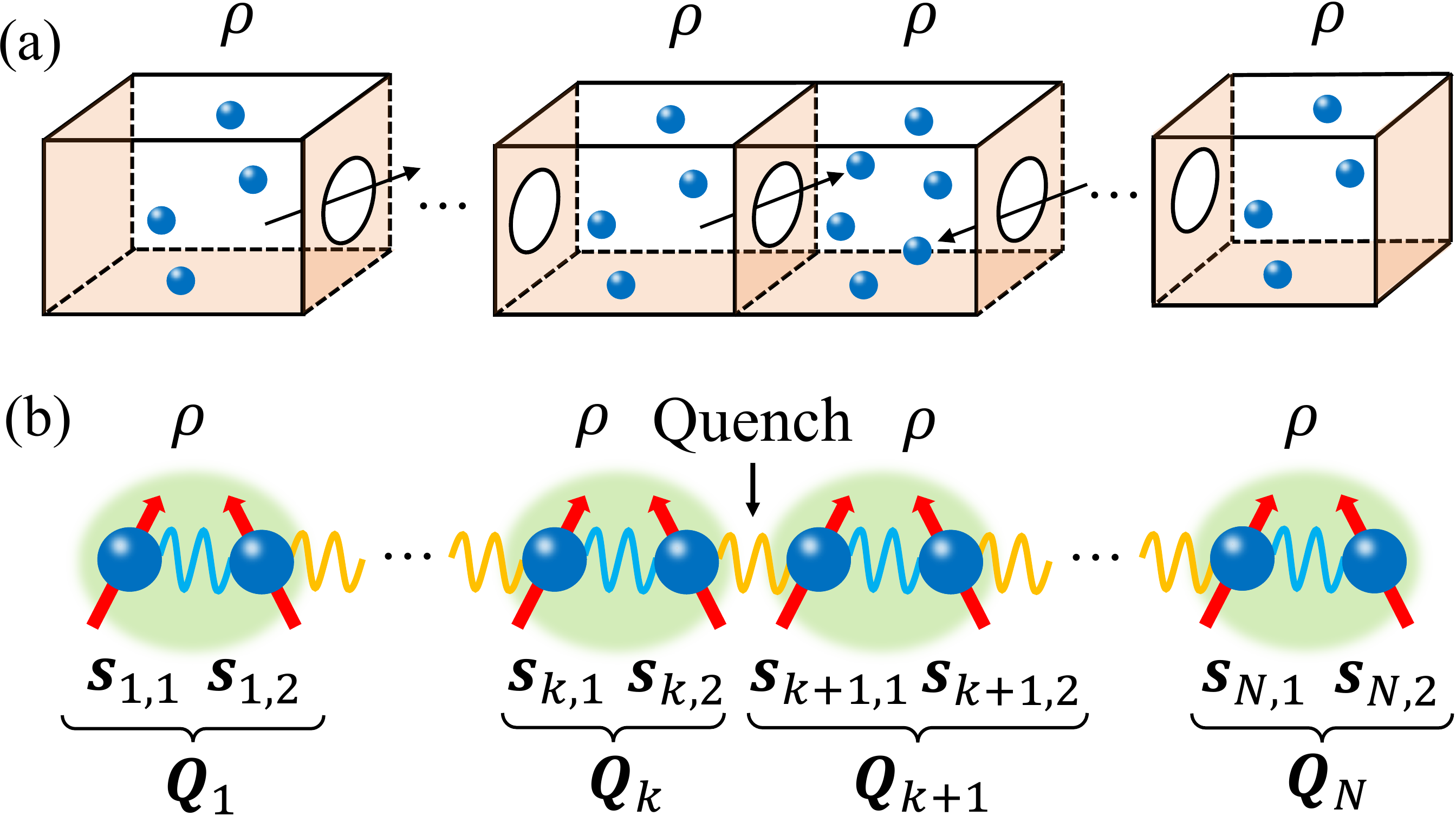}
\caption{Toy examples of multiple copies of the system under global symmetry constraints.
(a) $U(1)$ symmetry and the particle number conservation.
The total particle number should be conserved, while individual copies can exchange particles with each other.
(b) $SU(2)$ symmetry and the spin conservation.
Each copy of the system is represented by a ``dimer'' consisting of two spin-$1/2$ systems.
We consider the situation where the total spin is conserved in all the spatial directions, while each spin is not necessarily conserved.
Such symmetry-respecting interaction between the dimers can be induced, for example, by quenching the second term on the Hamiltonian~\eqref{dimer_Hamiltonian}.}
\label{fig3}
\end{center}
\end{figure}

In the following, let us elucidate the $SU(2)$ case by considering a ``dimer'' model (see Fig.~\ref{fig3}(b)). 
Suppose that the system consists of two spin-$1/2$ systems with the spin operators $\bm s_1, \bm s_2$.
The unitary representation of $SU(2)$ on this system is generated by $\bm{Q} := \bm s_1 \otimes I + I \otimes \bm s_2$, which consists of two irreducible sectors with total spin $0$ and $1$. 

We consider the Hamiltonian of the system with the isotropic Heisenberg-type interaction between the two spins, $H = \bm s_1 \cdot \bm s_2$. 
It is straightforward to check that $H$ commutes with all the components of $\bm{Q}$, implying the total spin conservation.
We consider $N$ copies of this system (i.e., $2N$ spin-$1/2$ systems).
Let $\bm s_{1, k}, \bm s_{2, k}$ be the spin operators of the $k$th copy.
The total Hamiltonian is given by $H^{(N)} = \sum_{k=1}^N \bm s_{1, k} \cdot \bm s_{2, k}$, which describes a trivial sequence of the dimers without interaction.

A simple example of symmetry-respecting unitaries is given by $U = \exp (-\mathrm{i} H')$ with $H'$ being the one-dimensional XXX Hamiltonian.
That is, $H'$ is obtained by quenching the interaction between the dimers:
\begin{align}
	H' = H^{(N)} + \sum_{k=1}^{N-1} \bm s_{2, k} \cdot \bm s_{1, k+1}. 
	\label{dimer_Hamiltonian}
\end{align}
It is again straightforward to check that $H'$ conserves the total spin of $N$ copies.

In this dimer example, Theorem~\ref{thm:GFHCPmain} states that a state is symmetry-protected completely passive, if and only if it is the GGE
\begin{align}
    \rho_\mathrm{GGE}=\frac{1}{Z_\mathrm{GGE}}\exp\left(-\beta H-\sum_{\alpha=x, y, z} \mu_\alpha Q^\alpha\right)
\end{align}
with  inverse temperature $\beta\geq0$ and generalized chemical potentials $\mu_x, \mu_y, \mu_z\in\mathbb{R}$.
We will give a full proof of Theorem~\ref{thm:GFHCPmain} for the case of this dimer model in Appendix~\ref{sec:proof}.

We note that the conserved charges already satisfy the orthogonal relation, and thus we only need to normalize the charges. 
Then, $\beta$ and $\mu_\alpha$'s are simply given by $\beta=-\mathrm{tr}(\log(\rho_\mathrm{GGE}) H)/\mathrm{tr}(H^2)=-4\mathrm{tr}(\log(\rho_\mathrm{GGE}) H)/3$, $\mu_\alpha=-\mathrm{tr}(\log(\rho_\mathrm{GGE}) Q^\alpha)/\mathrm{tr}((Q^{\alpha})^ 2)=-\mathrm{tr}(\log(\rho_\mathrm{GGE}) Q^\alpha)/2$.

\section{Role of work storage} \label{sec:work_storage_main}

So far, we have considered the setup where the  work is defined as the difference of the average energies of the system before and after a unitary operation.
On the other hand, there is another setup that reflects the first law of thermodynamics more explicitly~\cite{Horodecki2013, Skrzypczyk2013, Aberg2014, Brandao2013, Aberg2013, Skrzypczyk2014, Lieb1999,Malabarba2015}, where a quantum work storage is introduced in addition to the system of interest, and unitary operations on the total system should conserve the total energy. 
In this section, we consider this setup with the fully quantum treatment of the work storage.

In Sec.~\ref{subsec:qw_setup}, we describe our setup of the work storage.
In Sec.~\ref{subsec:qw_no_symmetry}, in the absence of symmetry constraints, we show the relationship between the maximal works with and without the work storage, which is of separate interest. 
In Sec.~\ref{subsec:qw_complete_passivity}, we prove that even in the presence of the work storage, a state is completely passive under a symmetry constraint if and only if it is a GGE, independently of the initial state of the work storage.

\subsection{Setup}
\label{subsec:qw_setup}


We introduce a work storage attached to the system of interest in line with Refs.~\cite{Lieb1999,Aberg2014, Skrzypczyk2014,Skrzypczyk2013,Malabarba2015}.  
The work storage is a continuous system described by position and momentum, and its Hamiltonian is given by the position operator $x$.

We impose the following two conditions on  implementable unitary operator $V$ that acts on the composite system including the work storage: I) $V$ conserves the total energy, i.e., 
\begin{equation}
    [V, H\otimes I+I\otimes x]=0;
\label{strict_conservation}
\end{equation}
II) $V$ is invariant under energy translation of the work storage, i.e., 
\begin{align}
    [V, I\otimes p]=0,
\end{align}
where $p$ is the momentum operator (the generator of energy translation) of the work storage.
Note that the canonical commutation relation is given by $[x,p] = {\rm i}$.

We also suppose that there is no correlation between the system and the work storage  in the initial state.
Then, we define the extracted work as the difference of the average energies of the work storage before and after an operation $V$:
\begin{align}
    &W^\mathrm{WS}(\rho, \rho_\mathrm{W}, V) \nonumber\\
    :=&\mathrm{tr}(V(\rho\otimes\rho_\mathrm{W})V^\dag(I\otimes x))-\mathrm{tr}((\rho\otimes\rho_\mathrm{W})(I\otimes x)),
\end{align}
where $\rho$ and $\rho_\mathrm{W}$ are the initial states of the system and the work storage, respectively.
From Condition II), the extracted work is invariant under energy translation of the initial state of the work storage.


We note that the average work extraction adopted here is different from a protocol investigated in Ref.~\cite{Horodecki2013}, which can be referred to as almost deterministic work extraction.
We also note that we adopted strict energy conservation (\ref{strict_conservation}) as in Refs.~\cite{Skrzypczyk2013, Aberg2014}, but not the average energy conservation in the sense of Ref.~\cite{Skrzypczyk2014}.


\subsection{Without symmetry constraints} \label{subsec:qw_no_symmetry}


We consider work extraction in the case without symmetry constraints.
We clarify the relationship between the maximal extracted work with and without work storage.
Moreover, we show that in the presence of the work storage, completely passive states are only Gibbs states.

First, we consider the correspondence between the implementable unitary operations in the setups with and without the work storage.
For a unitary operator $U$ on the system of interest, we can construct a unitary operator that acts on the composite system including the work storage and satisfies Conditions I) and II) by the following Kitaev construction~\cite{Kitaev2004}:
\begin{align}
	\mathcal{C}(U)=\int_{-\infty}^\infty dq\ e^{\mathrm{i}qH}U e^{-\mathrm{i}qH}\otimes\ket{q}\bra{q},
	\label{eq:unitary_correspondence}
\end{align}
where $\ket{q}$ represents the momentum eigenstate of the work storage with eigenvalue $q$.
This unitary operator is equivalent to the one that appears in Refs.~\cite{Skrzypczyk2013, Malabarba2015}.
It is shown in Ref.~\cite{Skrzypczyk2013} that all the unitary operators that satisfies Conditions I) and II) are represented by $\mathcal{C}(U)$ with some unitary operator $U$.
In Ref.~\cite{Kitaev2004}, the operator $\mathcal{C}(U)$ is introduced as an operator that simulates $U$ by giving the same action on $\rho$ if $\rho$ is symmetry-respecting.
However, we note that if the state $\rho$ is not symmetry-respecting, the action of $\mathcal{C}(U)$ on the system of interest is not necessarily the same as that of $U$, and moreover, the extracted work from $\rho\otimes\rho_\mathrm{W}$ by the action of $\mathcal{C}(U)$ is not necessarily the same as that from $\rho$ by the action of $U$.

\begin{figure}
\begin{center}
\includegraphics[width=\columnwidth]{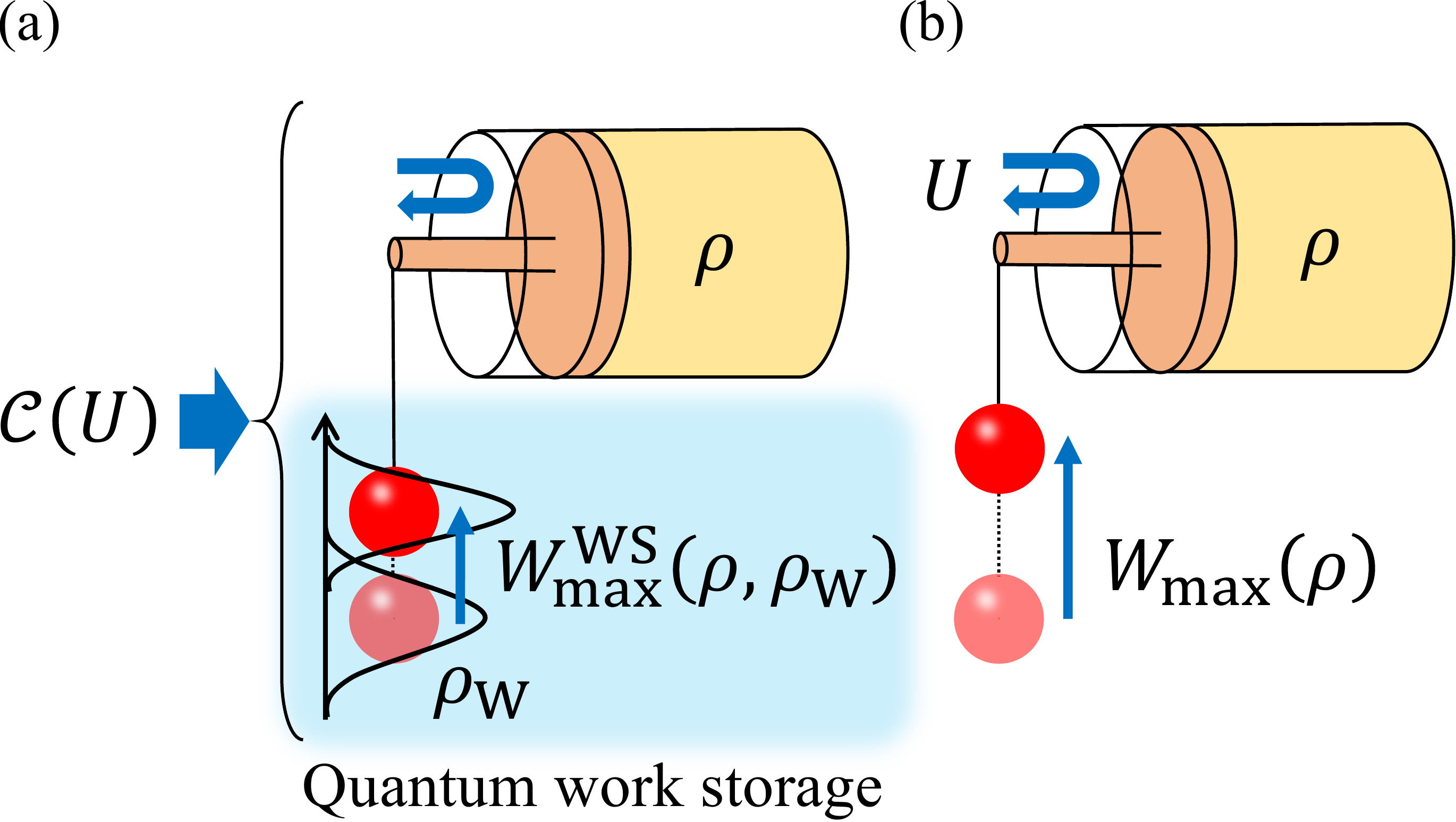}
\caption{Work extraction in the setups (a) with and (b) without the work storage.
$\mathcal{C}(U)$ is a unitary operator acting on the composite system of the system and the work storage, which is defined by a unitary operator $U$ acting only on the system through the Kitaev construction~\eqref{eq:unitary_correspondence}.
$\mathcal{C}(U)$ is constructed such that it satisfies energy conservation of the composite system and energy translation invariance of the work storage.
Moreover, $\mathcal{C}(U)$ gives a one-to-one correspondence between the possible operations in the setups with and without the work storage.
As shown in Eqs.~\eqref{eq:ergotropy} and \eqref{eq:ergotropy_inequalty}, the maximal extracted work with the work storage is no greater than that without the work storage.
If $U$ satisfies $[U^\dag HU, H]=0$, the extracted work from $\rho$ without the work storage is the same as that from $\rho$ with the work storage by unitary operation $\mathcal{C}(U)$ independently of the work storage $\rho_\mathrm{W}$, which is stated in Proposition~\ref{prop:extracted_work_correspondence}.
}
\label{fig6}
\end{center}
\end{figure}

By using the Kitaev construction (\ref{eq:unitary_correspondence}), we derive the relation between the maximal extracted works  with and without the work storage.
The maximal extracted work from a state $\rho$ in the setups with and without the work storage are respectively defined as
\begin{align}
    &W_\mathrm{max}^\mathrm{WS}(\rho, \rho_\mathrm{W}):=\max_{V} W^\mathrm{WS}(\rho, \rho_\mathrm{W}, V),\\
    &W_\mathrm{max}(\rho):=\max_{U} W(\rho, U),
\end{align}
where $V$ ranges over all the unitaries on the composite system that satisfy Conditions I) and II), and $U$ ranges over all the unitaries on the system of interest. 
Since any $V$ can be represented as $\mathcal{C}(U)$ for some $U$, $W_\mathrm{max}^\mathrm{WS}(\rho, \rho_\mathrm{W})$ is also written as 
\begin{align}
    W_\mathrm{max}^\mathrm{WS}(\rho, \rho_\mathrm{W})=\max_{U} W^\mathrm{WS}(\rho, \rho_\mathrm{W}, \mathcal{C}(U)).
    \label{eq:WS_ergotropy2}
\end{align}
From Eq.~\eqref{eq:unitary_correspondence}, the extracted work is given as
\begin{align} 
    &W^\mathrm{WS}(\rho, \rho_\mathrm{W}, \mathcal{C}(U)) \nonumber\\
    =&\int_{-\infty}^\infty dq\ \braket{q|\rho_\mathrm{W}|q}W\left(e^{-\mathrm{i}qH}\rho e^{\mathrm{i}qH}, U\right) \nonumber\\
    =&W\left(\mathcal{D}_{\rho_\mathrm{W}}(\rho), U\right),
    \label{eq:WS_ergotropy_relation}
\end{align}
where $\mathcal{D}_{\rho_\mathrm{W}}$ is defined as
\begin{align}
    \mathcal{D}_{\rho_\mathrm{W}}(\rho):=\int_{-\infty}^\infty dq\ \braket{q|\rho_\mathrm{W}|q} e^{-\mathrm{i}qH}\rho e^{\mathrm{i}qH}.
\end{align}
From Eqs.~\eqref{eq:WS_ergotropy2} and \eqref{eq:WS_ergotropy_relation}, we obtain
\begin{align}
    W_\mathrm{max}^\mathrm{WS}(\rho, \rho_\mathrm{W})
    =W_\mathrm{max}\left(\mathcal{D}_{\rho_\mathrm{W}}(\rho)\right).
    \label{eq:ergotropy}
\end{align}

We next prove that this value is no greater than the maximal extracted work from $\rho$ in the setup without the work storage:
\begin{align}
    W_\mathrm{max}\left(\mathcal{D}_{\rho_\mathrm{W}}(\rho)\right)\leq W_\mathrm{max}(\rho).
    \label{eq:ergotropy_inequalty}
\end{align}
This inequality is proved as follows:
\begin{align}
    &W_\mathrm{max}\left(\mathcal{D}_{\rho_\mathrm{W}}(\rho)\right) \nonumber\\
    =&W_\mathrm{max}\left(\int_{-\infty}^\infty dq\ \braket{q|\rho_\mathrm{W}|q} e^{-\mathrm{i}qH}\rho e^{\mathrm{i}qH}\right) \nonumber\\
    \leq&\int_{-\infty}^\infty dq\ \braket{q|\rho_\mathrm{W}|q} W_\mathrm{max}(e^{-\mathrm{i}qH}\rho e^{\mathrm{i}qH}) \nonumber\\
    =&\int_{-\infty}^\infty dq\ \braket{q|\rho_\mathrm{W}|q} W_\mathrm{max}(\rho) \nonumber\\
    =&W_\mathrm{max}(\rho),
\end{align}
where we used the concavity of $W_\mathrm{max} (\rho)$ to obtain the third line.
Here, the concavity of $W_\mathrm{max}(\rho)$ is shown as follows.
Take arbitrary $s\in[0, 1]$ and arbitrary states $\rho_1$ and $\rho_2$.
We take one of the unitary operators $U_0$ that extract the maximal extracted work from $\rho=s\rho_1+(1-s)\rho_2$.
Then, 
\begin{align}
    &W_\mathrm{max}(\rho)\nonumber\\
    =&\mathrm{tr}((\rho-U_0\rho U_0^\dag)H) \nonumber\\
    =&s\mathrm{tr}((\rho_1-U_0\rho_1 U_0^\dag)H)+(1-s)\mathrm{tr}((\rho_2-U_0\rho_2 U_0^\dag)H) \nonumber\\
    \leq& sW_\mathrm{max}(\rho_1)+(1-s)W_\mathrm{max}(\rho_2),
\end{align}
which implies the concavity of $W_\mathrm{max}(\rho)$.

We consider two extreme examples of $\mathcal{D}_{\rho_\mathrm{W}}$.
When $\rho_\mathrm{W}$ is a momentum eigenstate, $\mathcal{D}_{\rho_\mathrm{W}}$ is the identical mapping.
In this case, Eq.~\eqref{eq:ergotropy} states that $W_\mathrm{max}^\mathrm{WS}(\rho, \rho_\mathrm{W})=W_\mathrm{max}(\rho)$, i.e., the extracted works are the same in the setups with and without work storage, which is consistent with the result in Ref.~\cite{Aberg2014}. 
On the other hand, when $\rho_\mathrm{W}$ is a position eigenstate, $\mathcal{D}_{\rho_\mathrm{W}}$ is the dephasing mapping $\Delta$ defined as $\Delta(\rho):=\sum_E \Pi_E \rho\Pi_E$, where $\Pi_E$ is the projection operator onto the energy eigenspace of $E$.
In this case, Eq.~\eqref{eq:ergotropy} states that $W_\mathrm{max}^\mathrm{WS}(\rho, \rho_\mathrm{W})=W_\mathrm{max}(\Delta(\rho))$.
This means that we cannot extract work from coherence, which is reminiscent of a phenomenon called work locking~\cite{Horodecki2013}.
It is shown, however, in Refs.~\cite{Horodecki2013, Skrzypczyk2013} that if we have infinitely many copies of the state, we can again extract work from coherence.

We now prove that completely passive states are only Gibbs states in the presence of the work storage, which is of separate interest.
For that purpose, we prepare the following two propositions.

The first proposition states that in the setup without work storage, we can extract positive work from multiple copies of any state other than the Gibbs state, even if we further impose constraints $[U^\dag HU, H]=0$ on the possible operations $U$.
This is a stronger statement than the conventional characterization of complete passivity~\cite{Pusz1978, Lenard1978}.

\begin{proposition}
    Let $\rho$ be a state such that $W(\rho^{\otimes N}, U)\leq0$ holds
    for any $N\in\mathbb{N}$ and any unitary operator $U$ acting on $\rho^{\otimes N}$ satisfying $[U^\dag H^{(N)}U, H^{(N)}]=0$, where $H^{(N)}$ is defined by Eq.~\eqref{eq:extensive_Hamiltonian}.
    Then, $\rho$ is the Gibbs ensemble at positive temperature.
    \label{prop:U_complete}
\end{proposition}

The proof is the simplest case of that of Theorem~\ref{thm:GFHCPmain} without symmetry constraints (see also Appendix~\ref{sec:proof}).

\begin{proof}
We define a sequence of unitary operators $\{ U_m \}_{m \in \mathbb N}$ that satisfies $[U_m^\dag H^{(2m+1)}U_m, H^{(2m+1)}]=0$, and consider the extracted work from $\rho^{\otimes 2m+1}$ by  $U_m$.  
Since the Hamiltonian is not trivial, there exist energy eigenstates $\ket{E_{k_0}}$, $\ket{E_{k_1}}$ with different eigenvalues $E_{k_0}<E_{k_1}$.
For any $l$, we define $R_{ij}:=\frac{1}{2}(I-(-1)^i T)[\ket{E_l}\bra{E_l}\otimes(I-\ket{E_l}\bra{E_l})](I-(-1)^j T)$ with the swapping operator $T$ between two copies of the system.
We consider unitary operator $U_m:=I-\sum_{i,j\in\{0, 1\}} (-1)^{i-j} R_{ij}^{\otimes m}\otimes \ket{E_{k_i}}\bra{E_{k_j}}$, which satisfies $[U_m^\dag H^{(2m+1)}U_m, H^{(2m+1)}]=0$ and 
\begin{align}
    &H^{(2m+1)}-U_m^\dag H^{(2m+1)}U_m \nonumber\\
    =&(E_{k_1}-E_{k_0})(R_{11}^{\otimes m}\otimes\ket{E_{k_1}}\bra{E_{k_1}}-R_{00}^{\otimes m}\otimes\ket{E_{k_0}}\bra{E_{k_0}}).
\end{align}
Therefore, the extracted work is given by
\begin{align}
    &W(\rho^{\otimes 2m+1}, U_m) \nonumber\\
    =&\mathrm{tr}(\rho^{\otimes 2m+1}(H^{(2m+1)}-U_m^\dag H^{(2m+1)}U_m)) \nonumber\\
    =&(E_{k_1}-E_{k_0})\left[\mathrm{tr}(\rho^{\otimes 2}R_{11})^m\braket{E_{k_1}|\rho|E_{k_1}}\right. \nonumber\\
    &\ \ \ \ \ \ \ \ \ \ \ \ \ \ \  -\mathrm{tr}\left.(\rho^{\otimes 2}R_{00})^m\braket{E_{k_0}|\rho|E_{k_0}}\right] \nonumber\\
    =&a\left[1-b\left(\frac{\mathrm{tr}(\rho^{\otimes 2}R_{00})}{\mathrm{tr}(\rho^{\otimes 2}R_{11})}\right)^m\right] \nonumber\\
    =&a[1-b(1+c\|[\rho, \ket{E_l}\bra{E_l}]\|_\mathrm{HS}^2)^{-m}],
    \label{eq:ext_work}
\end{align}
where $a:=(E_{k_1}-E_{k_0})(\mathrm{tr}(\rho^{\otimes 2}R_{11}))^m\braket{E_{k_1}|\rho|E_{k_1}} > 0$, $b:=\braket{E_{k_0}|\rho|E_{k_0}}/\braket{E_{k_1}|\rho|E_{k_1}} > 0$, $c:=[\mathrm{tr}(\rho^{\otimes 2}R_{00})]^{-1} > 0$, $\|\cdot\|_\mathrm{HS}$ is the Hilbert-Schmidt norm, and we used 
\begin{align}
    \mathrm{tr}(\rho^{\otimes 2}R_{11})-\mathrm{tr}(\rho^{\otimes 2}R_{00})=\|[\rho, \ket{E_l}\bra{E_l}]\|_\mathrm{HS}^2.
\end{align}

From Eq.~\eqref{eq:ext_work}, if $W(\rho^{\otimes 2m+1}, U_m)\leq0$ holds for all $m$, then $\|[\rho, \ket{E_l}\bra{E_l}]\|_\mathrm{HS}=0$, i.e., $[\rho, \ket{E_l}\bra{E_l}]=0$ must be satisfied.
Therefore, $\rho$ can be written as $\rho=\sum_{l} p_l\ket{E_l}\bra{E_l}$ with $p_l\in(0, 1)$.
It is proved in Ref.~\cite{Skrzypczyk2015} that we can extract positive work from multiple copies of such a state by a unitary operator $U$ that satisfies $[U^\dag H^{(N)}U, H^{(N)}]=0$ for some $N\in\mathbb{N}$, unless the state is the Gibbs ensemble at positive temperature. 
\end{proof}

The second proposition states that if a unitary operator $U$ satisfies $[U^\dag HU, H]=0$, we can extract the same amount of work for both the cases where $U$ is implemented without the work storage and  $\mathcal{C}(U)$ is implemented  with the work storage, independently of its initial state.

\begin{proposition} \label{prop:extracted_work_correspondence}
If $U$ satisfies $[U^\dag HU, H]=0$, the extracted work from a state $\rho$ of the system of interest by the action of $U$ without the work storage is the same as that from the state $\rho$ with the work storage by the action of $\mathcal{C}(U)$ acting on the composite system with any state $\rho_\mathrm{W}$ of the work storage. 
\end{proposition}

It is proved in Ref.~\cite{Malabarba2015} that if a work-extracting unitary operator $U$ only permutes energy eigenstates, the extracted work  without the work storage equals  the corresponding extracted work with the work storage.
Since we find that a unitary operator $U$ satisfies $[U^\dag HU, H]=0$ if and only if $U$ only permutes energy eigenstates, we can prove Proposition~\ref{prop:extracted_work_correspondence}.
However, we can also prove it without using the fact that $U$ only permutes energy eigenstates, as shown in the direct proof below.

\begin{proof}
    Since $U$ satisfies $[U^\dag HU, H]=0$, we obtain for any $q\in\mathbb{R}$, 
    \begin{align}
        W\left(e^{-\mathrm{i}qH}\rho e^{\mathrm{i}qH}, U\right)&=\mathrm{tr}\left(\rho e^{\mathrm{i}qH}(H-U^\dag HU)e^{-\mathrm{i}qH}\right)\nonumber\\
        &=\mathrm{tr}\left(\rho (H-U^\dag HU)\right)\nonumber\\
        &=W(\rho, U).
    \end{align}
    Therefore, for any $\rho_\mathrm{W}$, we get
    \begin{align}
        &W^\mathrm{WS}(\rho, \rho_\mathrm{W}, \mathcal{C}(U))\nonumber\\
        =&\int_{-\infty}^\infty dq\ \braket{q|\rho_\mathrm{W}|q}W\left(e^{-\mathrm{i}qH}\rho e^{\mathrm{i}qH}, U\right)\nonumber\\
        =&\int_{-\infty}^\infty dq\ \braket{q|\rho_\mathrm{W}|q}W(\rho, U)\nonumber\\
        =&W(\rho, U).
    \end{align}
\end{proof}

From Proposition~\ref{prop:U_complete} and Proposition~\ref{prop:extracted_work_correspondence}, if a state is completely passive in the setup with the work storage, the state is the Gibbs ensemble independently of the initial state of the work storage.
The converse is obvious  from Eqs.~\eqref{eq:ergotropy}~and~\eqref{eq:ergotropy_inequalty}.
Therefore, we finally obtain the following proposition.

\begin{proposition} 
\label{prop:WS_complete_passivity}
     For any initial state $\rho_\mathrm{W}$ of the work storage,
     a state $\rho$ of the system of interest is completely passive in the setup with the work storage, if and only if the state $\rho$ is the Gibbs ensemble at positive temperature.
\end{proposition}

We here remark on the relation between the above proposition and some known results.
In fact, essentially the same result has been proved in Ref.~\cite{Horodecki2013} by allowing for the introduction of any number of auxiliary heat baths, while in our setup we do not allow for it.
Ref.~\cite{Sparaciari2017} also addresses a similar problem without allowing for auxiliary heat baths but by assuming that the work storage is initially in the uniform superposition of energy eigenstates, while in our setup we consider an arbitrary initial state of the work storage.

\subsection{With symmetry constraints} \label{subsec:qw_complete_passivity}


We consider symmetry-protected completely passive states in the setup with the work storage and show that they are only GGEs.
Since we consider the work extraction that is purely defined by the energy, we introduce the work storage that only stores the energy but does not store other conserved charges associated with symmetry constraints, in contrast to Ref.~\cite{Halpern2016}.
For example, we can imagine the situation where the system consists of atoms with the conserved particle number, which is coupled to light in a cavity as an external system.
In such a situation, it is natural to impose $U(1)$ symmetry only on the system of interest, instead of the total system.

We consider a symmetry-respecting unitary $V$ in the setup with the work storage, which is supposed to satisfy not only Conditions I) and II) mentioned in Sec.~\ref{subsec:qw_setup}, but also the following: III) $V$ respects the symmetry of the system of interest, i.e., 
\begin{align}
    [V, U_g\otimes I]=0.
\end{align}
This reflects the fact that the work storage does not store the conserved charges.
Then, we define that a state is symmetry-protected completely passive, if we cannot extract positive work from any number of multiple copies of the state by any unitary $V$ that satisfies Conditions I), II) and III).

The condition for symmetry-protected complete passivity in the setup with the work storage is now stated as the following theorem.

\begin{theorem} 
\label{thm:WS_GFHCP}
    For any initial state of the work storage, a state of the system of interest is symmetry-protected completely passive in the setup with the work storage, if and only if the state is the GGE. 
\end{theorem}

We can prove Theorem~\ref{thm:WS_GFHCP} in the following manner.
From Eqs.~\eqref{eq:ergotropy}~and~\eqref{eq:ergotropy_inequalty}, it is obvious that we cannot extract positive work from multiple copies of the GGE  with the work storage under symmetry constraints.
Then, we only need to show that we can extract positive work from multiple copies of any other state than the GGE by some unitary that satisfies Conditions I), II) and III).
In the proof of Theorem~\ref{thm:GFHCPmain} (see Proposition~S8 of Supplemental Material), we construct a unitary operator $U$ that extracts positive work from any state other than the GGE and satisfies $[U^\dag HU, H]=0$,
which is the generalization of Proposition~\ref{prop:U_complete} to the setup under symmetry constraints.
From Proposition~\ref{prop:extracted_work_correspondence}, in the setup with the work storage, we can extract the same amount of work by implementing $\mathcal{C}(U)$ for any initial state of the work storage.
We can also check that if $U$ is symmetry-respecting, $\mathcal{C}(U)$ satisfies Condition III). 
Therefore, if a state is not symmetry-protected completely passive in the setup without the work storage, then the state is not symmetry-protected completely passive in the setup with the work storage independently of the state of the work storage.

\section{Discussion} \label{sec:discussion}

In this paper, we have provided the characterization of complete passivity for systems under symmetry constraints, which is referred to as symmetry-protected thermal equilibrium.
We proved that a state is symmetry-protected completely passive if and only if it is a GGE of the form \eqref{GGE_main}, which is the main result of this paper (Theorem~\ref{thm:GFHCPmain}).
Remarkably, our result applies to non-commutative symmetries, as illustrated by the dimer model with $SU(2)$ symmetry discussed in Sec.~\ref{subsec:examples}.
While we leave the full proof of Theorem~\ref{thm:GFHCPmain} to Supplemental Material, that for the special case of the dimer model is provided in Appendix~\ref{sec:proof}.
In Appendix~\ref{sec:discrete_symmetry}, we also show that under a certain class of finite group symmetry constraints, completely passive states are only Gibbs states.

Moreover, we proved that the same characterization of symmetry-protected complete passivity holds true, by explicitly including the work storage as a quantum system (Theorem~\ref{thm:WS_GFHCP}).
As a bi-product (Proposition~\ref{prop:WS_complete_passivity}), we proved that, in a stronger form than the known results in literature, completely passive states without symmetry constraints are only Gibbs ensembles in the presence of the work storage.
We note that the energy levels of the work storage introduced in our setup is unbounded from below, but we expect that we can extend our argument for the work storage bounded from below by following the idea of Ref.~\cite{Aberg2014}.

We discuss the relationship between the present work and other approaches to symmetries and GGEs.
Let us first clarify the fundamental difference between our study and a previous study by Yunger Halpern \textit{et al.}~\cite{Halpern2016}, where non-commutative GGE \eqref{GGE_main} also appears.
In our study, a symmetry constraint is imposed solely on the system of interest, while in Ref.~\cite{Halpern2016}, it is imposed on the entire system including external charge storages.
That is, the charges of the system are solely conserved in our setup, while the charges can be transferred to the storages in their setup.
It should also be emphasized that our definition of work is given purely by the energy (i.e., the Hamiltonian), while they adopted a generalized notion called chemical work.
Related to this point, a characteristic of our approach to symmetry-protected (complete) passivity lies in the fact that its definition (Sec.~\ref{subsec:symmetry}) itself does not involve the parameters $\mu_i$ of the GGE \eqref{GGE_main}.
To summarize, our setup is fundamentally different from theirs, and thus complements their result by providing a further support that Eq.~(\ref{GGE_main}) is a proper form of the GGE including the non-commutative cases.

In Refs.~\cite{Guryanova2016, Lostaglio2017}, other types of passivity in the presence of conserved charges are defined and investigated.
In these studies, however, (complete) passivity is defined with a focus on extracting charges themselves instead of the energy (see also Ref.~\cite{Sparaciari2017}), which is the opposite case to our setup of purely extracting energy.

We finally note that the resource theory of asymmetry adopts a broader class of free operations~\cite{Marvian2013, Marvian2014} than our setup.
The resource theory that adopts our smaller class of free operations would also be useful in quantum thermodynamics in the  presence of conserved charges, because our setup requires the conservation of the expectation values of the charges, but their setup does not.
For example, in their setup the operation that increases the particle number independently of the initial state is allowed, while in our setup such operation is not allowed.
From a general perspective of resource theories, we can say that our result has determined the class of free states of the resource theory of thermodynamics with conserved charges, and thus would serve as a foundation of a  new class of resource theories in the presence of symmetries.


\textbf{Acknowledgment.}
The authors are grateful to 
Nicole Yunger Halpern, Philippe Faist, Marcos Rigol, and Yasunobu Nakamura for valuable discussions.
The authors are also grateful to Kousuke Kumasaki and Taro Sawada
for useful comments on an early draft.
Y.M. is supported by World-leading Innovative Graduate
Study Program for Materials Research, Industry, and Technology (MERIT-WINGS) of the University of Tokyo.
T.S. is supported by JSPS KAKENHI Grant No. JP16H02211 and No. JP19H05796,
and Institute of AI and Beyond of the University of Tokyo.

\appendix

\section{Symmetry-protected passivity} \label{sec:passivity}

In this Appendix, we reveal the characterization of symmetry-protected (not complete) passivity.
Symmetry-protected passive states are defined as the states from which no positive amount of work can be extracted by any symmetry-respecting unitary operations (see also Fig.~\ref{fig2}(a)). 
We will prove that a state is symmetry-protected passive, if and only if every sector of the symmetrized density operator of the state is passive with respect to the corresponding Hamiltonian.



First, we specify the form of symmetry-respecting operators.
We follow the method in Ref.~\cite{Bartlett2007} and make use of the decomposition of the Hilbert space induced by a representation of a group:
\begin{align}
\mathcal{H}=\bigoplus_{\lambda\in\Lambda} \mathcal{R}_\lambda\otimes\mathcal{M}_\lambda,
\end{align}
where $\Lambda$ is the set of the labels of inequivalent irreducible representations that appear in a given representation,  $\mathcal{R}_\lambda$ is a space carrying an irreducible representation, and $\mathcal{M}_\lambda$ is a space carrying a trivial representation.
We note that a representation is called irreducible if it cannot be seen as the composition of simpler representations, or equivalently, it does not have any invariant subspace.
Correspondingly, the symmetry representation $U_g$ can be decomposed in the following form (see Fig.~\ref{fig4}(a)): 
\begin{align}
    U_g=\bigoplus_{\lambda\in\Lambda} U_{\lambda g}\otimes I_{\mathcal{M}_\lambda}, \label{eq:irreducibledecomposition_main}
\end{align}
where $\{U_{\lambda g}\}_{g\in G}$ is an irreducible representation of $G$ acting on $\mathcal{R}_\lambda$, and $I_\mathcal{\mathcal{M}_\lambda}$ is the identity operator on $\mathcal{M}_\lambda$.
For example, $SU(2)$ symmetry representation on the system composed of two spin-1/2 systems can be decomposed into spin-0 representation on the singlet space and spin-1 representation on the triplet space.
Under the decomposition~\eqref{eq:irreducibledecomposition_main}, Schur's lemma (e.g., Proposition~4.8 of Ref.~\cite{Knapp2002}) states that every symmetry-respecting Hamiltonian $H$ and symmetry-respecting unitary operator $U$ can be written as 
\begin{align}
	H=\bigoplus_{\lambda\in\Lambda} I_{\mathcal{R}_\lambda}\otimes H_\lambda,\ U=\bigoplus_{\lambda\in\Lambda} I_{\mathcal{R}_\lambda}\otimes U_\lambda \label{eq:symmetry_respecting_HU}
\end{align}
with some Hermitian operator $H_\lambda$ and unitary operator $U_\lambda$ on $\mathcal{M}_\lambda$,
where $I_\mathcal{\mathcal{R}_\lambda}$ is the identity operator on $\mathcal{R}_\lambda$ (see Fig.~\ref{fig4}(b)).

In order to describe the characterization of symmetry-protected passivity, we introduce the following symmetrized state:
\begin{align}
    \sigma:=\int_G dg\ U_g \rho U_g^\dag, \label{eq:symmetrization}
\end{align}
where $dg$ is the group-invariant (Haar) measure over $G$. 
This symmetrizing mapping is studied in the resource theory of asymmetry~\cite{Bartlett2007}. 
Since $\sigma$ is symmetry-respecting, $\sigma$ can be written as
\begin{align}
    \sigma=\bigoplus_{\lambda\in\Lambda} I_{\mathcal{R}_\lambda}\otimes \sigma_\lambda \label{eq:symmetry_respecting_sigma}
\end{align}
with some Hermitian operator $\sigma_\lambda$ on $\mathcal{M}_\lambda$ (see Fig.~\ref{fig4}(c)).

\begin{figure}
\begin{center}
\includegraphics[width=\columnwidth]{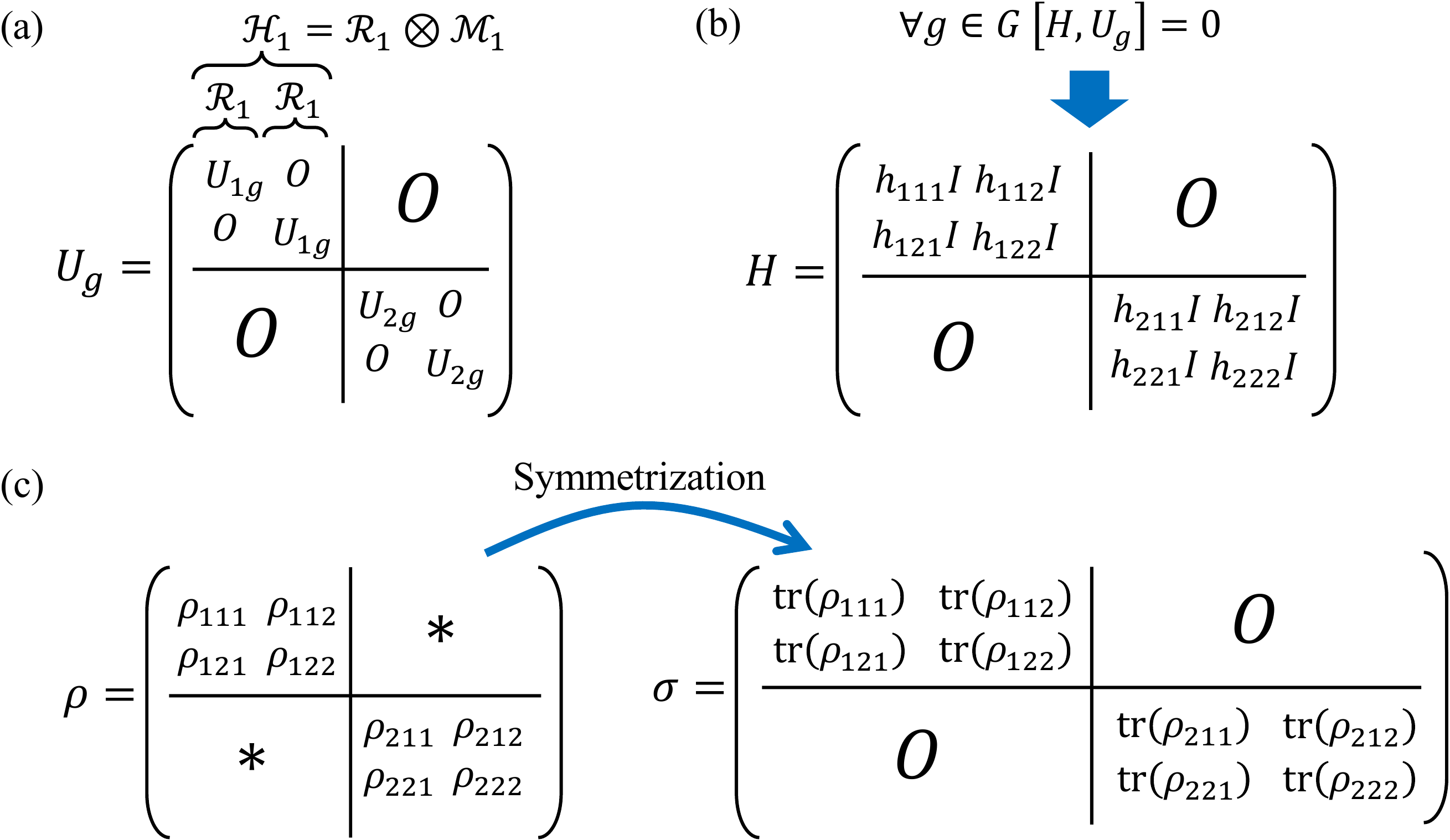}
\caption{(a) Schematic of the irreducible decomposition of a unitary representation.
(b) The corresponding form of a symmetry-respecting Hamiltonian.
(c) The effective density operator corresponding to the irreducible representation.}
\label{fig4}
\end{center}
\end{figure}

Now, symmetry-protected passivity of $\rho$ is equivalent to passivity of all $\sigma_\lambda$'s.
This can be formally stated as the following theorem. 

\begin{theorem}
    Let $G$ be a group and $\{U_g\}_{g\in G}$ be its unitary  representation.
	A state $\rho$ is symmetry-protected passive with respect to a symmetry-respecting Hamiltonian $H$, if and only if  $\sigma_\lambda$ defined by Eq.~\eqref{eq:symmetry_respecting_sigma} is passive with respect to $H_\lambda$ for all $\lambda\in\Lambda$.
\end{theorem}

\begin{proof}
    First, for any symmetry-respecting unitary operator $U$, we show that the extracted work from $\rho$ and $\sigma$ by  $U$ are the same.
    From Eq.~\eqref{eq:symmetrization},
    \begin{align}
        &W(\sigma, U) \nonumber\\
        =&\mathrm{tr}\left(\sigma \left(H-U^\dag HU\right)\right) \nonumber\\
        =&\mathrm{tr}\left(\left(\int_G dg\ U_g \rho U_g^\dag\right) \left(H-U^\dag HU\right)\right) \nonumber\\
        =&\int_G dg\ \mathrm{tr}\left(U_g \rho U_g^\dag \left(H-U^\dag HU\right)\right) \nonumber\\
        =&\int_G dg\ \mathrm{tr}\left(\rho U_g^\dag \left(H-U^\dag HU\right)U_g\right) \nonumber\\
        =&\int_G dg\ \mathrm{tr}\left(\rho \left(H-U^\dag HU\right)\right) \nonumber\\
        =&\mathrm{tr}\left(\rho \left(H-U^\dag HU\right)\right) \nonumber\\
        =&W(\rho, U). \label{eq:sigma_rho}
    \end{align}
    We next show that the extracted work from $\sigma$ can be written by the extracted work from $\sigma_\lambda$.
    From Eqs.~\eqref{eq:symmetry_respecting_HU}~and~\eqref{eq:symmetry_respecting_sigma},
    \begin{align}
        &W(\sigma, U) \nonumber\\
        =&\mathrm{tr}\left(\sigma \left(H-U^\dag HU\right)\right) \nonumber\\
        =&\mathrm{tr}\left(\left(\bigoplus_{\lambda\in\Lambda} I_{\mathcal{R}_\lambda}\otimes \sigma_\lambda\right) \left[\bigoplus_{\lambda\in\Lambda} I_{\mathcal{R}_\lambda}\otimes \left(H_\lambda-U_\lambda^\dag H_\lambda U_\lambda\right) \right]\right) \nonumber\\
        =&\mathrm{tr}\left(\bigoplus_{\lambda\in\Lambda} I_{\mathcal{R}_\lambda}\otimes \sigma_\lambda\left(H_\lambda-U_\lambda^\dag H_\lambda U_\lambda\right)\right) \nonumber\\
        =&\sum_{\lambda\in\Lambda} \mathrm{tr}\left(I_{\mathcal{R}_\lambda}\otimes \sigma_\lambda\left(H_\lambda-U_\lambda^\dag H_\lambda U_\lambda\right)\right) \nonumber\\
        =&\sum_{\lambda\in\Lambda} \mathrm{tr}(I_{\mathcal{R}_\lambda})\mathrm{tr}\left(\sigma_\lambda\left(H_\lambda-U_\lambda^\dag H_\lambda U_\lambda\right)\right) \nonumber\\
        =&\sum_{\lambda\in\Lambda} r_\lambda W(\sigma_\lambda, U_\lambda), \label{eq:sigma_lambda}
    \end{align}
    where $r_\lambda$ is the dimension of $\mathcal{R}_\lambda$.
    By comparing Eqs.~\eqref{eq:sigma_rho}~and~\eqref{eq:sigma_lambda}, the extracted work from $\rho$ is written as 
    \begin{align}
        W(\rho, U)=\sum_{\lambda\in\Lambda} r_\lambda W(\sigma_\lambda, U_\lambda).
    \end{align}
    Therefore, the maximal extracted work from $\rho$ under the symmetry constraint is given by 
    \begin{align}
        W_{\mathrm{max}, G}(\rho)=\sum_{\lambda\in\Lambda} r_\lambda W_\mathrm{max}(\sigma_\lambda),
    \end{align}
    where $W_\mathrm{max}(\sigma_\lambda)$ is the maximal extracted work from $\sigma_\lambda$ under no symmetry constraints and $W_{\mathrm{max}, G}(\rho)$ is the maximal extracted work from $\rho$ under the symmetry constraint.
    This shows that $\rho$ is symmetry-protected passive with respect to $H$, if and only if $\sigma_\lambda$ is passive in the ordinary sense with respect to $H_\lambda$.
\end{proof}

Finally, we remark on the setup with the work storage discussed in Sec.~\ref{sec:work_storage_main}. 
From Eq.~\eqref{eq:WS_ergotropy_relation}, we can prove that $\rho$ is symmetry-protected passive with the initial state $\rho_\mathrm{W}$ of the work storage, if and only if $\mathcal{D}_{\rho_\mathrm{W}}(\rho)$ is symmetry-protected passive in the setup without the work storage.
In order to prove this, we can prove (Lemma~S12 of Supplemental Material) that the Kitaev construction~\eqref{eq:unitary_correspondence} gives a one-to-one correspondence between symmetry-respecting unitaries in the setups with and without the work storage.
Therefore, the proof goes in the same way as that in Sec.~\ref{subsec:qw_no_symmetry}, and even under symmetry constraints, the maximal extracted work from $\rho$ with the work storage $\rho_\mathrm{W}$ equals the maximal extracted work from $\mathcal{D}_{\rho_\mathrm{W}}(\rho)$ without the work storage.

\section{Proof of Theorem~\ref{thm:GFHCPmain} for the dimer model} \label{sec:proof}

In this Appendix, we present a full proof of Theorem~\ref{thm:GFHCPmain} in the special case of the dimer model introduced in Sec.~\ref{subsec:examples}, as a simplest nontrivial example that has non-commutative symmetry.
See Supplemental Material for the complete proof for the general case.

In the proof of the \textit{only if} part of Theorem~\ref{thm:GFHCPmain}, we consider work extraction by a series of symmetry-respecting unitary operators.
We prove that if we cannot extract positive work from multiple copies of a state by any of those operations, then the state is a GGE at positive temperature.
This is a generalization of the proof of Proposition~\ref{prop:U_complete} in Sec.~\ref{subsec:qw_no_symmetry}.

We use the same notations as in Sec.~\ref{subsec:examples} for the dimer model.
Specifically, the Hamiltonian is given by $H = \bm s_1 \cdot \bm s_2$,
and we denote  the total spin operator in the $\alpha$-direction of the dimer by $Q^\alpha:=s_1^\alpha\otimes I+I\otimes s_2^\alpha$ for $\alpha=x, y, z$.
Let $\rho$ be the initial state of the dimer.

\begin{proof}

To prove Theorem~\ref{thm:GFHCPmain} for the dimer setup, we consider the following three steps.

\underline{Step 1.} (Proposition~S3 of Supplemental Material)
First, we prove that if a state $\rho$ is completely passive, then $\rho^{\otimes 2}$ commutes with the spin inner product $\bm{Q}\cdot\bm{Q}:=\sum_{\alpha=x, y, z} Q^\alpha\otimes Q^\alpha$.
For that purpose, we construct a series of unitary operators such that if positive work cannot be extracted from any number of multiple copies of a state $\rho$ by the operations, then $\rho^{\otimes 2}$ commutes with $\bm{Q}\cdot\bm{Q}$.

Let the spectral decomposition of $\bm{Q}\cdot\bm{Q}$ be written as $\bm{Q}\cdot\bm{Q}=\sum_{\omega} \omega P_\omega$, where $\omega$ is an eigenvalue and $P_\omega$ is the projection operator onto the eigenspace of $\omega$.
We take arbitrary $P_\omega$ and consider unitary operators $U_m:=I-\sum_{i, j\in\{0, 1\}} (-1)^{i-j}R_{ij}^{\otimes m}\otimes\ket{\Psi_i}\bra{\Psi_j}$ acting on  $4m+3$ copies of the system for $m\in\mathbb{N}$, where
$R_{ij}:=\frac{1}{2}[I-(-1)^i T][P_\omega\otimes(I-P_\omega)][I-(-1)^j T]$ with the swapping operator $T$ of the states of the two pairs of dimers, and
$\ket{\Psi_0}:=\ket{s}\ket{s}\ket{s}$, $\ket{\Psi_1}:=\frac{1}{\sqrt{6}}\sum_{i, j, k\in\{1, 2, 3\}} \epsilon_{ijk}\ket{t_i}\ket{t_j}\ket{t_k}$ with the singlet state $\ket{s}$, the triplet states $\ket{t_1}, \ket{t_2}, \ket{t_3}$ of a dimer and the Levi-Civita symbol $\epsilon_{ijk}$.
$U_m$ is symmetry-respecting because $P_\omega$ and $T$ commute with $Q^{\alpha(2)}$ and $Q^{\alpha(4)}$ respectively, and $\ket{\Psi_i}$ is an eigenstate of $Q^{\alpha(3)}$ with eigenvalue 0 for $\alpha=x, y, z$, where $Q^{\alpha(N)}$ is defined by Eq.~\eqref{eq:extensive_charge}.

We calculate the extracted work $W(\rho^{\otimes 4m+3}, U_m)$ from $\rho^{\otimes 4m+3}$ by  $U_m$.
Since $R_{ij}$ commutes with $H^{(4)}$ defined by Eq.~\eqref{eq:extensive_Hamiltonian} and $\ket{\Psi_i}$ satisfies $H^{(3)}\ket{\Psi_0}=-\frac{9}{4}\ket{\Psi_0}$, $H^{(3)}\ket{\Psi_1}=\frac{3}{4}\ket{\Psi_1}$, we obtain
\begin{align}
    &H^{(4m+3)}-U_m^\dag H^{(4m+3)}U_m \nonumber\\
    =&3\left(R_{11}^{\otimes m}\otimes\ket{\Psi_1}\bra{\Psi_1}-R_{00}^{\otimes m}\otimes\ket{\Psi_0}\bra{\Psi_0}\right).
\end{align}
Therefore, $W(\rho^{\otimes 4m+3}, U_m)$ is given by
\begin{align}
    &W(\rho^{\otimes 4m+3}, U_m) \nonumber\\
    =&\mathrm{tr}\left(\rho^{\otimes 4m+3}\left(H^{(4m+3)}-U_m^\dag H^{(4m+3)}U_m\right)\right) \nonumber\\
    =&3\left[\mathrm{tr}\left(\rho^{\otimes 4}R_{11}\right)^m \braket{\Psi_1|\rho^{\otimes 3}|\Psi_1}\right. \nonumber\\
    &\ \ \ \ \ \ \ \ \ \ \left.-\mathrm{tr}\left(\rho^{\otimes 4}R_{00}\right)^m\braket{\Psi_0|\rho^{\otimes 3}|\Psi_0}\right] \nonumber\\
    =&a\left[1-b\left(\frac{\mathrm{tr}(\rho^{\otimes 4}R_{00})}{\mathrm{tr}(\rho^{\otimes 4}R_{11})}\right)^m\right] \nonumber\\
    =&a\left[1-b\left(1+c\|[\rho^{\otimes 2}, P_\omega]\|_\mathrm{HS}^2\right)^{-m}\right], \label{eq:general_extracted_work}
\end{align}
where $a:=3\left[\mathrm{tr}\left(\rho^{\otimes 4}R_{11}\right)\right]^m \braket{\Psi_1|\rho^{\otimes 3}|\Psi_1} > 0$, $b:=\braket{\Psi_0|\rho^{\otimes 3}|\Psi_0}/\braket{\Psi_1|\rho^{\otimes 3}|\Psi_1} > 0$, $c:=\left[\mathrm{tr}\left(\rho^{\otimes 4}R_{00}\right)\right]^{-1} > 0$, and we used
\begin{align}
    \mathrm{tr}\left(\rho^{\otimes 4}R_{11}\right)-\mathrm{tr}\left(\rho^{\otimes 4}R_{00}\right)=\left\|\left[\rho^{\otimes 2}, P_\omega\right]\right\|_\mathrm{HS}^2. \label{eq:offdiagonal}
\end{align}

If $W(\rho^{\otimes 4m+3}, U_m)\leq0$ holds for all $m$, then $\|[\rho^{\otimes 2}, P_\omega]\|_\mathrm{HS}=0$, i.e., $[\rho^{\otimes 2}, P_\omega] = 0$ must be satisfied.
Since this holds for all $P_\omega$, $\rho^{\otimes 2}$ commutes with $\bm{Q}\cdot\bm{Q}$.
Note that $b\geq 1$ is shown later, implying that $\|[\rho^{\otimes 2}, P_\omega]\|_\mathrm{HS}=0$ is sufficient for  $W(\rho^{\otimes 4m+3}, U_m)\leq0$.

\underline{Step 2.} (Proposition~S5 of Supplemental Material)
Next, we prove that if $\rho^{\otimes 2}$ commutes with $\bm{Q}\cdot\bm{Q}$, then $\rho$ can be written as the product of a symmetry-respecting operator and the exponential of a linear combination of the conserved charges $\{Q^\alpha\}_{\alpha=x, y, z}$.
We define $\xi:=-\log(\rho)$, $\mathcal{P}(\xi):=\sum_{\alpha} \frac{1}{2}\mathrm{tr}(\xi Q^\alpha)Q^\alpha$ and $\eta:=\xi-\mathcal{P}(\xi)$.
$\mathcal{P}(\xi)$ can be seen as the projection of $\xi$ onto the linear subspace spanned by $\{Q^\alpha\}$ in terms of the Hilbert-Schmidt inner product in the operator space.
$\xi$ and $\eta$ satisfy the following relation for $\alpha=x, y, z$:
\begin{align}
    &\mathrm{tr}_{\mathcal{H}_2}\left(\left(I\otimes Q^\alpha\right)\left[\xi\otimes I+I\otimes\xi, \bm{Q}\cdot\bm{Q}\right]\right) \nonumber\\ 
    =&\sum_{\beta} \mathrm{tr}_{\mathcal{H}_2}\left(\left[\xi, Q^\beta\right]\otimes Q^\alpha Q^\beta+Q^\beta\otimes Q^\alpha\left[\xi, Q^\beta\right]\right) \nonumber\\
    =&\sum_{\beta} \left[\mathrm{tr}\left(Q^\alpha Q^\beta\right)\left[\xi, Q^\beta\right]-\mathrm{tr}\left(\xi\left[Q^\alpha, Q^\beta\right]\right)Q^\beta\right] \nonumber\\
    =&\sum_{\beta} \left[2\delta_{\alpha\beta}\left[\xi, Q^\beta\right]-\mathrm{tr}\left(\xi\sum_{\gamma}\epsilon_{\alpha\beta\gamma}Q^\gamma\right)Q^\beta\right] \nonumber\\
    =&2\left[\xi, Q^\alpha\right]-\sum_{\gamma} \left[\mathrm{tr}\left(\xi Q^\gamma\right)\sum_{\beta} \epsilon_{\gamma\alpha\beta}Q^\beta\right] \nonumber\\
    =&2\left[\xi, Q^\alpha\right]-\sum_{\gamma} \mathrm{tr}\left(\xi Q^\gamma\right)\left[Q^\gamma, Q^\alpha\right] \nonumber\\
    =&2\left[\xi-\sum_{\gamma} \frac{1}{2}\mathrm{tr}\left(\xi Q^\gamma\right)Q^\gamma, Q^\alpha\right] \nonumber\\
    =&2\left[\xi-\mathcal{P}(\xi), Q^\alpha\right] \nonumber\\
    =&2\left[\eta, Q^\alpha\right],  \label{eq:partial_trace}
\end{align}
where $\mathcal{H}_2$ is the Hilbert space of the second dimer and $\epsilon_{\alpha\beta\gamma}$ is the Levi-Civita symbol with $\epsilon_{xyz}=1$.
If $\rho^{\otimes 2}$ commutes with $\bm{Q}\cdot\bm{Q}$, $\xi^{(2)}=-\log(\rho^{\otimes 2})$ also commutes with $\bm{Q}\cdot\bm{Q}$.
Then from Eq.~\eqref{eq:partial_trace}, we get $[\eta, Q^\alpha]=0$, i.e., $\eta$ is symmetry-respecting.
Therefore, $\rho$ can be written as $\rho=\exp(-\xi)=\exp(-\eta)\exp(-\sum_{\alpha=x, y, z} \mu_\alpha Q^\alpha)$, where $\mu_\alpha:=\frac{1}{2}\mathrm{tr}(\xi Q^\alpha)$.

\underline{Step 3.} (Proposition~S4 of Supplemental Material)
Finally, we combine the results in Steps~1~and~2, consider work extraction again and prove Theorem~\ref{thm:GFHCPmain}.
Suppose that $\rho$ is symmetry-protected completely passive.
From Steps~1~and~2, $\rho$ can be written as $\rho=\exp(-\eta)\exp(-\sum_{\alpha=x, y, z} \mu_\alpha Q^\alpha)$ with some symmetry-respecting operator $\eta$ and $\mu_\alpha\in\mathbb{R}$. 
Since the total spin symmetry operators irreducibly act on the singlet space and the triplet space, Schur's lemma implies that symmetry-respecting $\eta$ can be written as $\eta=c_\mathrm{s}\ket{s}\bra{s}+c_\mathrm{t}\sum_{i=1}^3 \ket{t_i}\bra{t_i}$ with some $c_\mathrm{s}, c_\mathrm{t}\in\mathbb{R}$.
Since $H=-\frac{3}{4}\ket{s}\bra{s}+\frac{1}{4}\sum_{i=1}^3 \ket{t_i}\bra{t_i}$, $\eta$ can be written as $\eta=\mu I+\beta H$ with some $\mu, \beta\in\mathbb{R}$.
Note that such a simple relation between $\eta$ and $H$ is specific to this dimer model that has only two energy levels. 
Then $\rho$ can be written as $\rho=\exp(-\mu I-\beta H)\exp(-\sum_{\alpha=x, y, z} \mu_\alpha Q^\alpha)=\exp(-\beta H-\sum_{\alpha=x, y, z} \mu_\alpha Q^\alpha)/\exp(\mu)$.
From the normalization condition, $\exp(\mu)=\mathrm{tr}(\exp(-\beta H-\sum_{\alpha=x, y, z} \mu_\alpha Q^\alpha))=Z_\mathrm{GGE}$ and we obtain $\rho=\exp(-\beta H-\sum_{\alpha=x, y, z} \mu_\alpha Q^\alpha)/Z_\mathrm{GGE}$.
In order to prove that $\beta\geq0$, we consider the case where $m=0$ in Eq.~\eqref{eq:general_extracted_work}. 
In this case, the extracted work is given by $W(\rho^{\otimes 3}, U_0)=a(1-b)=a[1-\exp(3\beta)]$.
Since $\rho$ is symmetry-protected completely passive, $W(\rho^{\otimes 3}, U_0)\leq0$ and thus we get $\beta\geq0$ (i.e., $b \geq 1$).
\end{proof}

In the proof for the general case (see Supplemental Material), we decompose a connected compact Lie group into a compact Abelian Lie group and a semisimple Lie group by Levi decomposition (Theorem~4.29~of~\cite{Knapp2002}).  
When we deal with the semisimple Lie group, we use a Casimir operator (Lemma~3.3.7~of~\cite{Goodman2009}) as a generalization of the spin inner product $\bm{Q}\cdot\bm{Q}$ along with a generalized version of totally antisymmetric states $\ket{\Psi_i}$. 
We also use the fact that every conserved charge $Q$ associated with a semisimple Lie group symmetry satisfies $\mathrm{tr}(Q)=0$.
When we deal with the Abelian Lie group, we use the notion of virtual temperature introduced in Ref.~\cite{Skrzypczyk2015} under symmetry constraints.

\section{Finite-group symmetry and time-reversal symmetry} \label{sec:discrete_symmetry}

We consider the case where symmetry constraints on the operations are described by some finite groups. 
Specifically, we prove that if the symmetry group is a finite cyclic group or a dihedral group, every symmetry-protected completely passive state is just a conventional Gibbs ensemble. 
We note that  in a one-dimensional lattice with the periodic boundary condition, spatial translation generates a finite cyclic group, and the combination of spatial translation and inversion generates a dihedral group (see Fig.~\ref{fig5}).
In addition, we investigate the case of time-reversal symmetry without spin degrees of freedom, which is an anti-unitary symmetry (class AI).
In this case, every symmetry-protected completely passive state is again a conventional Gibbs ensemble.

First, we consider the case of finite group symmetry.
As for passivity, the same argument as in Appendix~\ref{sec:passivity} can be applied, where the symmetrizing mapping~\eqref{eq:symmetrization} can be replaced with $\frac{1}{|G|}\sum_{g\in G} U_g \rho U_g^\dag$ with $|G|$ being the order of $G$.
Then, $\rho$ is symmetry-protected passive if and only if all $\sigma_\lambda$'s are passive.

\begin{figure}
\begin{center}
\includegraphics[width=\columnwidth]{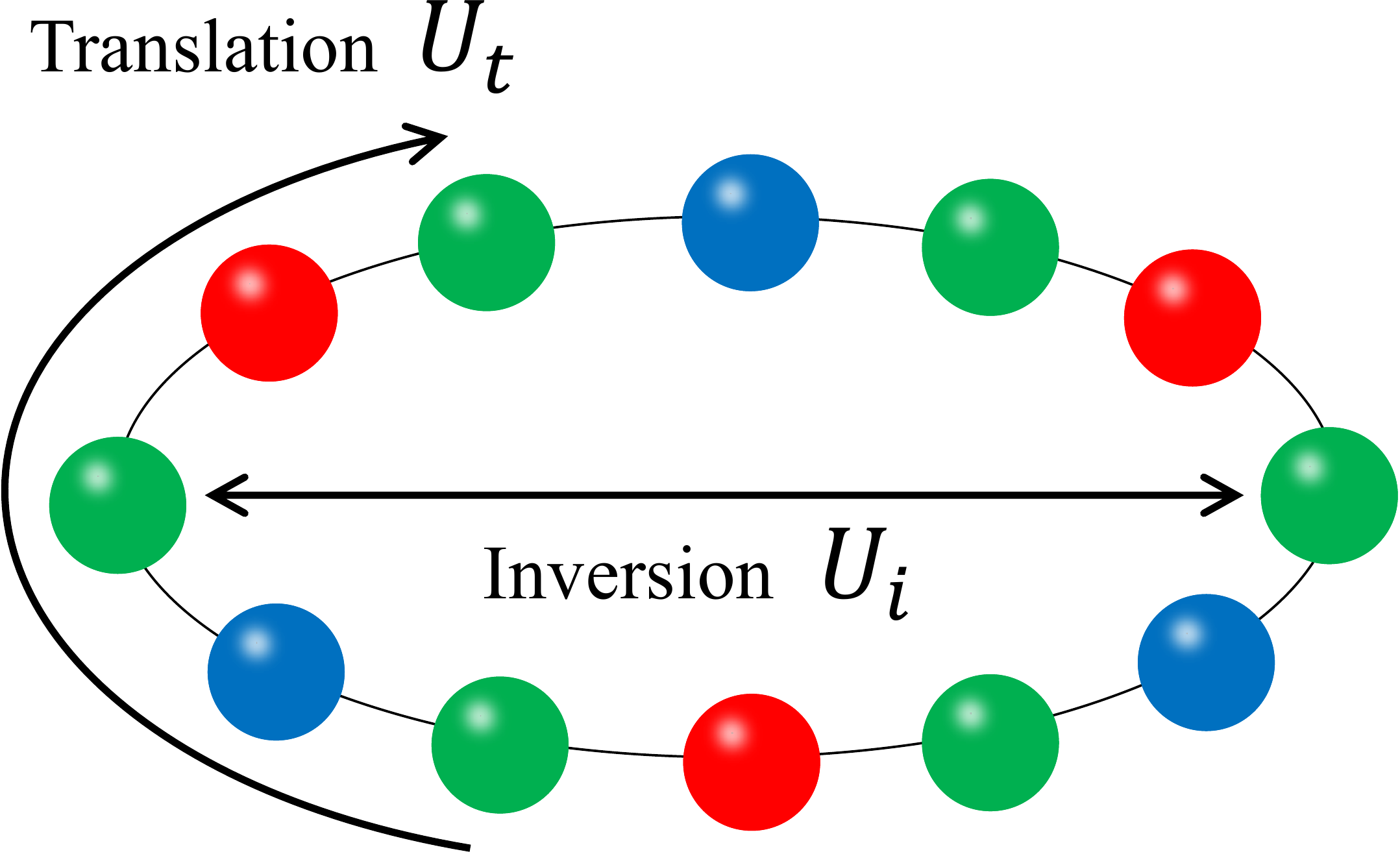}
\caption{An example of dihedral group symmetry in a one-dimensional lattice with the periodic boundary condition. 
The system is invariant under spatial translation and inversion;
These two operations generate a dihedral group.}
\label{fig5}
\end{center}
\end{figure}

As for complete passivity, we restrict ourselves to the cases of finite cyclic group symmetry and dihedral group symmetry, and prove that symmetry-protected completely passive states are only Gibbs ensembles at positive temperature.
The proof is parallel to that in Appendix~\ref{sec:proof} except for Step 2.

The reason why the same argument as Step 2 cannot be applied is that  for finite groups, there does not exist an explicit counterpart of Casimir operators in Lie groups.
Instead, we can show (Proposition~S9 in Supplemental Material) that if $\rho$ is symmetry-protected completely passive, then $\rho$ is symmetry-respecting, i.e., $[\rho, U_g]=0$ for all $g \in G$.
In the case of a finite cyclic group, the proof of this statement is straightforward.
Since $U_g$'s are symmetry-respecting, by a similar argument as Step 1 of Appendix~\ref{sec:proof}, every symmetry-protected completely passive state commutes with all $U_g$'s.
On the other hand, in the case of a dihedral group, the proof is more complicated due to its non-commutativity.
We can construct symmetry-respecting operators with the projection operators onto the eigenspaces of a symmetry operator $U_t$, where $t$ is an element of a dihedral group of order $2n$ and satisfies $t^n=1$ (see Supplemental Material for details). 

We note that for the case of general finite groups, the characterization of completely passive states is an open problem, while we conjecture that only Gibbs ensembles are symmetry-protected completely passive as is the foregoing cases.

Next, we consider the case of time-reversal symmetry without spin degrees of freedom.
In this case, time-reversal operator $\mathcal{T}$ is represented by the complex conjugation operator with respect to some basis of the Hilbert space of the system. 
We can prove (Theorem~S2 of Supplemental Material) that a state $\rho$ is passive under the time-reversal symmetry, if and only if the time-reversal symmetrized state $\sigma:=(\rho+\mathcal{T}\rho \mathcal{T}^{-1})/2$ is passive in the ordinary sense.

The outline of the proof of the above statement is as follows.
For any symmetry-respecting unitary $U$, the extracted work from $\rho$ and $\sigma$ by  $U$ are the same,
which implies that the maximal extracted work from $\rho$ and $\sigma$ under the time-reversal symmetry are the same.
Since both of $\sigma$ and $H$ are symmetry-respecting, $\sigma$ can be converted to an ordinary passive state by a symmetry-respecting unitary operation.
This implies that the maximal extracted work from $\sigma$ with and without the time-reversal symmetry are the same.
By combining these two relations, the maximal extracted work from $\rho$ under the time-reversal symmetry and that from $\sigma$ without the time-reversal symmetry are the same.
Therefore, $\rho$ is symmetry-protected completely passive if and only if $\sigma$ is passive.

We can also prove (Theorem~S5 of Supplemental Material) that completely passive states under the time-reversal symmetry are only Gibbs ensembles.
It is obvious that Gibbs ensembles are symmetry-protected completely passive, and therefore we only need to prove the converse.
Since the Hamiltonian is symmetry-respecting, all the projection operators onto the energy eigenspaces are symmetry-respecting. 
In the same way as Step 1 of Appendix~\ref{sec:proof}, every symmetry-protected completely passive state commutes with all symmetry-respecting operators, and thus it commutes with all the projection operators onto the energy eigenspaces.
This implies that the density operators of the state is diagonal in the energy eigenbasis.
The rest of the proof can be constructed by a standard technique considering virtual temperatures~\cite{Skrzypczyk2015}.

\unappendix
\clearpage

\onecolumngrid

\begin{center}
{\large \bf Supplemental Material: Characterizing symmetry-protected thermal equilibrium by work extraction}\\
\vspace*{0.3cm}
Yosuke Mitsuhashi$^{1}$, Kazuya Kaneko$^{1}$, and Takahiro Sagawa$^{1,2}$ \\
\vspace*{0.1cm}

{$^{1}$Department of Applied Physics, The University of Tokyo, Tokyo 113-8656, Japan} 

{$^{2}$Quantum-Phase Electronics Center (QPEC), The University of Tokyo, Tokyo 113-8656, Japan}%
\end{center}


\setcounter{equation}{0}
\renewcommand{\theequation}{S\arabic{equation}}
\setcounter{definition}{0}
\renewcommand{\thedefinition}{S\arabic{definition}}
\setcounter{theorem}{0}
\renewcommand{\thetheorem}{S\arabic{theorem}}
\setcounter{proposition}{0}
\renewcommand{\theproposition}{S\arabic{proposition}}
\setcounter{lemma}{0}
\renewcommand{\thelemma}{S\arabic{lemma}}
\setcounter{corollary}{0}
\renewcommand{\thecorollary}{S\arabic{corollary}}
\setcounter{section}{0}
\renewcommand{\thesection}{S\arabic{section}}

\setcounter{subsection}{0}
\renewcommand{\thesubsection}{\arabic{subsection}}

\setcounter{subsubsection}{0}
\renewcommand{\thesubsubsection}{\arabic{subsubsection}}
\setcounter{figure}{0}
\renewcommand{\thefigure}{S\arabic{figure}}
\setcounter{table}{0}
\renewcommand{\thetable}{S\arabic{table}}


\section*{}

The main purpose of this Supplemental Material is to present the complete proof of Theorem~1 of the main text (Ref.~\cite{Main}), which is rephrased as  Theorem~\ref{thm:GFHCP1} of this Supplemental Material in a more rigorous manner.
This theorem states that every symmetry-protected completely passive state is a generalized Gibbs ensemble (GGE). 
Another main theorem is Theorem 2 of the main text~\cite{Main}, which is Theorem~\ref{thm:GFHCP1} of this Supplemental Material.
This theorem  states that even if we explicitly include a quantum work storage, symmetry-protected completely passive states are still only GGEs.
Moreover, we prove the condition for symmetry-protected passivity in Theorem~\ref{thm:GFHP} (Theorem~3 of the main text~\cite{Main}).
We also deal with symmetry constraints described by finite group symmetry and time-reversal symmetry.

The structure of this Supplemental Material is as follows.
In Sec.~\ref{sec:formal_statement}, we formally define symmetry-protected passivity and complete passivity in the setup without the quantum work storage, and present the necessary and sufficient conditions for them.
Theorem~\ref{thm:GFHP} and Theorem~\ref{thm:THP} provide the conditions for passivity with group symmetry (including Lie groups and finite groups) and time-reversal symmetry, respectively.
Theorem~\ref{thm:GFHCP}, Theorem~\ref{thm:finiteGFHCP} and Theorem~\ref{thm:THCP} provide the conditions for complete passivity with Lie group symmetry, finite group symmetry, and time-reversal symmetry, respectively.
In Sec.~\ref{sec:main_theorems}, we prove Theorem~\ref{thm:GFHP} to Theorem~\ref{thm:THCP}.
In order to prove Theorem~\ref{thm:GFHCP}, Theorem~\ref{thm:finiteGFHCP} and Theorem~\ref{thm:THCP}, we respectively prove Proposition~\ref{prop:GFworkextop}, Proposition~\ref{prop:finiteGFworkextop}, and Proposition~\ref{prop:Tworkextop}, which give stronger statements than necessary for the proof of the above-mentioned theorems, as these propositions are used in the setup with a quantum work storage later.
In Sec.~\ref{sec:work_storage}, we introduce the setup with a quantum work storage and present the conditions for symmetry-protected passivity and complete passivity.
Theorem~\ref{thm:GFHP1} to Theorem~\ref{thm:THCP1} are the counterparts of Theorem~\ref{thm:GFHP} to Theorem~\ref{thm:THCP}, respectively.
In Sec.~\ref{sec:technical_lemmas}, we prove technical lemmas used in the foregoing sections.

We list symbols that often appear in this paper in Table~\ref{table1}.

\begin{longtable}[c]{ll}
\caption{List of the symbols.}
\label{table1}\\
	\hline
 	Symbol  & Meaning \\
	 \hline\hline
 	\endfirsthead

 	\hline
	 Symbol  & Meaning \\
	 \hline\hline
	 \endhead

	\hline
	\endfoot

	\hline

 	$\mathbb{N}$ & $\{n\in\mathbb{Z}\ |\ n\geq0\}$.\\
 	$\mathcal{B}(\mathcal{K})$ & Set of all bounded linear operators on a Hilbert space $\mathcal{K}$. \\ 
	$\mathcal{B}^\mathrm{H}(\mathcal{K})$ & Set of all bounded Hermitian operators on a Hilbert space $\mathcal{K}$. \\
	$\mathcal{B}^+(\mathcal{K})$ & Set of all bounded positive operators on a Hilbert space $\mathcal{K}$. \\
	$\mathcal{B}^{++}(\mathcal{K})$ & Set of all bounded positive definite operators on a Hilbert space $\mathcal{K}$. \\
	$\mathcal{C}$ & Mapping from $\mathcal{U}(\mathcal{H})$ to $\mathcal{U}^\mathrm{WS}(\mathcal{H}, \mathcal{H}_\mathrm{W})$ defined as Eq.~\eqref{eq:SCFQcorrespondence}. \\
	$\mathcal{D}_{H, \rho_\mathrm{W}}$ & Mapping from $\mathcal{B}^+(\mathcal{H})$ to $\mathcal{B}^+(\mathcal{H})$ defined as Eq.~\eqref{eq:mapDdef}. \\
	$\mathcal{H}$ & Hilbert space of the system of interest. \\
	$\mathcal{H}_\mathrm{W}$ & Hilbert space of the work storage. \\
	$\mathcal{H}_\lambda$ & Subspace of $\mathcal{H}$ decomposed by equivalent irreducible representations.\\
	$\mathcal{M}_\lambda$ & Subspace carrying the trivial representation. \\
	$\mathcal{O}$ & Mapping from an operator quartet to an operator defined as Eq.~\eqref{eq:unitary}. \\
	$\mathcal{R}_\lambda$ & Subspace carrying an irreducible representation. \\
	$\mathcal{S}_\mathcal{T}$ & Time-reversal symmetrizing mapping defined as Eq.~\eqref{eq:Tsymmetrizer}.\\
	$\mathcal{T}$ & Complex conjugation operator. \\
	$\mathcal{U}(\mathcal{K})$ & Set of all unitary operators on a Hilbert space $\mathcal{K}$. \\
	$\mathcal{U}_{G, F}(\mathcal{H}),\ \mathcal{U}_\mathcal{T}(\mathcal{H})$ & Set of all symmetry-respecting operators defined in Definition~\ref{def:symmetry_respecting_operator}.\\
	$\mathcal{U}^\mathrm{WS}(\mathcal{H}, \mathcal{H}_\mathrm{W})$ & Set of all WS-operators defined in Definition~\ref{def:WSoperator}.\\
	$\mathcal{W}_H(\rho)$ & Ergotropy of $\rho$ defined in Definition~\ref{def:ergotropy}. \\
	$\mathcal{W}_H^\mathrm{WS}(\rho)$ & WS-Ergotropy of $\rho$ defined in Definition~\ref{def:WSergotropy}. \\
	$B_f$ & Bilinear form on $\mathfrak{g}$ defined as Eq.~\eqref{eq:bilineardef}.\\
	$E_{\lambda j}$ & Eigenvalue of $H_\lambda$ defined in Eq.~\eqref{eq:Hamiltoniandecomp2}. \\
	$F$ & Unitary representation of a group $G$ acting on $\mathcal{H}$. \\
	$G$ & Group. \\
	$H$ & Hamiltonian of the system of interest. \\
	$H_\lambda$ & Component of $H$ acting on $\mathcal{M}_\lambda$ defined as Eq.~\eqref{eq:Hamiltoniandecomp}. \\
	$I_\mathcal{K}$ & Identity operator on a Hilbert space $\mathcal{K}$. \\
	$K$ & $\{(\lambda, j)\ |\ \lambda\in\Lambda_F,\ j=1, \cdots, m_\lambda\}$. \\
	$N$ & Connected center of a Lie group. \\
	$S$ & Semisimple Lie group. \\
	$T^a$ & Torus of degree $a$. \\
	$W(\rho, H, U)$ & Extracted work defined in Definition~\ref{def:extwork}. \\
	$W^\mathrm{WS}(\rho, H, \rho_\mathrm{W}, U)$ & WS-extracted work defined in Definition~\ref{def:WSextwork}. \\
	$\{W_i\}$ & Basis of the Lie algebra of $N$ defined in Lemma~\ref{lem:Liealgebrabasis}.  \\
	$d$ & Dimension of $\mathcal{H}$. \\
	$f$ & Associated representation of the Lie algebra of a Lie group. \\
	$\mathrm{i}$ & Imaginary unit. \\
	$m_\lambda$ & Dimension of $\mathcal{M}_\lambda$.\\
	$p$ & Momentum operator on the work storage.\\
	$\ket{q}_p$ & Momentum eigenstate with eigenvalue $q$.\\
	$r_\lambda$ & Dimension of $\mathcal{R}_\lambda$.\\
	$x$ & Position operator on the work storage.\\
	$\Lambda_F$ & Set of labels of irreducible representations in $F$. \\
	$\ket{\Phi(\bm{n})}$ & State defined as Eq.~\eqref{eq:determinantstate}.\\
	$\iota_\lambda$ & Inclusion of $\mathcal{H}_\lambda$ into $\mathcal{H}$. \\
	$\rho$ & Density operator of a state of the system of interest.\\
	$\rho_\mathrm{W}$ & Density operator of a state of the work storage.\\
	$\{\ket{\phi_{\lambda i}}\}_{i=1}^{r_\lambda}$ & Arbitrary fixed orthonormal basis of $\mathcal{R}_\lambda$.\\
	$\{\ket{\psi_{\lambda j}}\}_{j=1}^{m_\lambda}$ & Eigenstates of $H_\lambda$ defined in Eq.~\eqref{eq:Hamiltoniandecomp2}.\\
	$\mathfrak{g}$ & Lie algebra of a Lie group.\\
\end{longtable}

\section{Formal statement of the main results}
\label{sec:formal_statement}

For a Hilbert space $\mathcal{K}$, we denote by $\mathcal{B}(\mathcal{K})$, $\mathcal{B}^\mathrm{H}(\mathcal{K})$, $\mathcal{B}^+(\mathcal{K})$, $\mathcal{B}^{++}(\mathcal{K})$ and $\mathcal{U}(\mathcal{K})$ the set of all bounded linear, Hermitian, positive, positive definite and unitary operators on $\mathcal{K}$, respectively. 
We denote by $I_\mathcal{K}$ the identity operator on $\mathcal{K}$.
Let $\mathcal{H}$ be a finite-dimensional Hilbert space of the system, $d$ be the dimension of $\mathcal{H}$, $G$ be a group with a unitary representation $F$ acting on $\mathcal{H}$, and $\mathcal{T}$ be the complex conjugation operator w.r.t. some basis. 
The imaginary unit is denoted by $\mathrm{i}$.
We include 0 in $\mathbb{N}$.

\subsection{Formal definitions}

First, we formulate work extraction in the setup without explicitly including the work storage.
Extracted work is defined as the difference between the energy expectation values of the states before and after an operation.

\begin{definition} [Extracted work] \label{def:extwork}
	Let $\rho\in\mathcal{B}^+(\mathcal{H})$, $H\in\mathcal{B}^\mathrm{H}(\mathcal{H})$ and $U\in\mathcal{U}(\mathcal{H})$. 
	The extracted work from a state $\rho$ under a Hamiltonian $H$ by the action of $U$ is defined as
	\begin{align}
		W(\rho, H, U):=\mathrm{tr}(\rho H)-\mathrm{tr}(\rho U^\dagger HU).
	\end{align}
\end{definition}

	In the conventional definition of passive states, it is supposed that all unitary operations are allowed \cite{Pusz1978, Lenard1978}.
	In the present study, on the other hand, we consider the class of operations that respect symmetry constraints imposed on the system.
	When we consider group symmetry, the symmetry constraints are described by the commutativity with a unitary representation of the group.
	On the other hand, for time-reversal symmetry, we only deal with the case without spin degrees of freedom (Class AI) in the present study, where the constraints are described by the commutativity with the (anti-unitary) complex conjugation operator.
	For these two cases, we define symmetry-respecting operators as follows.

\begin{definition} [$(G, F)$-respecting operators, $\mathcal{T}$-respecting operators] \label{def:symmetry_respecting_operator}
	An operator $U\in\mathcal{B}(\mathcal{H})$ is a $(G, F)$-respecting operator, if $U$ satisfies $U\in\mathcal{U}(\mathcal{H})$ and commutes with $F(G)$, i.e., commutes with $F(g)$ for all $g\in G$.
	An operator $U\in\mathcal{B}(\mathcal{H})$ is a $\mathcal{T}$-respecting operator, if $U$ satisfies $U\in\mathcal{U}(\mathcal{H})$ and commutes with $\mathcal{T}$.
	We define $\mathcal{U}_{G, F}(\mathcal{H})$ (resp. $\mathcal{U}_\mathcal{T}(\mathcal{H})$) as the set of all $(G, F)$-respecting (resp. $\mathcal{T}$-respecting) operators on $\mathcal{H}$. 
\end{definition}

	The set of all symmetry-respecting operators forms a group and these operators can be regarded as free operations in terms of the resource theory \cite{Chitambar2019}. 
	
	The maximal extracted work from a state is called ergotropy \cite{Allahverdyan2004}.
	We generalize this notion in the case where allowed operations are restricted by symmetry constraints, and define $(G, F)$-ergotropy and $\mathcal{T}$-ergotropy as the maximal extracted works by symmetry-respecting operations.

\begin{definition} [Ergotropy, $(G, F)$-ergotropy, $\mathcal{T}$-ergotropy] \label{def:ergotropy}
	Let $H\in\mathcal{B}^\mathrm{H}(\mathcal{H})$ and $\rho\in\mathcal{B}^+(\mathcal{H})$. 
	Ergotropy $\mathcal{W}_H(\rho)$ (resp. $(G, F)$-ergotropy $\mathcal{W}_{G, F, H}(\rho)$ or $\mathcal{T}$-ergotropy $\mathcal{W}_{\mathcal{T}, H}(\rho)$) of a state $\rho$ under a Hamiltonian $H$ is defined as the maximal extracted work from $\rho$ under the Hamiltonian $H$ by the action of operators in $\mathcal{U}(\mathcal{H})$ (resp. $\mathcal{U}_{G, F}(\mathcal{H})$ or $\mathcal{U}_\mathcal{T}(\mathcal{H})$), i.e.,
	\begin{align}
		&\mathcal{W}_H(\rho):=\max_{U\in\mathcal{U}(\mathcal{H})} W(\rho, H, U),\\
		&\mathcal{W}_{G, F, H}(\rho):=\max_{U\in\mathcal{U}_{G, F}(\mathcal{H})} W(\rho, H, U),\\
		&\mathcal{W}_{\mathcal{T}, H}(\rho):=\max_{U\in\mathcal{U}_\mathcal{T}(\mathcal{H})} W(\rho, H, U).
	\end{align}
\end{definition}

	Passive states are defined as the states from which no positive work can be extracted, i.e., the states with zero ergotropy.
	We define symmetry-protected passive states as the states from which no positive work can be extracted by any symmetry-respecting operations.

\begin{definition} [Passivity, $(G, F)$-passivity, $\mathcal{T}$-passivity]
	Let $H\in\mathcal{B}^\mathrm{H}(\mathcal{H})$ and $\rho\in\mathcal{B}^+(\mathcal{H})$. 
	A state $\rho$ is passive (resp. $(G, F)$-passive or $\mathcal{T}$-passive) w.r.t. a Hamiltonian $H$, if $\mathcal{W}_H(\rho)=0$ (resp. $\mathcal{W}_{G, F, H}(\rho)=0$ or $\mathcal{W}_{\mathcal{T}, H}(\rho)=0$).
\end{definition}

	Note that passivity is defined for positive operators that are not necessarily normalized. 
This is due to the convenience in the description of Theorem~\ref{thm:GFHP}. 
	
	Even if a state is passive, multiple copies of it are not necessarily passive, i.e., it is possible that positive work can be extracted from multiple copies of a passive state by a collective operation. 
	We define symmetry-protected completely passive states as the states such that no positive work can be extracted from any number of copies of them by any symmetry-respecting collective operations.
	Here we adopt the commutativity with the tensor product of representation of a group as the condition for symmetry-respecting collective operations.

\begin{definition} [Complete passivity, $(G, F)$-complete passivity, $\mathcal{T}$-complete passivity]
	Let $H\in\mathcal{B}^\mathrm{H}(\mathcal{H})$ and $\rho\in\mathcal{B}^+(\mathcal{H})$. 
	A state $\rho$ is completely passive (resp. $(G, F)$-completely passive or $\mathcal{T}$-completely passive) w.r.t. a Hamiltonian $H$, if $\rho^{\otimes N}$ is passive (resp. $(G, F^{\otimes N})$-passive or $\mathcal{T}$-passive) w.r.t. $H^{(N)}$ for all $N\in\mathbb{N}$, where $F^{\otimes N}$ is the tensor product of the representation, and for $\Omega\in\mathcal{B}(\mathcal{H})$, $\Omega^{(N)}\in\mathcal{B}(\mathcal{H}^{\otimes N})$ is defined as 
	\begin{align}
		\Omega^{(N)}=\sum_{i=1}^N I_\mathcal{H}^{\otimes i-1}\otimes \Omega\otimes I_\mathcal{H}^{\otimes N-i}. \label{eq:extensiveoperator}
	\end{align}
\end{definition}

Finally, we define trivial operators.
We say that $H\in\mathcal{B}^\mathrm{H}(\mathcal{H})$ is trivial, if $H$ is a scalar multiple of the identity operator.
If the Hamlitonian of a system is trivial, all states are completely passive because its energy expectation value is not changed by any unitary operations.
When a system has a Lie group symmetry, the system has conserved charges associated with the symmetry.
If the Hamiltonian is a linear combination of the identity operator and the conserved charges, the energy expectation value is invariant. 
We call such a Hamiltonian $(G, F)$-trivial.
If the Hamlitonian of a system is trivial or $(G, F)$-trivial, we cannot construct a proof for identifying completely passive states in the same way as in the general case, and thus we exclude this case from Theorems~\ref{thm:GFHCP},~\ref{thm:finiteGFHCP}~and~\ref{thm:THCP}.

In order to exactly define $(G, F)$-trivial operators, we introduce the Lie algebra associated with a Lie group.
Since a compact Lie group is isomorphic to a matrix Lie group (Corollary 4.22 of \cite{Knapp2002}), when we consider symmetry constraints described by a compact Lie group, we suppose that it is a matrix Lie group.  
We follow the convention in the physics literature, and define the Lie algebra $\mathfrak{g}$ of $G\subset GL(n, \mathbb{C})$ as $\mathfrak{g}:=\{X\in GL(n, \mathbb{C})\ |\ \forall \theta\in\mathbb{R}\ e^{\mathrm{i}\theta X}\in G\}$ and the associated representation $f$ of $\mathfrak{g}$ as the mapping that satisfies 
\begin{align}
	F(e^{\mathrm{i}X})=e^{\mathrm{i}f(X)} 
\end{align}
for all $X\in\mathfrak{g}$.
Note that the associated representation of the tensor product of representation $F^{\otimes N}$ is $f^{(N)}$. 
Then, $(G, F)$-trivial operators are defined as follows.

\begin{definition} [$(G, F)$-trivial operator]
	Let $G$ be a compact Lie group with a unitary representation $F$ acting on $\mathcal{H}$ and $\mathfrak{g}$ be the Lie algebra of $G$ with the associated representation $f$.
	$H\in\mathcal{B}(\mathcal{H})$ is $(G, F)$-trivial, if $H$ can be written as 
	\begin{align}
		H=\alpha I+f(X) \label{eq:trivialHamiltonian}
	\end{align}
with some $\alpha\in\mathbb{R}$ and $X\in\mathfrak{g}$.
\end{definition}

\subsection{Main results}
\label{subsec:main_theorems}

In this subsection, we present our main theorems.
In order to present Theorem~\ref{thm:GFHP}, we describe the irreducible decomposition of a representation of a group.
This part follows Ref.~\cite{Bartlett2007}.
The representation $F$ allows for the decomposition  of the Hilbert space $\mathcal{H}$ depending on the inequivalent irreducible representations: 
\begin{align}
	\mathcal{H}=\bigoplus_{\lambda\in\Lambda_F} \mathcal{H}_\lambda, \label{eq:Hilbert_space_decomp1}
\end{align}
where $\Lambda_F$ is the set of the labels of irreducible representations that appear once or more in $F$.
$\mathcal{H}_\lambda$ can be further decomposed as
\begin{align}
	\mathcal{H}_\lambda=\mathcal{R}_\lambda\otimes\mathcal{M}_\lambda, \label{eq:Hilbert_space_decomp2}
\end{align}
where $\mathcal{R}_\lambda$ is the subspace carrying an irreducible representation $F_\lambda$, and $\mathcal{M}_\lambda$ is the subspace carrying the trivial representation. 
We define $r_\lambda$ and $m_\lambda$ as the dimensions of $\mathcal{R}_\lambda$ and $\mathcal{M}_\lambda$.
Then, $F$ can be written as 
\begin{align}
	F=\sum_{\lambda\in\Lambda_F} \iota_\lambda (F_\lambda\otimes I_{\mathcal{M}_\lambda})\iota_\lambda^\dag, \label{eq:irreducibledecomposition}
\end{align}
where $\iota_\lambda$ is the inclusion of $\mathcal{H}_\lambda$ into $\mathcal{H}$.
From Schur's lemma, the Hamiltonian $H$ commuting with $F(G)$ can be written as
\begin{align}
	H=\sum_{\lambda\in\Lambda_F} \iota_\lambda (I_{\mathcal{R}_\lambda}\otimes H_\lambda)\iota_\lambda^\dag \label{eq:Hamiltoniandecomp}
\end{align}
with some $H_\lambda\in\mathcal{B}^\mathrm{H}(\mathcal{M}_\lambda)$.
Theorem~\ref{thm:GFHP} gives the necessary and sufficient condition for group symmetry-protected passivity.

\begin{theorem} \label{thm:GFHP}
	Let $G$ be a group with a unitary representation $F$ acting on $\mathcal{H}$, $H\in\mathcal{B}^\mathrm{H}(\mathcal{H})$ commute with $F(G)$, and $\rho\in\mathcal{B}^+(\mathcal{H})$.
	Then, a state $\rho$ is $(G, F)$-passive w.r.t. a Hamiltonian $H$, if and only if 
\begin{align}
	\rho_\lambda:=\mathrm{tr}_{\mathcal{R}_\lambda} (\iota_\lambda^\dag\rho \iota_\lambda) \label{eq:rholambdadef}
\end{align}
is passive w.r.t. $H_\lambda$ for all $\lambda\in\Lambda_F$, where $H_\lambda$ is defined as Eq.~\eqref{eq:Hamiltoniandecomp}.
\end{theorem}

Theorem~\ref{thm:THP} gives the necessary and sufficient condition for time-reversal symmetry-protected passivity.
\begin{theorem} \label{thm:THP}
	Let $H\in\mathcal{B}^\mathrm{H}(\mathcal{H})$ commute with $\mathcal{T}$ and $\rho\in\mathcal{B}^+(\mathcal{H})$.
	Then, a state $\rho$ is $\mathcal{T}$-passive w.r.t. a Hamiltonian $H$, if and only if $\mathcal{S}_\mathcal{T}(\rho)$ is passive w.r.t. $H$, where time-reversal symmetrizing mapping $\mathcal{S}_\mathcal{T}$ is defined as
	\begin{align}
		\mathcal{S}_\mathcal{T}(\rho):=\frac{1}{2}(\rho+\mathcal{T}\rho \mathcal{T}^{-1}). \label{eq:Tsymmetrizer}
	\end{align}
\end{theorem}

	Next, we identify the condition for symmetry-protected completely passive states.
	Here we assume that the density operator of the state is positive definite because we take $\log(\rho)$ in the proof.
	Unlike in the case of passivity, we separately deal with the cases where symmetry constraints are described by Lie groups and by finite groups.
	The condition for complete passivity protected by Lie group symmetry is as follows.

\begin{theorem} \label{thm:GFHCP}
	Let $G$ be a connected compact Lie group with a faithful unitary representation $F$ acting on $\mathcal{H}$, $\mathfrak{g}$ be the Lie algebra of $G$ with the associated representation $f$, $H\in\mathcal{B}^\mathrm{H}(\mathcal{H})$ commute with $F(G)$ and be not $(G, F)$-trivial and $\rho\in\mathcal{B}^{++}(\mathcal{H})$ satisfy $\mathrm{tr}(\rho)=1$.
	Then, a state $\rho$ is $(G, F)$-completely passive w.r.t. a Hamiltonian $H$, if and only if $\rho$ can be written as $\rho=\frac{1}{Z}e^{-\beta H-f(X)}$ with some $\beta\in[0, \infty)$ and $X\in\mathfrak{g}$, where $Z:=\mathrm{tr}(e^{-\beta H-f(X)})$.
\end{theorem}

	In this theorem, we assume that the representation $F$ is faithful, i.e., injective.
	This condition is usually satisfied, because when we consider a system that has symmetry of $G$ in a physics context, the faithfulness of its representation is implicitly assumed. 
	This assumption is used to prove that the bilinear form $B_f$ on $\mathfrak{g}$ defined as Eq.~\eqref{eq:bilineardef} is nondegenerate.

	For the case of finite groups, we investigate a finite cyclic group and a dihedral group, and we prove that $(G, F)$-completely passive states are only Gibbs ensembles.
	Although we have not succeeded identifying symmetry-protected completely passive states for general finite groups, we conjecture that they are also only Gibbs ensembles. 
	The difficulty of the proof for the general case comes from the fact that the generalization of Proposition~\ref{prop:prefiniteGEderivation} to general finite groups is quite nontrivial.
	If we can prove this proposition for general finite groups, we can prove that completely passive states protected by finite group symmetry are only Gibbs ensembles.
	
	For $n\geq3$, a dihedral group $D_n$ is defined as $D_n:=\{1, t, \cdots, t^{n-1}, r, rt, \cdots, rt^{n-1}\}$, where $1$ is the identity element, and $t$ and $r$ satisfy $t^n=r^2=1$ and $tr=rt^{-1}$. 
	Note that $D_n$ is a noncommutative finite group. 
	As a physical example, $t$ and $r$ can be regarded as the translation and inversion operators on a one-dimensional lattice with the periodic boundary condition.

\begin{theorem} \label{thm:finiteGFHCP}
	Let $G$ be a finite cyclic group or a dihedral group with a unitary representation $F$ acting on $\mathcal{H}$, $H\in\mathcal{B}^\mathrm{H}(\mathcal{H})$ commute with $F(G)$ and be not trivial, and $\rho\in\mathcal{B}^{++}(\mathcal{H})$ satisfy $\mathrm{tr}(\rho)=1$.
	Then, a state $\rho$ is $(G, F)$-completely passive w.r.t. a Hamiltonian $H$, if and only if $\rho$ can be written as $\rho=\frac{1}{Z}e^{-\beta H}$ with some $\beta\in[0, \infty)$, where $Z:=\mathrm{tr}(e^{-\beta H})$.
\end{theorem}

We also consider time-reversal symmetry as an example of symmetry described by an antiunitary operator, and identify symmetry-protected completely passive states.

\begin{theorem} \label{thm:THCP}
	Let $H\in\mathcal{B}^\mathrm{H}(\mathcal{H})$ commute with $\mathcal{T}$ and be not trivial, and $\rho\in\mathcal{B}^{++}(\mathcal{H})$ satisfy $\mathrm{tr}(\rho)=1$.
	Then, a state $\rho$ is $\mathcal{T}$-completely passive w.r.t. a Hamiltonian $H$, if and only if $\rho$ can be written as $\rho=\frac{1}{Z}e^{-\beta H}$ with some $\beta\in[0, \infty)$, where $Z:=\mathrm{tr}(e^{-\beta H})$. 
\end{theorem}

\section{Proofs of the main theorems}
\label{sec:main_theorems}

In this section, we prove the main theorems presented in Sec.~\ref{subsec:main_theorems}.

\subsection{Proof of Theorem~\ref{thm:GFHP}}

In order to prove Theorem~\ref{thm:GFHP}, we prove that $\mathcal{W}_{G, F, H}(\rho)$ is given by the sum of $\mathcal{W}_{H_\lambda}(\rho_\lambda)$ in Proposition~\ref{prop:ergotropy}.
The proof of this proposition shows that the $(G, F)$-respecting operator $U$ that extracts the maximal work $\mathcal{W}_{G, F, H}(\rho)$ from $\rho$ can be written as $U=\sum_{\lambda\in\Lambda_F} \iota_\lambda (I_{\mathcal{R}_\lambda}\otimes U_\lambda) \iota_\lambda^\dag$ with $U_\lambda\in\mathcal{U}(\mathcal{M}_\lambda)$ that extracts the maximal work $\mathcal{W}_{H_\lambda}(\rho_\lambda)$ from $\rho_\lambda$ under a Hamiltonian $H_\lambda$.

\begin{proposition} \label{prop:ergotropy}
	Let $G$ be a group with a unitary representation $F$ acting on $\mathcal{H}$, $H\in\mathcal{B}^\mathrm{H}(\mathcal{H})$ commute with $F(G)$ and $\rho\in\mathcal{B}^+(\mathcal{H})$.
	Then, the $(G, F)$-ergotropy of a state $\rho$ under a Hamiltonian $H$ is written as
	\begin{align}
		\mathcal{W}_{G, F, H}(\rho)=\sum_{\lambda\in\Lambda_F} \mathcal{W}_{H_\lambda}(\rho_\lambda),
	\end{align}
	where $\Lambda_F$, $H_\lambda$ and $\rho_\lambda$ are respectively defined in Eqs.~\eqref{eq:Hilbert_space_decomp1}, \eqref{eq:Hamiltoniandecomp} and \eqref{eq:rholambdadef}.
\end{proposition}	 

\begin{proof}
	For any $U\in\mathcal{U}_{G, F}(\mathcal{H})$, in the same way as $H$ in Eq.~\eqref{eq:Hamiltoniandecomp}, $U$ can be written as $U=\sum_{\lambda\in\Lambda_F} \iota_\lambda (I_{\mathcal{R}_\lambda}\otimes U_\lambda) \iota_\lambda^\dag$ with some $U_\lambda\in \mathcal{U}(\mathcal{M}_\lambda)$, where $\iota_\lambda$ is defined in Eq.~\eqref{eq:irreducibledecomposition}. 
	Then, we get
	\begin{align}
		W(\rho, H, U)&=\mathrm{tr}(\rho (H-U^\dag HU)) \nonumber\\
		&=\sum_{\lambda\in\Lambda_F} \mathrm{tr}(\rho \iota_\lambda [I_{\mathcal{R}_\lambda}\otimes (H_\lambda-U_\lambda^\dag H_\lambda U_\lambda)]\iota_\lambda^\dag) \nonumber\\
		&=\sum_{\lambda\in\Lambda_F} \mathrm{tr}(\iota_\lambda^\dag\rho \iota_\lambda [I_{\mathcal{R}_\lambda}\otimes (H_\lambda-U_\lambda^\dag H_\lambda U_\lambda)]) \nonumber\\
		&=\sum_{\lambda\in\Lambda_F} \mathrm{tr}_{\mathcal{M}_\lambda}(\rho_\lambda (H_\lambda-U_\lambda^\dagger H_\lambda U_\lambda)) \nonumber\\
		&=\sum_{\lambda\in\Lambda_F} W(\rho_\lambda,H_\lambda,U_\lambda) \nonumber\\
		&\leq \sum_{\lambda\in\Lambda_F} \mathcal{W}_{H_\lambda}(\rho_\lambda),
	\end{align}
	where the equality holds if and only if $U_\lambda$ satisfies $W(\rho_\lambda,H_\lambda,U_\lambda)=\mathcal{W}_{H_\lambda}(\rho_\lambda)$ for all $\lambda\in\Lambda_F$. 
	Therefore, $\mathcal{W}_{G, F, H}(\rho)=\sum_{\lambda\in\Lambda_F} \mathcal{W}_{H_\lambda}(\rho_\lambda)$.
\end{proof}

	Theorem~\ref{thm:GFHP} immediately follows from Proposition~\ref{prop:ergotropy}.
	The proof is as follows:\\\\
\textit{Proof of Theorem~\ref{thm:GFHP}.}
	By definition, $\rho$ is $(G, F)$-passive w.r.t. $H$, if and only if $\mathcal{W}_{G, F, H}(\rho)=0$. 
	From Proposition~\ref{prop:ergotropy}, this is equivalent to $\mathcal{W}_{H_\lambda}(\rho_\lambda)=0$ for all $\lambda\in\Lambda_F$. 
	This means that $\rho_\lambda$ is passive w.r.t. $H_\lambda$ for all $\lambda\in\Lambda_F$. \hspace{\fill} $\Box$\\

\subsection{Proof of Theorem~\ref{thm:THP}}

In order to prove Theorem~\ref{thm:THP}, we prove that $\mathcal{W}_{\mathcal{T}, H}(\rho)$ is the same as $\mathcal{W}_{\mathcal{T}, H}(\mathcal{S}_\mathcal{T}(\rho))$ in Proposition~\ref{prop:Twork}.

\begin{proposition} \label{prop:Twork}
	Let $H\in\mathcal{B}^\mathrm{H}(\mathcal{H})$ commute with $\mathcal{T}$ and $\rho\in\mathcal{B}^+(\mathcal{H})$.
	Then, $\mathcal{W}_{\mathcal{T}, H}(\rho)=\mathcal{W}_H(\mathcal{S}_\mathcal{T}(\rho))$, where $\mathcal{S}_\mathcal{T}$ is defined as Eq.~\eqref{eq:Tsymmetrizer}.
\end{proposition}

\begin{proof}
	Since for any $U\in\mathcal{U}_\mathcal{T}(\mathcal{H})$, 
	\begin{align}
		W(\mathcal{T}\rho \mathcal{T}^{-1}, H, U)=&\mathrm{tr}(\mathcal{T}\rho \mathcal{T}^{-1}(H-U^\dagger HU)) \nonumber\\
		=&\mathrm{tr}(\rho \mathcal{T}^{-1}(H-U^\dagger HU)\mathcal{T})^* \nonumber\\
		=&\mathrm{tr}(\rho(H-U^\dagger HU))^* \nonumber\\
		=&\mathrm{tr}(\rho(H-U^\dagger HU)) \nonumber\\
		=&W(\rho, H, U),
	\end{align}
	we get
	\begin{align}
		W(\rho, H, U)=\frac{1}{2}(W(\rho, H, U)+W(\mathcal{T}\rho \mathcal{T}^{-1}, H, U))=W(\mathcal{S}_\mathcal{T}(\rho), H, U)\leq\mathcal{W}_H(\mathcal{S}_\mathcal{T}(\rho)). \label{eq:Tergotropy_ineqality}
	\end{align}
	Since $H$ and $\mathcal{S}_\mathcal{T}(\rho)$ commute with $\mathcal{T}$, $H$ and $\mathcal{S}_\mathcal{T}(\rho)$ can be diagonalized as $H=\sum_{i=1}^d E_i\ket{\psi_i}\bra{\psi_i}$ and $\mathcal{S}_\mathcal{T}(\rho)=\sum_{i=1}^d p_i\ket{\phi_i}\bra{\phi_i}$ with $E_i, p_i\in\mathbb{R}$ and $\ket{\psi_i}, \ket{\phi_i}\in\mathcal{H}$ that satisfy $\mathcal{T}\ket{\psi_i}=\ket{\psi_i}, \mathcal{T}\ket{\phi_i}=\ket{\phi_i}$.
	Without loss of generality, we can suppose that $E_1\leq E_2\leq\cdots\leq E_d$ and $p_1\geq p_2\geq\cdots\geq p_d$.	
	Then, the equality in Eq.~\eqref{eq:Tergotropy_ineqality} holds, if $U=\sum_{i=1}^d \ket{\psi_i}\bra{\phi_i}$.
	Therefore, $\mathcal{W}_{\mathcal{T}, H}(\rho)=\mathcal{W}_H(\mathcal{S}_\mathcal{T}(\rho))$.
\end{proof}

Theorem~\ref{thm:THP} immediately follows from Proposition~\ref{prop:Twork}.
The proof is as follows:\\\\
\textit{Proof of Theorem~\ref{thm:THP}.}
	By definition, $\rho$ is $\mathcal{T}$-passive w.r.t. $H$, if and only if $\mathcal{W}_{\mathcal{T}, H}(\rho)=0$.
	From Proposition~\ref{prop:Twork}, this is equivalent to $\mathcal{W}_H(\mathcal{S}_\mathcal{T}(\rho))=0$.
	This means that $\mathcal{S}_\mathcal{T}(\rho)$ is passive w.r.t. $H$. \hspace{\fill} $\Box$\\

\subsection{Common part of the proofs of Theorems~\ref{thm:GFHCP}, \ref{thm:finiteGFHCP} and \ref{thm:THCP}}

In order to make clear the proofs for identifying completely passive states, we show the overall structure of the proofs in Fig.~\ref{fig:proof_structure}.
In Propositions~\ref{prop:workextop1} and \ref{prop:virtualtemperature}, we illustrate how to construct two kinds of unitary operators by which positive work is extracted from multiple copies of a state other than the GGE or the Gibbs ensemble.

Proposition~\ref{prop:workextop1} shows the construction of a unitary operator by which positive work is extracted from multiple copies of a state, when there exists an operator that commutes with the Hamiltonian but does not with the density operator of the state.

\begin{figure}
\begin{center}
\includegraphics[width=\columnwidth]{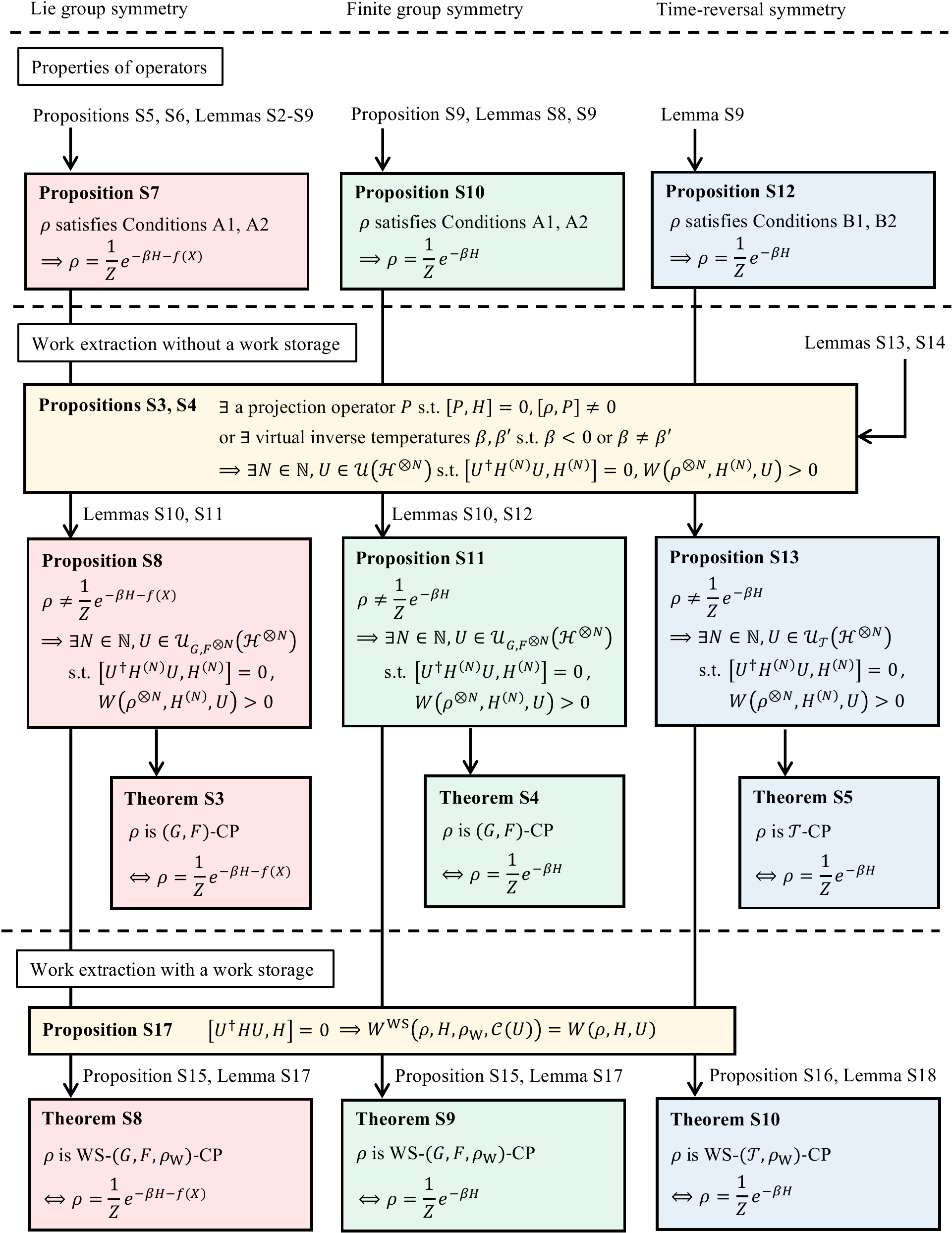}
\caption{Overall structure of the proofs for identifying completely passive states.
CP is an abbreviation for ``completely passive.''}
\label{fig:proof_structure}
\end{center}
\end{figure}

\begin{proposition} \label{prop:workextop1}
	Let $\rho\in\mathcal{B}^{++}(\mathcal{H})$, $M, L\in\mathbb{N}$, $P$ be the projection operator onto a linear subspace of $\mathcal{H}^{\otimes M}$ satisfying $[H^{(M)}, P]=0$, $T$ be the swapping operator on $\mathcal{H}^{\otimes M}\otimes\mathcal{H}^{\otimes M}$, and $\ket{\Psi_0}, \ket{\Psi_1}\in\mathcal{H}^{\otimes L}$ be  energy eigenstates with eigenvalues $\mathcal{E}_0, \mathcal{E}_1$ satisfying $\mathcal{E}_0<\mathcal{E}_1$.
	For $m\in\mathbb{N}$, we define an operator quartet $\{A_{ij}\}_{i, j\in\{0, 1\}}$ as
	\begin{align}
		A_{ij}:=\left\{\frac{1}{2}[I-(-1)^i T][P\otimes(I-P)][I-(-1)^j T]\right\}^{\otimes m}\otimes\ket{\Psi_i}\bra{\Psi_j},\ i,j\in\{0,1\},
	\end{align}
	and a mapping $\mathcal{O}$ from an operator quartet to an operator as
	\begin{align}
		\mathcal{O}(\{C_{ij}\}):=I-\sum_{i,j\in\{0,1\}} (-1)^{i-j} C_{ij}. \label{eq:unitary}
	\end{align}
	Then, $\mathcal{O}(\{A_{ij}\})$ is unitary.
	In addition, if $[\rho^{\otimes M}, P]\neq0$, then $W(\rho^{\otimes 2mM+L}, H^{(2mM+L)}, \mathcal{O}(\{A_{ij}\}))>0$ for some $m\in\mathbb{N}$.
\end{proposition}

\begin{proof}
	First, we prove that $\mathcal{O}(\{A_{ij}\})$ is unitary.
	From Lemma~\ref{lem:U1extractedwork}, it is sufficient to prove that $\{A_{ij}\}$ satisfy 
	\begin{align}
		 A_{ij}^\dagger=A_{ji},\ A_{ij}A_{kl}=\delta_{jk}A_{il} \label{eq:Aijcond}
	\end{align}
	for all $i, j, k, l\in\{0, 1\}$.
	For $i,j\in\{0,1\}$, we define $R_{ij}\in\mathcal{B}(\mathcal{H}^{\otimes 2M})$ as
	\begin{align}
		R_{ij}:=\frac{1}{2}[I-(-1)^i T][P\otimes(I-P)][I-(-1)^j T].
	\end{align}
	Then, $A_{ij}=R_{ij}^{\otimes m}\otimes\ket{\Psi_i}\bra{\Psi_j}$.
	Since $R_{ij}^\dag=R_{ji}$ and $(\ket{\Psi_i}\bra{\Psi_j})^\dag=\ket{\Psi_j}\bra{\Psi_i}$,  
	\begin{align}
		A_{ij}^\dag=(R_{ij}^\dag)^{\otimes m}\otimes(\ket{\Psi_i}\bra{\Psi_j})^\dag=R_{ji}^{\otimes m}\otimes\ket{\Psi_j}\bra{\Psi_i}=A_{ji}.
	\end{align}
	Since $\frac{1}{2}[I-(-1)^i T]$ is the antisymmetrizer for $i=0$ and the symmetrizer for $i=1$,  
	\begin{align}
		\frac{1}{2}[I-(-1)^i T]\cdot\frac{1}{2}[I-(-1)^j T]=\delta_{ij}\frac{1}{2}[I-(-1)^i T]. \label{eq:symmetrizerprojection}
	\end{align}
	From a property of the swapping operator $T$, 
	\begin{align}
		[P\otimes(I-P)]T[P\otimes(I-P)]=&[P\otimes(I-P)][(I-P)\otimes P]T=0. \label{eq:swappingprojection}
	\end{align}
	From Eqs.~\eqref{eq:symmetrizerprojection} and \eqref{eq:swappingprojection}, for $i, j, k, l\in\{0, 1\}$, 
	\begin{align}
		&R_{ij}R_{kl} \nonumber\\
		=&\frac{1}{2}[I-(-1)^i T][P\otimes(I-P)][I-(-1)^j T]\frac{1}{2}[I-(-1)^k T][P\otimes(I-P)][I-(-1)^l T] \nonumber\\
		=&\frac{1}{2}[I-(-1)^i T][P\otimes(I-P)]\delta_{jk}[I-(-1)^j T][P\otimes(I-P)][I-(-1)^l T] \nonumber\\
		=&\delta_{jk}\frac{1}{2}[I-(-1)^i T][P\otimes(I-P)][I-(-1)^l T] \nonumber\\
		=&\delta_{jk}R_{il}.
	\end{align}
	Since $(\ket{\Psi_i}\bra{\Psi_j})(\ket{\Psi_k}\bra{\Psi_l})=\delta_{jk}\ket{\Psi_i}\bra{\Psi_l}$, we obtain
	\begin{align}
		A_{ij}A_{kl}=&(R_{ij}^{\otimes m}\otimes\ket{\Psi_i}\bra{\Psi_j})(R_{kl}^{\otimes m}\otimes\ket{\Psi_k}\bra{\Psi_l}) \nonumber\\
		=&(R_{ij}R_{kl})^{\otimes m}\otimes(\ket{\Psi_i}\bra{\Psi_j})(\ket{\Psi_k}\bra{\Psi_l}) \nonumber\\
		=&\delta_{jk}(R_{il}^{\otimes m}\otimes\ket{\Psi_i}\bra{\Psi_l}) \nonumber\\
		=&\delta_{jk}A_{il}.
	\end{align}
	Therefore, $\{A_{ij}\}$ satisfies Eq.~\eqref{eq:Aijcond}.
	From Lemma~\ref{lem:U1extractedwork}, $\mathcal{O}(\{A_{ij}\})$ is unitary.
	
	Next, we calculate the extracted work from $\rho^{\otimes 2mM+L}$ by $\mathcal{O}(\{A_{ij}\})$.
	We define $\Delta\mathcal{E}:=\mathcal{E}_1-\mathcal{E}_0(>0)$.
	Since $[H^{(M)}, P]=0$ and $[H^{(2M)}, T]=0$, we obtain $[H^{(2M)}, R_{ij}]=0$ and thus
	\begin{align}
		[H^{(2mM+L)}, A_{ij}]=&[H^{(2mM)}\otimes I+I\otimes H^{(L)}, R_{ij}^{\otimes m}\otimes\ket{\Psi_i}\bra{\Psi_j}] \nonumber\\
		=&[H^{(2mM)}, R_{ij}^{\otimes m}]\otimes\ket{\Psi_i}\bra{\Psi_j}+R_{ij}^{\otimes m}\otimes[H^{(L)}, \ket{\Psi_i}\bra{\Psi_j}] \nonumber\\
		=&R_{ij}^{\otimes m}\otimes\Delta\mathcal{E}(i-j)\ket{\Psi_i}\bra{\Psi_j} \nonumber\\
		=&\Delta\mathcal{E}(i-j)A_{ij}. \label{eq:HAcommutation}
	\end{align}
	From Lemma~\ref{lem:U1extractedwork},
	\begin{align}
		&W(\rho^{\otimes 2mM+L}, H^{(2mM+L)}, \mathcal{O}(\{A_{ij}\})) \nonumber\\
		=&\Delta\mathcal{E}(\mathrm{tr}(\rho^{\otimes 2mM+L}A_{11})-\mathrm{tr}(\rho^{\otimes 2mM+L}A_{00})) \nonumber\\
		=&\Delta\mathcal{E}\left[(\mathrm{tr}(\rho^{\otimes 2M} R_{11}))^m\braket{\Psi_1|\rho^{\otimes L}|\Psi_1}-(\mathrm{tr}(\rho^{\otimes 2M} R_{00}))^m\braket{\Psi_0|\rho^{\otimes L}|\Psi_0}\right]. \label{eq:extractedwork2}
	\end{align}
	Suppose that $[\rho^{\otimes M}, P]\neq0$.
	Then, $P\neq0$ or $I$, and thus we can take states $\ket{a}\in\mathrm{supp}(P)$ and $\ket{b}\in\mathrm{supp}(I-P)$.
	From the positive definiteness of $\rho$, 
	\begin{align}
		\mathrm{tr}(\rho^{\otimes 2M}R_{ii})=&\frac{1}{2}\mathrm{tr}([P\otimes(I-P)](I-(-1)^i T)\rho^{\otimes 2M}(I-(-1)^i T)[P\otimes(I-P)]) \nonumber\\
		\geq&\frac{1}{2}\bra{a}\bra{b}[P\otimes(I-P)](I-(-1)^i T)\rho^{\otimes 2M}(I-(-1)^i T)[P\otimes(I-P)]\ket{a}\ket{b} \nonumber\\
		=&\frac{1}{2}[\bra{a}\bra{b}-(-1)^i \bra{b}\bra{a}]\rho^{\otimes 2M}[\ket{a}\ket{b}-(-1)^i \ket{b}\ket{a}] \nonumber\\
		>&0.
	\end{align}
	This allows us to transform Eq.~\eqref{eq:extractedwork2} into the following: 
	\begin{align}
		&W(\rho^{\otimes 2mM+L}, H^{(2mM+L)}, \mathcal{O}(\{A_{ij}\})) \nonumber\\
		=&\Delta\mathcal{E}(\mathrm{tr}(\rho^{\otimes 2M}R_{11}))^m\braket{\Psi_1|\rho^{\otimes L}|\Psi_1}\left[1-\left(\frac{\mathrm{tr}(\rho^{\otimes 2M}R_{00})}{\mathrm{tr}(\rho^{\otimes 2M}R_{11})}\right)^m \frac{\braket{\Psi_0|\rho^{\otimes L}|\Psi_0}}{\braket{\Psi_1|\rho^{\otimes L}|\Psi_1}}\right].
	\end{align}
	Since $[\rho^{\otimes M}, P]\neq0$,  
	\begin{align}
		\mathrm{tr}(\rho^{\otimes 2M}R_{11})-\mathrm{tr}(\rho^{\otimes 2M}R_{00})=&\mathrm{tr}(\rho^{\otimes 2M}\{T[P\otimes(I-P)]+[P\otimes(I-P)]T\}) \nonumber\\
		=&2\mathrm{tr}(T\rho^{\otimes 2M}[P\otimes(I-P)]) \nonumber\\
		=&2\mathrm{tr}(T[P \rho^{\otimes M} \otimes (I-P) \rho^{\otimes M}]) \nonumber\\
		=&2\mathrm{tr}(P\rho^{\otimes M}(I-P)\rho^{\otimes M}) \nonumber\\
		=&\mathrm{tr}([\rho^{\otimes M}, P]^\dag [\rho^{\otimes M}, P]) \nonumber\\
		=&\|[\rho^{\otimes M}, P]\|_\mathrm{HS}^2 \nonumber\\
		>&0,
	\end{align}
	where $\|\cdot\|_\mathrm{HS}$ is the Hilbert-Schmidt norm, and we used the fact that $T$ satisfies $\mathrm{tr}[T(\Omega_1\otimes\Omega_2)]=\mathrm{tr}(\Omega_1\Omega_2)$ for $\Omega_1, \Omega_2\in\mathcal{B}(\mathcal{H}^{\otimes M})$.
	Therefore, $0<\frac{\mathrm{tr}(\rho^{\otimes 2M}R_{00})}{\mathrm{tr}(\rho^{\otimes 2M}R_{11})}<1$, and thus for sufficiently large $m\in\mathbb{N}$, $W(\rho^{\otimes 2mM+L}, H^{(2mM+L)}, \mathcal{O}(\{A_{ij}\}))>0$.
\end{proof}	

From Proposition~\ref{prop:workextop1}, we know that completely passive states commute with the Hamiltonian, i.e., do not have any coherence in the energy eigenbasis.
The notion of virtual temperature introduced in Ref.~\cite{Skrzypczyk2015} is useful for the investigation of work extraction from incoherent states.  
When the state has a negative virtual inverse temperature or a pair of different virtual inverse temperatures, Proposition~\ref{prop:virtualtemperature} shows the construction of a unitary operator by which positive work is extracted from multiple copies of a state.

\begin{proposition} \label{prop:virtualtemperature}
	Let $H\in\mathcal{B}^\mathrm{H}(\mathcal{H})$, $\rho\in\mathcal{B}^{++}(\mathcal{H})$, $L, L'\in\mathbb{N}$, $\ket{\Psi_0}, \ket{\Psi_1}\in\mathcal{H}^{\otimes L}$ be eigenstates of $H^{(L)}$ with different eigenvalues, and $\ket{\Psi'_0}, \ket{\Psi'_1}\in\mathcal{H}^{\otimes L'}$ be eigenstates of $H^{(L')}$.
	For $i=0, 1$, we define
	\begin{align}
		&\mathcal{E}_i:=\braket{\Psi_i|H^{(L)}|\Psi_i},\ \Delta\mathcal{E}:=\mathcal{E}_1-\mathcal{E}_0,\ S_i:=-\log(\bra{\Psi_i}\rho^{\otimes L}\ket{\Psi_i}),\ \Delta S:=S_1-S_0,\ \beta:=\frac{\Delta S}{\Delta\mathcal{E}}, \nonumber\\
		&\mathcal{E}'_i:=\braket{\Psi'_i|H^{(L')}|\Psi'_i},\ \Delta\mathcal{E}':=\mathcal{E}'_1-\mathcal{E}'_0,\ S'_i:=-\log(\bra{\Psi'_i}\rho^{\otimes L'}\ket{\Psi'_i}),\ \Delta S':=S'_1-S'_0. \label{eq:ESdef}
	\end{align}
	Let $\Delta\mathcal{E}$ satisfy $\Delta\mathcal{E}>0$. 
	Let $\Delta\mathcal{E}'$ and $\Delta S'$ satisfy (I) $\Delta\mathcal{E}'>0$ or (II) $\Delta\mathcal{E}'=0$ and $\Delta S'\geq 0$. 
	We define $\mathcal{O}$ as Eq.~\eqref{eq:unitary} and an operator quartet $\{B_{ij}\}_{i,j\in\{0, 1\}}$ as
	\begin{align}
		B_{ij}:=(\ket{\Psi_i}\bra{\Psi_j})^{\otimes m}\otimes(\ket{\Psi'_{1-i}}\bra{\Psi'_{1-j}})^{\otimes m'}.
	\end{align} 
	Then, $\mathcal{O}(\{B_{ij}\})$ is unitary.
	In addition, if $\beta<0$ or $\Delta S'-\beta\Delta\mathcal{E}'\neq0$, then $W(\rho^{\otimes mL+m'L'}, H^{(mL+m'L')}, \mathcal{O}(\{B_{ij}\}))>0$ for some $m, m'\in\mathbb{N}$.
\end{proposition}

\begin{proof}
	Since $\{B_{ij}\}$ satisfies $B_{ij}^\dag=B_{ji}$ and $B_{ij}B_{kl}=\delta_{jk}B_{il}$ for $i, j, k, l\in\{0, 1\}$,  $\mathcal{O}(\{B_{ij}\})$ is unitary from Lemma~\ref{lem:U1extractedwork}.
	Since $\{B_{ij}\}$ also satisfies 
	\begin{align}
		[H^{(mL+m'L')}, B_{ij}]=&(m\mathcal{E}_i+m'\mathcal{E}'_{1-i})B_{ij}-(m\mathcal{E}_j+m'\mathcal{E}'_{1-j})B_{ij} \nonumber\\
		=&[m(\mathcal{E}_i-\mathcal{E}_j)+m'(\mathcal{E}'_{1-i}-\mathcal{E}'_{1-j})]B_{ij} \nonumber\\
		=&\{m\Delta\mathcal{E}(i-j)+m'\Delta\mathcal{E}'[(1-i)-(1-j)]\}B_{ij} \nonumber\\
		=&(m\Delta\mathcal{E}-m'\Delta\mathcal{E}')(i-j)B_{ij}, \label{eq:HBcommutation}
	\end{align}
	 the extracted work from $\rho^{\otimes mL+m'L'}$ by $\mathcal{O}(\{B_{ij}\}$ is given by, from Lemma~\ref{lem:U1extractedwork},
	\begin{align}
		&W(\rho^{\otimes mL+m'L'}, H^{(mL+m'L')}, \mathcal{O}(\{B_{ij}\})) \nonumber\\
		=&(m\Delta\mathcal{E}-m'\Delta\mathcal{E}')(\mathrm{tr}(\rho^{\otimes mL+m'L'}B_{11})-\mathrm{tr}(\rho^{\otimes mL+m'L'}B_{00})) \nonumber\\
		=&(m\Delta\mathcal{E}-m'\Delta\mathcal{E}')\left(\braket{\Psi_1|\rho^{\otimes L}|\Psi_1}^m\braket{\Psi'_0|\rho^{\otimes L'}|\Psi'_0}^{m'}-\braket{\Psi_0|\rho^{\otimes L}|\Psi_0}^m\braket{\Psi'_1|\rho^{\otimes L'}|\Psi'_1}^{m'}\right) \nonumber\\
		=&(m\Delta\mathcal{E}-m'\Delta\mathcal{E}')\left(e^{-m S_1-m' S'_0}-e^{-m S_0-m' S'_1}\right) \nonumber\\
		=&-(m\Delta\mathcal{E}-m'\Delta\mathcal{E}')\left(e^{m\Delta S-m' \Delta S'}-1\right)e^{-m S_1-m' S'_0}.
	\end{align}
	The condition $\beta<0$ or $\Delta S'-\beta\Delta\mathcal{E}'\neq0$ can be classified into the following 4 cases:
	\begin{align}
		&\mathrm{(i)}\ \beta<0, \\
		&\mathrm{(ii)}\ \beta\geq0,\ \Delta S'\neq\beta\Delta\mathcal{E}',\ \Delta S'<0, \\
		&\mathrm{(iii)}\ \beta=0,\ \Delta S'\neq\beta\Delta\mathcal{E}',\ \Delta S'\geq0, \\
		&\mathrm{(iv)}\ \beta>0,\ \Delta S'\neq\beta\Delta\mathcal{E}',\ \Delta S'\geq0.
	\end{align}
	In the case (i), since $\Delta\mathcal{E}>0$ and $\Delta S<0$, we obtain $W(\rho^{\otimes mL+m'L'}, H^{(mL+m'L')}, \mathcal{O}(\{B_{ij}\}))>0$ for $m=1$ and $m'=0$.
	In the case (ii), since $\Delta\mathcal{E}'>0$ and $\Delta S'<0$, we obtain $W(\rho^{\otimes mL+m'L'}, H^{(mL+m'L')}, \mathcal{O}(\{B_{ij}\}))>0$ for $m=0$ and $m'=1$.
	In the case (iii), since $\Delta\mathcal{E}>0$, $\Delta\mathcal{E}'\geq0$, $\Delta S=0$ and $\Delta S'>0$, we obtain $W(\rho^{\otimes mL+m'L'}, H^{(mL+m'L')}, \mathcal{O}(\{B_{ij}\}))>0$ for $m, m'\in\mathbb{N}$ satisfying $m'>0$ and $\frac{m}{m'}>\frac{\Delta\mathcal{E}'}{\Delta\mathcal{E}}$.
	In the case (iv), since $\Delta\mathcal{E}>0$, $\Delta\mathcal{E}'\geq0$, $\Delta S>0$ and $\Delta S'\geq0$, we get $\frac{\Delta\mathcal{E}'}{\Delta\mathcal{E}}\geq0$ and $\frac{\Delta S'}{\Delta S}\geq0$.
	Since $\Delta S'\neq\beta\Delta\mathcal{E}'$, we obtain $\frac{\Delta S'}{\Delta S}=\frac{\Delta S'}{\beta\Delta\mathcal{E}}\neq\frac{\beta\Delta\mathcal{E}'}{\beta\Delta\mathcal{E}}=\frac{\Delta\mathcal{E}'}{\Delta\mathcal{E}}$.
	Thus we can take $m, m'\in\mathbb{N}$ such that $m'>0$ and $\frac{m}{m'}$ is between $\frac{\Delta S'}{\Delta S}$ and $\frac{\Delta\mathcal{E}'}{\Delta\mathcal{E}}$.
	For such $m$ and $m'$, $W(\rho^{\otimes mL+m'L'}, H^{(mL+m'L')}, \mathcal{O}(\{B_{ij}\}))>0$.
	Therefore, in all cases satisfying $\beta<0$ or $\Delta S'-\beta\Delta\mathcal{E}'\neq0$, there exist $m, m'\in\mathbb{N}$ such that $W(\rho^{\otimes mL+m'L'}, H^{(mL+m'L')}, \mathcal{O}(\{B_{ij}\}))>0$.
\end{proof}

	Finally, we introduce a state $\ket{\Phi(\bm{n})}$ that is useful in the proofs of Theorems~\ref{thm:GFHCP}~and~\ref{thm:finiteGFHCP}. 
	For a Hilbert space $\mathcal{R}$ with dimension $r$, we define a totally antisymmetric mapping $\mathcal{A}_\mathcal{R}: \mathcal{R}^r\to\mathcal{R}^{\otimes r}$ as  
	\begin{align}
		\mathcal{A}_\mathcal{R}(\ket{\phi_1}, \cdots, \ket{\phi_r}):=\frac{1}{\sqrt{r!}}\sum_{\sigma\in S_r}\mathrm{sgn}(\sigma)\ket{\phi_{\sigma(1)}}\otimes\cdots\otimes\ket{\phi_{\sigma(r)}}, \label{eq:defdetstate}
	\end{align}
	where $S_r$ is the symmetric group of degree $r$.
	We take an arbitrary orthonormal basis $\{\ket{\phi_{\lambda i}}\}_{i=1}^{r_\lambda}$ of $\mathcal{R}_\lambda$ and define $D:=\prod_{\lambda\in\Lambda_F} r_\lambda$ and $K:=\{(\lambda, j)\ |\ \lambda\in\Lambda_F, j=1, \cdots, m_\lambda\}$.
	We suppose that for any $\lambda\in\Lambda_F$, $H_\lambda$ can be diagonalized as $H_\lambda=\sum_{j=1}^{m_\lambda}E_{\lambda j}\ket{\psi_{\lambda j}}\bra{\psi_{\lambda j}}$ with $\ket{\psi_{\lambda j}}\in\mathcal{M}_\lambda$ and $E_{\lambda j}\in\mathbb{R}$.
	Therefore, from Eq.~\eqref{eq:Hamiltoniandecomp}, $H$ can be written as 
\begin{align}
	H=\sum_{\lambda\in\Lambda_F} \iota_\lambda \left(I_{\mathcal{R}_\lambda}\otimes \sum_{j=1}^{m_\lambda}E_{\lambda j}\ket{\psi_{\lambda j}}\bra{\psi_{\lambda j}}\right)\iota_\lambda^\dag. \label{eq:Hamiltoniandecomp2}
\end{align}  
	We take an arbitrary orthonormal basis $\{\ket{\phi_{\lambda i}}\}_{i=1}^{r_\lambda}$ of $\mathcal{R}_\lambda$ and define $\ket{\chi_k}\in\mathcal{H}^{\otimes D}$ for $k=(\lambda, j)\in K$ as 
	\begin{align}
		\ket{\chi_k}:=\left[\iota_\lambda^{\otimes r_\lambda} (\mathcal{A}_{\mathcal{R}_\lambda}(\{\ket{\phi_{\lambda i}}\})\otimes\ket{\psi_{\lambda j}}^{\otimes r_\lambda})\right]^{\otimes \frac{D}{r_\lambda}}. \label{eq:chi_k_definition}
	\end{align}
	From Lemma~\ref{lem:determinantstate}, $\ket{\chi_k}$ is independent of the choice of the basis $\{\ket{\phi_{\lambda i}}\}_{i=1}^{r_\lambda}$ up to phase factor.
	For any $\bm{n}=(n_k)_{k\in K}\in\mathbb{N}^{K}$, we define $\ket{\Phi(\bm{n})}\in\mathcal{H}^{\otimes D\bm{n}\cdot\bm{w}_0}$ as 
	\begin{align}
		\ket{\Phi(\bm{n})}:=\bigotimes_{k\in K} \ket{\chi_k}^{\otimes n_k}, \label{eq:determinantstate}
	\end{align}
	where $\bm{w}_0:=(1)_{k\in K}$ and $\bm{a}\cdot\bm{b}:=\sum_{k\in K} a_kb_k$ for $\bm{a}=(a_k)_{k\in K}, \bm{b}=(b_k)_{k\in K}\in\mathbb{R}^K$.
	Then, from Lemma~\ref{lem:eigenvectorcomposition}, $\ket{\Phi(\bm{n})}$ is a simultaneous eigenstate of $H^{(D\bm{n}\cdot\bm{w}_0)}$ and $F^{\otimes D\bm{n}\cdot\bm{w}_0}(G)$.

\subsection{Proof of Theorem~\ref{thm:GFHCP}}
	In this subsection, we suppose that $G$ is a connected compact matrix Lie group with a faithful unitary representation $F$ acting on $\mathcal{H}$, and $\mathfrak{g}$ is the Lie algebra of $G$ with the associated Hermitian representation $f$. 
	Since $F$ is faithful, $f$ is also faithful.

\subsubsection{Preliminaries on Lie groups} \label{subsubsec:Lie_groups}

The Levi decomposition states that a connected compact Lie group $G$ is the commuting product of $G=NS$ (Theorem~4.29~of~\cite{Knapp2002}), where $N$ is the identity component of the center of $G$, and $S$ is the connected compact subgroup of $G$ with semisimple Lie algebra $\mathrm{i}[\mathfrak{g}, \mathfrak{g}]:=\mathrm{span}(\{\mathrm{i}[X, Y]\ |\ X, Y\in\mathfrak{g}\})$. 
For a connected compact matrix Lie group, the exponential mapping from its Lie algebra to the Lie group is surjective (Corollary~4.48~of~\cite{Knapp2002}).
Thus $e^{\mathrm{i}\mathfrak{n}}=N$ and $e^{\mathrm{i}\mathfrak{s}}=S$, where $\mathfrak{n}$ and $\mathfrak{s}$ are respectively the Lie algebras of $N$ and $S$. 
This implies that 
\begin{align}
	G=NS=e^{\mathrm{i}\mathfrak{n}}e^{\mathrm{i}\mathfrak{s}}. \label{eq:Levi_decomposition}
\end{align}
The representation of a Lie group and its associated representation of the Lie algebra satisfy $F(e^{\mathrm{i}(X)})=e^{\mathrm{i}f(X)}$ for all $X\in\mathfrak{g}$.
Therefore, the commutativity with $F(G)$ is equivalent to the commutativity with $f(\mathfrak{n})$ and $f(\mathfrak{s})$.
We will discuss the properties of $\mathfrak{n}$ and $\mathfrak{s}$.

\paragraph{Connected compact Abelian matrix Lie groups.}

First, we introduce a useful basis of $\mathfrak{n}$ based on Lemma~\ref{lem:Liealgebrabasis}. 
In order to prove Lemma~\ref{lem:Liealgebrabasis}, we make use of the following fact: 
If $N\subset GL(n, \mathbb{C})$ is a connected compact Abelian matrix Lie group, then $N$ is isomorphic to a torus $T^a$ of degree $a$ for some $a\in\mathbb{N}$ (Corollary~1.103~of~\cite{Knapp2002}).
$T^a$ is defined as $T^a:=\{\sum_{i=1}^a e^{\mathrm{i}\theta_i}D_i\ |\ \theta_1, \cdots, \theta_a\in\mathbb{R}\}$, where $D_i$ is the $a$ by $a$ matrix such that only the $(i, i)$ element is 1 and all the other elements are 0. 
$\{D_i\}_{i=1}^a$ is a basis of the Lie algebra $\mathfrak{t}^a$ of $T^a$.

\begin{lemma} \label{lem:Liealgebrabasis} 
	Let $G$ be a connected compact Lie group with a unitary representation $F$ acting on $\mathcal{H}$, $f$ be the associated representation of the Lie algebra of $G$, $N$ be the connected center of $G$, $\mathfrak{n}$  be the Lie algebra of $N$, and $F$ be irreducibly decomposed as Eq.~\eqref{eq:irreducibledecomposition}. 
	Then, there exists a basis $\{W_i\}$ of $\mathfrak{n}$ such that for any $i=1, \cdots, a$, $f(W_i)$ can be written as 
	\begin{align}
		f(W_i)=\sum_{\lambda\in\Lambda_F} w_{i\lambda}\iota_\lambda\iota_\lambda^\dag,
	\end{align}  
	with $w_{i\lambda}\in\mathbb{Z}$.
\end{lemma}

\begin{proof}
	Since $N$ is a connected compact Abelian matrix Lie group, we can take an isomorphism $\Phi: T^a\to N$.
	We denote the derivative of $\Phi$ at the identity by $\phi: \mathfrak{t}^a\to \mathfrak{n}$. 
	Then $\phi$ is a Lie algebra isomorphism.
	For any $i=1, \cdots, a$, we define $W_i:=\phi(D_i)$.
	Let $f_\lambda$ be the representation associated with $F_\lambda$.
	Then, in the same way as $F$ in Eq.~\eqref{eq:irreducibledecomposition}, $f$ can be written as 
	\begin{align}
		f=\sum_{\lambda\in\Lambda_F} \iota_\lambda(f_\lambda\otimes I_{\mathcal{M}_\lambda})\iota_\lambda^\dag. \label{eq:f_decomposition}
	\end{align}
	Since $W_i\in\mathfrak{n}$, $f_\lambda(W_i)$ commutes with $F_\lambda(G)$.
	From Schur's lemma, $f_\lambda(W_i)$ can be written as 
	\begin{align}
	f_\lambda(W_i)=w_{i\lambda}I_{\mathcal{R}_\lambda} \label{eq:W_i_representation}
	\end{align}
	with some $w_{i\lambda}\in\mathbb{R}$.
	From Eqs.~\eqref{eq:f_decomposition} and \eqref{eq:W_i_representation},
	\begin{align}
		f(W_i)
		=\sum_{\lambda\in\Lambda_F} \iota_\lambda(w_{i\lambda}I_{\mathcal{R}_\lambda}\otimes I_{\mathcal{M}_\lambda})\iota_\lambda^\dag
		=\sum_{\lambda\in\Lambda_F} w_{i\lambda}\iota_\lambda\iota_\lambda^\dag.
	\end{align}
	Therefore, we get
	\begin{align} 
		\sum_{\lambda\in\Lambda_F} e^{\mathrm{i}2\pi w_{i\lambda}}\iota_\lambda\iota_\lambda^\dag
		=&e^{\mathrm{i}2\pi f(W_i)} \nonumber\\
		=&e^{\mathrm{i}2\pi f(\phi(D_i))} \nonumber\\
		=&F\left(e^{\mathrm{i}2\pi \phi(D_i)}\right) \nonumber\\
		=&F\left(\Phi\left(e^{\mathrm{i}2\pi D_i}\right)\right) \nonumber\\
		=&F\left(\Phi(I)\right) \nonumber\\
		=&I \nonumber\\
		=&\sum_{\lambda\in\Lambda_F} \iota_\lambda\iota_\lambda^\dag.
	\end{align}
	By comparing the both sides of this equasion, we get $e^{\mathrm{i}2\pi w_{i\lambda}}=1$. 
	This implies that $w_{i\lambda}\in\mathbb{Z}$.
\end{proof}

\paragraph{Semisimple Lie algebras.}
	Next, we introduce properties of semisimple Lie algebra $\mathfrak{s}$ and define an operator $C\in\mathcal{B}(\mathcal{H}^{\otimes 2})$ that commutes with $f^{(2)}(\mathfrak{s})$.
	We define a symmetric bilinear form $B_f$ on $\mathfrak{s}$ as 
	\begin{align}
		B_f(V, V'):=\mathrm{tr}(f(V)f(V')),\ V, V'\in\mathfrak{s}. \label{eq:bilineardef}
	\end{align}
	If we take arbitrary $V\in\mathfrak{s}$ that satisfies $B_f(V, V')=0$ for all $V'\in\mathfrak{s}$, then we obtain $B_f(V, V)=0$ and thus $f(V)=0$.
	From the faithfulness of $f$, we get $V=0$.
	This implies that $B_f$ is nondegenerate.
	Therefore, for any basis $\{V_i\}_{i=1}^b$ of $\mathfrak{s}$, there exists the $B_f$-dual basis $\{V^i\}_{i=1}^b$, which satisfies $B_f(V_i,V^j)=\delta_{ij}$ for all $i, j=1, \cdots, b$.

\begin{lemma} \label{lem:casimir}
	Let $\mathfrak{s}$ be a semisimple Lie algebra with a faithful Hermitian representation $f$  acting on $\mathcal{H}$, $\{V_i\}_{i=1}^b$ be a basis of $\mathfrak{s}$, and $\{V^i\}_{i=1}^b$ be the $B_f$-dual basis.
	We define an operator $C\in\mathcal{B}(\mathcal{H}^{\otimes 2})$ as
	\begin{align}
		C:=\sum_{i=1}^b f(V_i)\otimes f(V^i). \label{eq:casimirdefinition}
	\end{align}
	Then, $C$ is independent of the choice of basis of $\mathfrak{s}$ and commutes with $f^{(2)}(\mathfrak{s})$. 
\end{lemma}

\begin{proof}
	The Casimir operator of the representation $f$ is defined as
	\begin{align}
		C_f:=\sum_{i=1}^b f(V_i)f(V^i).
	\end{align}
	From Lemma~3.3.7~of~\cite{Goodman2009}, $C_f$ is independent of the choice of the basis and commutes with $f(\mathfrak{s})$ .
	Since
	\begin{align}
		\sum_{i=1}^b f(V_i)\otimes f(V^i)=&\sum_{i=1}^b f\left(\sum_{j=1}^b B_f(V_j, V_i)V^j\right)\otimes f(V^i) \nonumber\\
		=&\sum_{j=1}^b f(V^j)\otimes f\left(\sum_{i=1}^b B_f(V_i, V_j)V^i\right) \nonumber\\
		=&\sum_{j=1}^b f(V^j)\otimes f(V_j),
	\end{align}
	$C$ can be written as $C=\frac{1}{2}(C_{f^{(2)}}-C_f^{(2)})$ with Casimir operators $C_f$ and $C_{f^{(2)}}$.
	Therefore, $C$ is independent of the choice of the basis and commutes with $f^{(2)}(\mathfrak{s})$.
\end{proof}

Proposition~\ref{prop:Casimir} shows the condition of $\xi\in\mathcal{B}^\mathrm{H}(\mathcal{H})$ such that $\xi^{(2)}$ commutes with $C$.
Along with  Proposition~\ref{prop:workextop1}, this fact plays an important role in the proof of Proposition~\ref{prop:preGGEderivation}.

\begin{proposition} \label{prop:Casimir}
	Let $\mathfrak{s}$ be a semisimple Lie algebra with a faithful Hermitian representation $f$   acting on $\mathcal{H}$, and $\xi\in\mathcal{B}^\mathrm{H}(\mathcal{H})$ satisfy $[\xi^{(2)}, C]=0$, where $C$ is defined by Eq.~\eqref{eq:casimirdefinition}.
	Then, there exists $X^\mathrm{S}\in\mathfrak{s}$ such that $\xi-f(X^\mathrm{S})$ commutes with $f(\mathfrak{s})$.
\end{proposition} 

\begin{proof}
	We introduce the Hilbert-Schmidt inner product $(A, B)_\mathrm{HS}:=\mathrm{tr}(A^\dag B)$ on $\mathcal{B}(\mathcal{H})$ and define $\mathcal{P}$ as the projection from $\mathcal{B}(\mathcal{H})$ onto $f(\mathfrak{s})$.
	We take $X^\mathrm{S}\in\mathfrak{s}$ satisfying $\mathcal{P}(\xi)=f(X^\mathrm{S})$ and define $\eta:=\xi-f(X^\mathrm{S})$.
	Then, $\mathrm{tr}(\eta f(V))=0$ for all $V\in\mathfrak{s}$. 
	Therefore, for any $V\in\mathfrak{s}$, 
	\begin{align}
		\mathrm{tr}_{\mathcal{H}_1}([\eta^{(2)}, C](f(V)\otimes I))=&\sum_{i=1}^b \mathrm{tr}_{\mathcal{H}_1}([\eta\otimes I+I\otimes\eta, f(V_i)\otimes f(V^i)](f(V)\otimes I)) \nonumber\\
		=&\sum_{i=1}^b (\mathrm{tr}([\eta, f(V_i)]f(V))f(V^i)+\mathrm{tr}(f(V_i)f(V))[\eta, f(V^i)]) \nonumber\\
		=&\sum_{i=1}^b \mathrm{tr}(\eta f([V_i, V]))f(V^i)+\left[\eta, f\left(\sum_{i=1}^b B_f(V_i, V)V^i\right)\right] \nonumber\\
		=&[\eta, f(V)], \label{eq:casimirpartialtr}
	\end{align}
	where $\mathcal{H}_1$ is the Hilbert space of the first copy of the system. 
	On the other hand, from the commutativity of $\xi^{(2)}$ with $C$ and Lemma~\ref{lem:casimir}, $[\eta^{(2)}, C]=[\xi^{(2)}, C]-[f^{(2)}(X^\mathrm{S}), C]=0$.
	By substituting this into Eq.~\eqref{eq:casimirpartialtr}, we get $[\xi-f(X^\mathrm{S}), f(V)]=[\eta,f(V)]=0$.
\end{proof}

Lemma~\ref{lem:semisimpletrrep} states that the trace of the representation of every element of a semisimple Lie algebra is 0.
This property allows us to ignore the symmetry constraints described by $S$ when we deal with Abelian Lie group $N$ in the proof of Proposition~\ref{prop:GGEderivation}.

\begin{lemma} \label{lem:semisimpletrrep}
	Let $\mathfrak{s}$ be a semisimple Lie algebra with a representation $\tilde{f}$ acting on a Hilbert space $\mathcal{R}$. 
	Then, for any $V\in\mathfrak{s}$, $\mathrm{tr}(\tilde{f}(V))=0$.	 
\end{lemma}

\begin{proof}
	We define $\mathrm{i}[\mathfrak{s}, \mathfrak{s}]:=\mathrm{span}(\{\mathrm{i}[X, Y]\ |\ X, Y\in\mathfrak{s}\})$.
	From Corollary~2.5.9~of~\cite{Goodman2009}, $\mathrm{i}[\mathfrak{s}, \mathfrak{s}]=\mathfrak{s}$.
	Therefore, for any $V\in\mathfrak{s}$, $V$ can be written as $V=\sum_i \mathrm{i}c_i[X_i, Y_i]$ with $c_i\in\mathbb{R}$ and $X_i, Y_i\in\mathfrak{s}$, and thus $\mathrm{tr}(\tilde{f}(V))=\sum_i \mathrm{i}c_i\mathrm{tr}([\tilde{f}(X_i), \tilde{f}(Y_i)])=0$.
\end{proof}

\subsubsection{Derivation of the generalized Gibbs ensemble}

In order to prove Theorem~\ref{thm:GFHCP}, we consider the following two conditions about a state $\rho$.

\begin{enumerate}[1]
\renewcommand{\theenumi}{A\arabic{enumi}}
 \item $\rho^{\otimes M}$ commutes with $\Omega$, for any $M\in\mathbb{N}$ and $\Omega\in\mathcal{B}^\mathrm{H}(\mathcal{H}^{\otimes M})$ that commutes with $F^{\otimes M}(G)$ and $H^{(M)}$. \label{conditionA1}
 \item There exists $\beta\in[0, \infty)$ that satisfies the following: $-\log(\braket{\Psi'_1|\rho^{\otimes L'}|\Psi'_1})+\log(\braket{\Psi'_0|\rho^{\otimes L'}|\Psi'_0})=\beta(\braket{\Psi'_1|H^{(L')}|\Psi'_1}-\braket{\Psi'_0|H^{(L')}|\Psi'_0})$ holds, for any $L'\in\mathbb{N}$ and any pair of simultaneous eigenstates $\ket{\Psi'_0}, \ket{\Psi'_1}\in\mathcal{H}^{\otimes L'}$ of $F^{\otimes L'}(G)$ and $H^{(L')}$ with $\braket{\Psi'_0|F^{\otimes L'}(g)|\Psi'_0}=\braket{\Psi'_1|F^{\otimes L'}(g)|\Psi'_1}$ for all $g\in G$. \label{conditionA2}
\end{enumerate}

	In Proposition~\ref{prop:GFworkextop}, we prove that symmetry-protected completely passive states satisfy these two conditions.
	We deal with Condition~\ref{conditionA1} in Proposition~\ref{prop:preGGEderivation} and Condition~\ref{conditionA2} in Proposition~\ref{prop:GGEderivation}, and then finally derive the generalized Gibbs ensemble.
	In order to make clear the meaning of Condition~\ref{conditionA2}, we consider all pairs of simultaneous eigenstates $\ket{\Psi'_0}$, $\ket{\Psi'_1}$ of $F^{\otimes L'}(G)$ that have the same eigenvalue for each $g\in G$.
	Then, Condition~\ref{conditionA2} means that any such pair of states have a common positive virtual temperature.
	The state $\ket{\Phi(\bm{n})}$ defined as Eq.~\eqref{eq:determinantstate} is useful for the construction of the simultaneous eigenstate $\ket{\Psi'_i}$.


	In Proposition~\ref{prop:preGGEderivation}, we consider Condition~\ref{conditionA1} in the case where $M=2$ and $\Omega$ is $C$ defined as Eq.~\eqref{eq:casimirdefinition}, and prove that if a state $\rho$ satifies Condition~\ref{conditionA1}, $\rho$ can be written as the product of the exponential of the representation of some element of the semisimple Lie subalgebra $\mathfrak{s}$ and an operator that commutes with $F(G)$.  
	When the symmetry group $G$ is a finite cyclic group or a dihedral group, Proposition~\ref{prop:prefiniteGEderivation} holds instead of Proposition~\ref{prop:preGGEderivation}.
	This is the difference between the cases of Lie group symmetry and finite group symmetry.

\begin{proposition} \label{prop:preGGEderivation}
	Let $G$ be a connected compact Lie group with a faithful unitary representation $F$ acting on $\mathcal{H}$, $\mathfrak{g}$ be the Lie algebra of $G$ with the associated representation $f$, $H\in\mathcal{B}(\mathcal{H})$ commute with $F(G)$ and $\rho\in\mathcal{B}^{++}(\mathcal{H})$.
	If $\rho$ satisfies Condition~\ref{conditionA1}, then $\rho$ can be written as 
	\begin{align}
		\rho=e^{-f(X^\mathrm{S})}\tau \label{eq:expsemisimple}
	\end{align}
	with some $X^\mathrm{S}\in\mathfrak{s}$ and $\tau\in\mathcal{B}^{++}(\mathcal{H})$ that commutes with $F(G)$, where $\mathfrak{s}$ is the Lie algebra of the Lie group $S$ defined in Eq.~\eqref{eq:Levi_decomposition}.
\end{proposition} 

\begin{proof}
	Suppose that $\rho\in\mathcal{B}^{++}(\mathcal{H})$ satisfies Condition~\ref{conditionA1}.
	From Lemma \ref{lem:casimir}, $C$ defined as Eq.~\eqref{eq:casimirdefinition} commutes with $f^{(2)}(\mathfrak{s})$.
	Since all elements of $\mathfrak{n}$ commute with $\mathfrak{s}$, $C$ also commutes with $f^{(2)}(\mathfrak{n})$.
	From Eq.~\eqref{eq:Levi_decomposition}, we obtain $F^{\otimes 2}(G)=F^{\otimes 2}(e^{\mathrm{i}\mathfrak{n}}e^{\mathrm{i}\mathfrak{s}})=e^{\mathrm{i}f^{(2)}(\mathfrak{n})}e^{\mathrm{i}f^{(2)}(\mathfrak{s})}$ and thus $C$ commutes with $F^{\otimes 2}(G)$.
	Since $H$ commutes with $f(\mathfrak{s})$, $C$ also commutes with $H^{(2)}$.	
	Since $[C, F^{\otimes 2}(G)]=[C, H^{(2)}]=0$ and $\rho$ satisfies Condition~\ref{conditionA1}, $\rho^{\otimes 2}$ commutes with $C$.
	Since $\rho$ is positive definite, we can define $\xi:=-\log(\rho)$ and $\xi^{(2)}$ commutes with $C$.
	From Proposition~\ref{prop:Casimir}, there exists $X^\mathrm{S}\in\mathfrak{s}$ such that $\eta:=\xi-f(X^\mathrm{S})$ commutes with $f(\mathfrak{s})$.
	Since $[f(\mathfrak{n}), F(G)]=[f(\mathfrak{n}), H]=0$ and $\rho$ satisfies Condition~\ref{conditionA1}, $\rho$ commutes with $f(\mathfrak{n})$ and thus $\xi$ commutes with $f(\mathfrak{n})$.
	Moreover, since $[f(X^\mathrm{S}), f(\mathfrak{n})]=f([X^\mathrm{S}, \mathfrak{n}])=0$, $\eta$ commutes with $f(\mathfrak{n})$.
	From Eq.~\eqref{eq:Levi_decomposition}, we obtain $F(G)=F(e^{\mathrm{i}\mathfrak{n}}e^{\mathrm{i}\mathfrak{s}})=e^{\mathrm{i}f(\mathfrak{n})}e^{\mathrm{i}f(\mathfrak{s})}$ and thus $\eta$ commutes with $F(G)$.
	Then, $\tau:=e^{-\eta}$ is a positive definite operator and commutes with $F(G)$, and $\rho$ can be written as $\rho=e^{-\xi}=e^{-f(X^\mathrm{S})-\eta}=e^{-f(X^\mathrm{S})}e^{-\eta}=e^{-f(X^\mathrm{S})}\tau$.
\end{proof}

Proposition~\ref{prop:GGEderivation} states that if a state $\rho$ satisfies Conditions~\ref{conditionA1} and \ref{conditionA2}, $\rho$ is the GGE at positive temperature.

\begin{proposition} \label{prop:GGEderivation}
	Let $G$ be a connected compact Lie group with a faithful unitary representation $F$ acting on $\mathcal{H}$, $\mathfrak{g}$ be the Lie algebra of $G$ with the associated representation $f$, $H\in\mathcal{B}(\mathcal{H})$ commute with $F(G)$, and $\rho\in\mathcal{B}^{++}(\mathcal{H})$ satisfy $\mathrm{tr}(\rho)=1$. 
	If $\rho$ satisfies Conditions~\ref{conditionA1} and \ref{conditionA2}, $\rho$ can be written as $\rho=\frac{1}{Z}e^{-\beta H-f(X)}$ with some $\beta\in[0, \infty)$ and $X\in\mathfrak{g}$, where $Z:=\mathrm{tr}(e^{-\beta H-f(X)})$.
\end{proposition}

\begin{proof}
	Suppose that $\rho$ satisfies Conditions~\ref{conditionA1} and \ref{conditionA2}. 
	From Proposition~\ref{prop:preGGEderivation}, $\rho$ can be written as
	\begin{align}
		\rho=e^{-f(X^\mathrm{S})} \tau
	\end{align}
	with some $X^\mathrm{S}\in\mathfrak{s}$ and $\tau\in\mathcal{B}^{++}(\mathcal{H})$ that commutes with $F(G)$.
	Since $\tau$ commutes with $F(G)$, in the same way as $H$ in Eq.~\eqref{eq:Hamiltoniandecomp}, $\tau$ can be written as $\tau=\sum_{\lambda\in\Lambda_F} \iota_\lambda (I_{\mathcal{R}_\lambda} \otimes \tau_\lambda) \iota_\lambda^\dag$ with some $\tau_\lambda\in\mathcal{B}^{++}(\mathcal{M}_\lambda)$. 
	For $\lambda\in\Lambda_F$ and $j=1, \cdots, m_\lambda$, we define $\Pi_{\lambda j}:=\iota_\lambda(I\otimes\ket{\psi_{\lambda j}}\bra{\psi_{\lambda j}})\iota_\lambda^\dag$, where $\ket{\psi_{\lambda j}}$ is defined in Eq.~\eqref{eq:Hamiltoniandecomp2}.
	Since $[\Pi_{\lambda j}, F(G)]=[\Pi_{\lambda j}, H]=0$ and $\rho$ satisfies Condition~\ref{conditionA1}, $\rho$ commutes with $\Pi_{\lambda j}$.
	Moreover, since $[f(X^\mathrm{S}), \Pi_{\lambda j}]=0$, $\sum_{\lambda\in\Lambda_F} \iota_\lambda (I_{\mathcal{R}_\lambda} \otimes [\tau_\lambda, \ket{\psi_{\lambda j}}\bra{\psi_{\lambda j}}]) \iota_\lambda^\dag=[\tau, \Pi_{\lambda j}]=[e^{-f(X^\mathrm{S})}\rho, \Pi_{\lambda j}]=0$ and thus $\tau_\lambda$ can be diagonalized as $\tau_\lambda=\sum_{j=1}^{m_\lambda} p_{\lambda j} \ket{\psi_{\lambda j}}\bra{\psi_{\lambda j}}$ with some $p_{\lambda j}\in(0, 1]$.
	Therefore, $\tau$ can be written as 
	\begin{align}
	\tau=\sum_{\lambda\in\Lambda_F} \iota_\lambda \left(I_{\mathcal{R}_\lambda} \otimes \sum_{j=1}^{m_\lambda} p_{\lambda j}\ket{\psi_{\lambda j}}\bra{\psi_{\lambda j}}\right) \iota_\lambda^\dag.
	\end{align}
		
	Take arbitrary $\bm{n}'_0=(n'_{0k})_{k\in K}, \bm{n}'_1=(n'_{1k})_{k\in K}\in\mathbb{N}^K$ satisfying $\bm{n}'_0\cdot\bm{w}_i=\bm{n}'_1\cdot\bm{w}_i$ for all $i=0, \cdots, a$, where $K$, $\bm{w}_0$ and $\bm{w}_i$ for $i=1, \cdots, a$ are defined as $K:=\{(\lambda, j)\ |\ \lambda\in\Lambda_F, j=1, \cdots, m_\lambda\}$, $\bm{w}_0:=(1)_{k\in K}$, $\bm{w}_i:=(w_{ik})_{k\in K}$, $w_{ik}:=w_{i \lambda}$ for $k=(\lambda, j)\in K$, and $w_{i \lambda}$ is defined in Lemma~\ref{lem:Liealgebrabasis} for $i=1, \cdots, a$.
	We define $\ket{\Psi'_0}, \ket{\Psi'_1}\in\mathcal{H}^{\otimes L'}$ as $\ket{\Psi'_i}:=\ket{\Phi(\bm{n}'_i)}$ with $\ket{\Phi(\bm{n})}$ defined as Eq.~\eqref{eq:determinantstate}, where $L':=D\bm{n}'_0\cdot\bm{w}_0=D\bm{n}'_1\cdot\bm{w}_0$.
	For any $i=1, \cdots, a$ and $V\in\mathfrak{s}$, $f^{(L')}(W_i)$, $f^{(L')}(V)$, and $H^{(L')}$ can be regarded as $\Omega$ in Eq.~\eqref{eq:Omega_decomp} in the case where $\zeta_k=f_\lambda(W_i)$, $\zeta_k=f_\lambda(V)$ and $\zeta_k=E_{\lambda j}I_{\mathcal{R}_\lambda}$ respectively, where $W_i$ is defined in Lemma~\ref{lem:Liealgebrabasis} and $E_{\lambda j}$ is defined in Eq.~\eqref{eq:Hamiltoniandecomp2}.
	From Lemma~\ref{lem:eigenvectorcomposition}, $\ket{\Psi'_l}$ is a simultaneous eigenstate of $f^{(L')}(W_i)$, $f^{(L')}(V)$ and $H^{(L')}$.
	For any $k=(\lambda, j)\in K$, $u_k$ in Eq.~\eqref{eq:Omega_eigenvalue} corresponding to these three cases is given by 
	\begin{align}
		&u_k=\frac{1}{r_\lambda}\mathrm{tr}(f_\lambda(W_i))=\frac{1}{r_\lambda}\mathrm{tr}(w_{i \lambda} I_{\mathcal{R}_\lambda})=w_{i \lambda}=w_{ik},\\
		&u_k=\frac{1}{r_\lambda}\mathrm{tr}(f_\lambda(V))=0, \label{eq:u_k_for_V}\\
		&u_k=\frac{1}{r_\lambda}\mathrm{tr}(E_{\lambda j}I_{\mathcal{R}_\lambda})=E_{\lambda j},
	\end{align}
	where we used Lemma~\ref{lem:semisimpletrrep} in Eq.~\eqref{eq:u_k_for_V}.
	Therefore, we obtain
	\begin{align}
		&f^{(L')}(W_i)\ket{\Psi'_l}=D\bm{n}'_l\cdot\bm{w}_i\ket{\Psi'_l}, \label{eq:f_L'_W_i_eigenvalue}\\
		&f^{(L')}(V)\ket{\Psi'_l}=0, \label{eq:f_L'_V_eigenvalue}\\
		&H^{(L')}\ket{\Psi'_l}=D\bm{n}'_l\cdot\bm{E}\ket{\Psi'_l}, \label{eq:H_L'_eigenvalue}
	\end{align}
	where $\bm{E}:=(E_k)_{k\in K}$ and $E_k:=E_{\lambda j}$ for $k=(\lambda, j)\in K$.
	For any $g\in G$, from Eq.~\eqref{eq:casimirdefinition}, $g$ can be written as $g=e^{\mathrm{i}\left(\sum_{i=1}^a \alpha_i W_i+V\right)}$ with some $\alpha_1, \cdots, \alpha_a\in\mathbb{R}$ and $V\in\mathfrak{s}$, and thus 
	\begin{align}
		F^{\otimes L'}(g)=F^{\otimes L'}\left(e^{\mathrm{i}\left(\sum_{i=1}^a \alpha_i W_i+V\right)}\right)=e^{\mathrm{i}\left(\sum_{i=1}^a \alpha_i f^{(L')}(W_i)+f^{(L')}(V)\right)}. \label{eq:F_L'_decomposition}
	\end{align}
	From Eqs.~\eqref{eq:f_L'_W_i_eigenvalue}, \eqref{eq:f_L'_V_eigenvalue} and \eqref{eq:F_L'_decomposition}, 
	\begin{align}
		F^{\otimes L'}(g)\ket{\Psi'_l}=e^{\mathrm{i}\sum_{i=1}^a \alpha_i D\bm{n}'_l\cdot\bm{w}_i}\ket{\Psi'_l}.
	\end{align}
	Since $\bm{n}'_0\cdot\bm{w}_i=\bm{n}'_1\cdot\bm{w}_i$ for all $i=1, \cdots, a$, $\ket{\Psi'_0}$ and $\ket{\Psi'_1}$ are eigenstates of $F^{\otimes L'}(g)$ with the same eigenvalue.
	Since $\rho$ satisfies Condition~\ref{conditionA2}, there exists $\beta\in[0, \infty)$ such that 
	\begin{align}
		-\log(\braket{\Psi'_1|\rho^{\otimes L'}|\Psi'_1})+\log(\braket{\Psi'_0|\rho^{\otimes L'}|\Psi'_0})=\beta(\braket{\Psi'_1|H^{(L')}|\Psi'_1}-\braket{\Psi'_0|H^{(L')}|\Psi'_0}). \label{eq:virtualtemperature}
	\end{align}
	$\tau$ can be regarded as $\Omega$ in Eq.~\eqref{eq:Omega_decomp} in the case where $\zeta_k=p_k I_{\mathcal{R}_\lambda}$ for all $k\in K$, where $p_k:=p_{\lambda j}$ for $k=(\lambda, j)\in K$.
	From Lemma~\ref{lem:eigenvectorcomposition}, for $l=0, 1$, $\ket{\Psi'_l}$ is an eigenstate of $\tau^{\otimes L'}$ and the eigenvalue is given by
	\begin{align}
		\braket{\Psi'_l|\tau^{\otimes L'}|\Psi'_l}=\prod_{k\in K} (\mathrm{det}(p_k I_{\mathcal{R}_\lambda}))^{n'_{lk}\frac{D}{r_\lambda}}
		=\prod_{k\in K} \left(p_k^{r_\lambda}\right)^{n'_{lk}\frac{D}{r_\lambda}}
		=\prod_{k\in K} p_k^{D n'_{lk}}. \label{eq:rho_L'_eigenvalue}
	\end{align}
	From Eqs.~\eqref{eq:f_L'_V_eigenvalue} and \eqref{eq:rho_L'_eigenvalue}, for $l=0, 1$, 
	\begin{align}
		\braket{\Psi'_l|\rho^{\otimes L'}|\Psi'_l}=\braket{\Psi'_l|(e^{-f(X^\mathrm{S})}\tau)^{\otimes L'}|\Psi'_l}=\braket{\Psi'_l|e^{-f^{(L')}(X^\mathrm{S})}\tau^{\otimes L'}|\Psi'_l}=\prod_{k\in K} p_k^{D n'_{lk}}. \label{eq:F_L'_g_eigenvalue}
	\end{align}
	Therefore,
	\begin{align}
		-\log(\braket{\Psi'_l|\rho^{\otimes L'}|\Psi'_l})
		=-\sum_{k\in K} \log\left(p_k^{D n'_{lk}}\right)
		=\sum_{k\in K} D n'_{lk}(-\log(p_k))
		=D \bm{n}'_l\cdot\bm{s}, \label{eq:log_rho_L'_eigenvalue}
	\end{align}
	where $\bm{s}:=(s_k)_{k\in K}$, $s_k:=s_{\lambda j}$ and $s_{\lambda j}:=-\log(p_{\lambda j})$. 
	By substituting Eqs.~\eqref{eq:H_L'_eigenvalue} and \eqref{eq:log_rho_L'_eigenvalue} into Eq.~\eqref{eq:virtualtemperature}, we get
	\begin{align}
		D\bm{n}'_0\cdot\bm{s}-D\bm{n}'_1\cdot\bm{s}=\beta(D\bm{n}'_0\cdot\bm{E}-D\bm{n}'_1\cdot\bm{E}).
	\end{align}
	This implies that $\bm{n}'_0\cdot\bm{t}=\bm{n}'_1\cdot\bm{t}$, where $\bm{t}:=\bm{s}-\beta\bm{E}$.
	Since this holds for any $\bm{n}'_0,\ \bm{n}'_1\in\mathbb{N}^K$ satisfying $\bm{n}'_0\cdot\bm{w}_i=\bm{n}'_1\cdot\bm{w}_i$ for all $i=0, \cdots, a$, $\bm{t}$ can be written as $\bm{t}=\sum_{i=0}^a \mu_i\bm{w}_i$ with some $\mu_0,\cdots,\mu_a\in\mathbb{R}$ from Lemma~\ref{lem:latticeduality}. 
	Then, for any $\lambda\in\Lambda_F$ and $j=1, \cdots, m_\lambda$, we obtain $s_{\lambda j}-\beta E_{\lambda j}=\mu_0+\sum_{i=1}^a \mu_i w_{i\lambda}$.
	Therefore, we obtain
	\begin{align}
		-\log(\tau)-\beta H
		=&\sum_{\lambda\in\Lambda_F} \sum_{j=1}^{m_\lambda} (s_{\lambda j}-\beta E_{\lambda j})\iota_\lambda(I_{\mathcal{R}_\lambda}\otimes\ket{\psi_{\lambda j}}\bra{\psi_{\lambda j}})\iota_\lambda^\dag \nonumber\\
		=&\sum_{\lambda\in\Lambda_F} \sum_{j=1}^{m_\lambda} \left(\mu_0+\sum_{i=1}^a \mu_i w_{i\lambda}\right)\iota_\lambda(I_{\mathcal{R}_\lambda}\otimes\ket{\psi_{\lambda j}}\bra{\psi_{\lambda j}})\iota_\lambda^\dag \nonumber\\
		=&\sum_{\lambda\in\Lambda_F} \left(\mu_0+\sum_{i=1}^a \mu_i w_{i\lambda}\right)\iota_\lambda\left(I_{\mathcal{R}_\lambda}\otimes\sum_{j=1}^{m_\lambda}\ket{\psi_{\lambda j}}\bra{\psi_{\lambda j}}\right)\iota_\lambda^\dag \nonumber\\
		=&\sum_{\lambda\in\Lambda_F} \left(\mu_0+\sum_{i=1}^a \mu_i w_{i\lambda}\right)\iota_\lambda\iota_\lambda^\dag \nonumber\\
		=&\mu_0\sum_{\lambda\in\Lambda_F} \iota_\lambda\iota_\lambda^\dag+\sum_{i=1}^a  \left(\mu_i \sum_{\lambda\in\Lambda_F} w_{i\lambda}\iota_\lambda\iota_\lambda^\dag\right) \nonumber\\
		=&\mu_0 I+\sum_{i=1}^a \mu_i f(W_i) \nonumber\\
		=&\mu_0 I+f(X^\mathrm{C}),
	\end{align}
	where $X^\mathrm{C}:=\sum_{i=1}^a \mu_i W_i$.	
	This implies that 
	\begin{align}
		\rho=e^{-f(X^\mathrm{S})}\tau=e^{-f(X^\mathrm{S})}e^{-\mu_0 I-\beta H-f(X^\mathrm{C})}=\frac{1}{e^{\mu_0}}e^{-\beta H-f(X)}.
	\end{align}
	where $X:=X^\mathrm{C}+X^\mathrm{S}\in\mathfrak{g}$.
	Since $\mathrm{tr}(\rho)=1$, we obtain $e^{\mu_0}=\mathrm{tr}(e^{-\beta H-f(X)})=Z$ and thus  $\rho=\frac{1}{Z}e^{-\beta H-f(X)}$.
\end{proof}

	Proposition~\ref{prop:GFworkextop} states that we can extract positive work from multiple copies of a state other than the GGE by some symmetry-respecting operation.
	From Proposition~\ref{prop:GGEderivation}, if a state is not the GGE,  Conditions~\ref{conditionA1} and \ref{conditionA2} cannot be simultaneously satisfied for the state.
 	From Propositions~\ref{prop:workextop1}~and~\ref{prop:virtualtemperature}, for some $N\in\mathbb{N}$, we can construct a symmetry-respecting operator $U$ that extracts positive work from $N$ copies of such a state.
	We also prove in Proposition~\ref{prop:GFworkextop} that $U$ satisfies $[U^\dag H^{(N)}U, H^{(N)}]=0$.
	This property is used in Theorem~\ref{thm:GFHCP1}, which adopts a setup with a quantum work storage.

\begin{proposition} \label{prop:GFworkextop}
	Let $G$ be a connected compact Lie group with a faithful unitary representation $F$ acting on $\mathcal{H}$, $\mathfrak{g}$ be the Lie algebra of $G$ with the associated representation $f$, $H\in\mathcal{B}(\mathcal{H})$ commute with $F(G)$ and be not $(G, F)$-trivial, and $\rho\in\mathcal{B}^{++}(\mathcal{H})$ satisfy $\mathrm{tr}(\rho)=1$.
	If $\rho$ cannot be written as $\rho=\frac{1}{Z}e^{-\beta H-f(X)}$ with $\beta\in[0, \infty)$, $X\in\mathfrak{g}$ and $Z:=\mathrm{tr}(e^{-\beta H-f(X)})$, then there exist $N\in\mathbb{N}$ and $U\in\mathcal{U}_{G, F^{\otimes N}}(\mathcal{H}^{\otimes N})$ that satisfy $[U^\dag H^{(N)}U, H^{(N)}]=0$ and $W(\rho^{\otimes N}, H^{(N)}, U)>0$.
\end{proposition}

\begin{proof}
	Suppose that $\rho$ cannot be written as $\rho=\frac{1}{Z}e^{-\beta H-f(X)}$ with $\beta\in[0, \infty)$, $X\in\mathfrak{g}$, $Z:=\mathrm{tr}(e^{-\beta H-f(X)})$.
	From Proposition~\ref{prop:GGEderivation}, Conditions~\ref{conditionA1} and \ref{conditionA2} cannot be simultaneously satisfied for $\rho$.
	First, we consider the case where $\rho$ does not satisfy Condition~\ref{conditionA1}.
	Since $H$ is not $(G, F)$-trivial, from Lemma~\ref{lem:PsiexistenceforLie}, we can take $L\in\mathbb{N}$ and a pair of simultaneous eigenstates $\ket{\Psi_0}, \ket{\Psi_1}\in\mathcal{H}^{\otimes L}$ of $F^{\otimes L}(G)$ and $H^{(L)}$ such that the eigenvalues satisfy 
	\begin{align}
		\forall g\in G\ \braket{\Psi_0|F^{\otimes L}(g)|\Psi_0}=\braket{\Psi_1|F^{\otimes L}(g)|\Psi_1},\ \braket{\Psi_0|H^{(L)}|\Psi_0}\neq\braket{\Psi_1|H^{(L)}|\Psi_1}. \label{eq:Psi_condition}
	\end{align}
	We define $\Delta\mathcal{E}$ as Eq.~\eqref{eq:ESdef}.
	Without loss of generality, we can suppose that $\Delta\mathcal{E}>0$.
	Since $\rho$ does not satisfy Condition~\ref{conditionA1}, we can take $M\in\mathbb{N}$ and $\Omega\in\mathcal{B}^\mathrm{H}(\mathcal{H}^{\otimes M})$ such that $[\Omega, F^{\otimes M}(G)]=[\Omega, H^{(M)}]=0$ but $[\rho^{\otimes M}, \Omega]\neq0$.
	Let the spectral decomposition of $\Omega$ be $\Omega=\sum_\omega \omega P_\omega$, where $\omega$ is an eigenvalue of $\Omega$ and $P_\omega$ is the projection operator onto the eigenspace of $\omega$.
	Then, there exists some $P_\omega$ such that $[P_\omega, F^{\otimes M}(G)]=[P_\omega, H^{(M)}]=0$ but $[\rho^{\otimes M}, P_\omega]\neq0$.
	For $m\in\mathbb{N}$ and $i, j\in\{0, 1\}$, we define $A_{ij}:=\{[I-(-1)^i T][P_\omega\otimes (I-P_\omega)][I-(-1)^j T]\}^{\otimes m}\otimes\ket{\Psi_i}\bra{\Psi_j}$, where $T$ is the swapping operator on $\mathcal{H}^{\otimes M}\otimes\mathcal{H}^{\otimes M}$.
	From Proposition~\ref{prop:workextop1}, $\mathcal{O}(\{A_{ij}\})\in\mathcal{U}(\mathcal{H}^{\otimes 2mM+L})$ and $W(\rho^{\otimes 2mM+L}, H^{(2mM+L)}, \mathcal{O}(\{A_{ij}\}))>0$ for some $m\in\mathbb{N}$.
	Since $P_\omega$, $T$ and $\ket{\Psi_i}\bra{\Psi_j}$ respectively commute with $F^{\otimes M}(G)$, $F^{\otimes 2M}(G)$ and $F^{\otimes L}(G)$, $A_{ij}$ commutes with $F^{\otimes 2mM+L}(G)$ and thus $\mathcal{O}(\{A_{ij}\})\in\mathcal{U}_{G, F^{\otimes 2mM+L}}(\mathcal{H}^{\otimes 2mM+L})$.
	In the same way as Eq.~\eqref{eq:HAcommutation}, $[H^{(2mM+L)}, A_{ij}]=\Delta\mathcal{E}(i-j)A_{ij}$.
	Therefore, from Lemma~\ref{lem:U1extractedwork}, $[\mathcal{O}(\{A_{ij}\})^\dag H^{(2mM+L)}\mathcal{O}(\{A_{ij}\}), H^{(2mM+L)}]=0$.
		
	Next, we consider the case where $\rho$ does not satisfy Condition~\ref{conditionA2}.
	We define $\beta$ as Eq.~\eqref{eq:ESdef}.
	Then, $\beta<0$ or the following is satisfied: there exist $L'\in\mathbb{N}$ and a pair of simultaneous eigenstates $\ket{\Psi'_0}, \ket{\Psi'_1}\in\mathcal{H}^{\otimes L'}$ of $F^{\otimes L'}(G)$ and $H^{(L')}$ such that $\braket{\Psi'_0|F^{\otimes L'}(g)|\Psi'_0}=\braket{\Psi'_1|F^{\otimes L'}(g)|\Psi'_1}$ for all $g\in G$ and $\Delta S'\neq\beta\Delta\mathcal{E}'$, where $\Delta\mathcal{E}'$ and $\Delta S'$ are defined as Eq.~\eqref{eq:ESdef}.
	Equivalently, we can take $L'\in\mathbb{N}$ and a pair of simultaneous eigenstates $\ket{\Psi'_0}, \ket{\Psi'_1}\in\mathcal{H}^{\otimes L'}$ of $F^{\otimes L'}(G)$ and $H^{(L')}$ such that (i) $\braket{\Psi'_0|F^{\otimes L'}(g)|\Psi'_0}=\braket{\Psi'_1|F^{\otimes L'}(g)|\Psi'_1}$ for all $g\in G$ and (ii) $\beta<0$ or $\Delta S'\neq\beta\Delta\mathcal{E}'$.
	By noting that the condition $\Delta S'\neq\beta\Delta\mathcal{E}'$ is invariant under the exchange of $\ket{\Psi'_0}$ and $\ket{\Psi'_1}$, we can suppose that (I) $\Delta\mathcal{E}'>0$ or (II) $\Delta\mathcal{E}'=0$ and $\Delta S'\geq0$ by exchanging $\ket{\Psi'_0}$ and $\ket{\Psi'_1}$ if necessary.
	For $m, m'\in\mathbb{N}$ and $i, j\in\{0, 1\}$, we define $B_{ij}:=(\ket{\Psi_i}\bra{\Psi_j})^{\otimes m}\otimes(\ket{\Psi'_{1-i}}\bra{\Psi'_{1-j}})^{\otimes m'}$.
	From Proposition~\ref{prop:virtualtemperature}, $\mathcal{O}(\{B_{ij}\})\in\mathcal{U}(\mathcal{H}^{\otimes mL+m'L'})$ and $W(\rho^{\otimes mL+m'L'}, H^{(mL+m'L')}, \mathcal{O}(\{B_{ij}\}))>0$ for some $m, m'\in\mathbb{N}$.
	Since $\ket{\Psi_i}\bra{\Psi_j}$ and $\ket{\Psi'_{1-i}}\bra{\Psi'_{1-j}}$ respectively commute with $F^{\otimes L}(G)$ and $F^{\otimes L'}(G)$, $B_{ij}$ commutes with $F^{\otimes mL+m'L'}(G)$ and thus $\mathcal{O}(\{B_{ij}\})\in\mathcal{U}_{G, F^{\otimes mL+m'L'}}(\mathcal{H}^{\otimes mL+m'L'})$.  
	In the same way as Eq.~\eqref{eq:HBcommutation}, $[H^{(mL+m'L')}, B_{ij}]=(m\Delta\mathcal{E}-m'\Delta\mathcal{E}')B_{ij}$.
	Therefore, from Lemma~\ref{lem:U1extractedwork}, $[\mathcal{O}(\{B_{ij}\})^\dag H^{(mL+m'L')}\mathcal{O}(\{B_{ij}\}), H^{(mL+m'L')}]=0$.
\end{proof}

We are now ready to prove Theorem~\ref{thm:GFHCP}.
Proposition~\ref{prop:GFworkextop} shows that $(G, F)$-completely passive states are GGEs.
We can easily prove the converse by an argument based on Ref.~\cite{Lostaglio2017}.
The proof of Theorem~\ref{thm:GFHCP} is as follows:\\\\
\textit{Proof of Theorem~\ref{thm:GFHCP}.}
	First, suppose that $\rho$ cannot be written as $\rho=\frac{1}{Z}e^{-\beta H-f(X)}$ with $\beta\in[0, \infty)$, $X\in\mathfrak{g}$, and $Z:=\mathrm{tr}(e^{-\beta H-f(X)})$.
	From Proposition~\ref{prop:GFworkextop}, there exist $N\in\mathbb{N}$ and $U\in\mathcal{U}_{G, F^{\otimes N}}(\mathcal{H}^{\otimes N})$ such that $W(\rho^{\otimes N}, H^{(N)}, U)>0$.
	This implies that $\rho$ is not $(G, F)$-completely passive w.r.t. $H$.

	Next, we show the converse.
	Suppose that $\rho$ can be written as $\rho=\frac{1}{Z}e^{-\beta H-f(X)}$ with some $\beta\in[0, \infty)$, $X\in\mathfrak{g}$,  $Z:=\mathrm{tr}(e^{-\beta H-f(X)})$.
	Take arbitrary $N\in\mathbb{N}$ and $U\in\mathcal{U}_{G, F^{\otimes N}}(\mathcal{H}^{\otimes N})$. 
	In the case of $\beta=0$, $U$ does not change the state $\rho^{\otimes N}$ and thus the extracted work is $0$. 
	In the case of $\beta>0$, the extracted work satisfies
	\begin{align}
		W(\rho^{\otimes N}, H^{(N)}, U)=-\frac{1}{\beta}S(U\rho^{\otimes N}U^\dagger\|\rho^{\otimes N})\leq 0,
	\end{align}
where $S(\cdot\|\cdot)$ is the quantum relative entropy and we used its positivity. 
	This implies that $\rho$ is $(G, F)$-completely passive w.r.t. $H$. \hspace{\fill} $\Box$\\

\subsection{Proof of Theorem~\ref{thm:finiteGFHCP}}

	In this subsection, we suppose that $G$ is a finite cyclic group or a dihedral group.
	Conditions~\ref{conditionA1} and \ref{conditionA2} are also useful for the proof of Theorem~\ref{thm:finiteGFHCP}. 
	The structure of the proof of Theorem~\ref{thm:finiteGFHCP} is parallel to that of Theorem~\ref{thm:GFHCP}.
	In fact, we prove in Proposition~\ref{prop:finiteGFworkextop} that $(G, F)$-completely passive states satisfy Conditions~\ref{conditionA1} and \ref{conditionA2}.
	We deal with Condition~\ref{conditionA1} in Proposition~\ref{prop:prefiniteGEderivation} and Condition~\ref{conditionA2} in Proposition~\ref{prop:finiteGEderivation}, and then derive the Gibbs ensemble.
	However, the  details of the proofs are different, because there do not exist the counterparts of Lie algebras or Casimir operators in finite groups.
	Instead, we make use of the finiteness of $G$.

	Proposition~\ref{prop:prefiniteGEderivation} states that if a state $\rho$ satisfies Condition~\ref{conditionA1}, $\rho$ commutes with $F(G)$.
	This corresponds to Proposition~\ref{prop:preGGEderivation} in the case of Lie group symmetry.
	The proof is straightforward in the case of a finite cyclic group, but it is complicated in the case of a dihedral group due to its noncommutativity.
	We construct several operators that play the role of $\Omega$ in Condition~\ref{conditionA1} with the projection operators onto the eigenspaces of a symmetry operator.

\begin{proposition} \label{prop:prefiniteGEderivation}
	Let $G$ be a finite cyclic group or a dihedral group with a unitary representation $F$ acting on $\mathcal{H}$, and $H\in\mathcal{B}^\mathrm{H}(\mathcal{H})$ commute with $F(G)$.
	If $\rho\in\mathcal{B}^{+}(\mathcal{H})$ satisfies Condition~\ref{conditionA1}, $\rho$ commutes with $F(G)$. 
\end{proposition}

\begin{proof}
	Suppose that $\rho$ satisfies Condition~\ref{conditionA1}.
	In the case where $G$ is a finite cyclic group, for any $g\in G$, $F(g)$ commutes with $F(G)$ and $H$. 
	Since $\rho$ satisfies Condition~\ref{conditionA1}, $[\rho, F(g)]=0$ for all $g\in G$, i.e., $\rho$ commutes with $F(G)$. 
	
	We consider the case where $G$ is a dihedral group.
	Let the dihedral group be written as $D_n=\{1, t, \cdots, t^{n-1}, r, rt, \cdots, rt^{n-1}\}$ with $t$ and $r$ satisfying $t^n=r^2=1$ and $tr=rt^{-1}$. 
	Since every element of $D_n$ is generated by $t$ and $r$, it is sufficient to prove that $[\rho, F(t)]=0$ and $[\rho, F(r)]=0$ in order to prove that $\rho$ commutes with $F(G)$.
	First, we prove that $[\rho, F(t)]=0$. 
	Since $F(t)^n=F(t^n)=F(1)=I$, every eigenvalue $z$ of $F(t)$ satisfies $z^n=1$.
	We define $\Pi_z$ as the projection operator onto the eigenspace of $z$. 
	Then, the spectral decomposition of $F(t)$ is written as $F(t)=\sum_{z: z^n=1} z\Pi_z$, and thus \begin{align}
		F(r)^\dag F(t)F(r)=\sum_{z: z^n=1} zF(r)^\dag \Pi_zF(r). \label{eq:rtr1}
	\end{align}
	On the other hand, since $r^{-1}tr=t^{-1}$, we get 
	\begin{align}
		F(r)^\dag F(t)F(r)=F(r^{-1}tr)=F(t^{-1})=F(t)^\dag=\sum_{z: z^n=1} z^*\Pi_z=\sum_{z: z^n=1} z\Pi_{z^*}. \label{eq:rtr2}
	\end{align}
	By comparing Eqs.~\eqref{eq:rtr1} and \eqref{eq:rtr2}, we get 
	\begin{align}
		F(r)^\dag \Pi_zF(r)=\Pi_{z^*}. \label{eq:rtr3}
	\end{align}
	If $z=z^*$, Eq.~\eqref{eq:rtr3} implies that $[\Pi_z, F(r)]=0$.
	Since $\Pi_z$ also commutes with $F(t)$, $\Pi_z$ commutes with $F(G)$.
	Since $F(t)$ commutes with $H$, $\Pi_z$ commutes with $H$. 
	Since $\rho$ satisfies Condition~\ref{conditionA1}, $[\rho, \Pi_z]=0$.
	If $z\neq z^*$, from Eq.~\eqref{eq:rtr3},
	\begin{align}
		F^{\otimes 2}(r)^\dag (\Pi_z^{\otimes2}+\Pi_{z^*}^{\otimes2})F^{\otimes 2}(r)=\Pi_{z^*}^{\otimes2}+\Pi_z^{\otimes2} \label{eq:rtr4}
	\end{align}
	and thus $[\Pi_z^{\otimes2}+\Pi_{z^*}^{\otimes2}, F^{\otimes 2}(r)]=0$.
	Since $\Pi_z^{\otimes2}+\Pi_{z^*}^{\otimes2}$ also commutes with $F^{\otimes 2}(t)$, $\Pi_z^{\otimes2}+\Pi_{z^*}^{\otimes2}$ commutes with $F^{\otimes 2}(G)$.
	Since $\Pi_z$ and $\Pi_{z^*}$ commute with $H$, $\Pi_z^{\otimes2}+\Pi_{z^*}^{\otimes2}$ commutes with $H^{(2)}$.
	Since $\rho$ satisfies Condition~\ref{conditionA1}, $[\rho^{\otimes 2}, \Pi_z^{\otimes2}+\Pi_{z^*}^{\otimes2}]=0$.
	Therefore, we get
	\begin{align}
		\mathrm{tr}(\rho\Pi_z)[\rho,\Pi_z]=\mathrm{tr}_{\mathcal{H}_1}([\rho^{\otimes 2}, \Pi_z^{\otimes2}+\Pi_{z^*}^{\otimes2}](\Pi_z\otimes I))=0,
	\end{align}
	where $\mathcal{H}_1$ is the Hilbert space of the first copy of the system.
	If $\mathrm{tr}(\rho\Pi_z)\neq0$, then $[\rho,\Pi_z]=0$.
	If $\mathrm{tr}(\rho\Pi_z)=0$, then $\|\rho^{\frac{1}{2}}\Pi_z\|_\mathrm{HS}^2=\mathrm{tr}(\rho\Pi_z)=0$.
	Thus we get $\rho\Pi_z=\rho^{\frac{1}{2}}\cdot\rho^{\frac{1}{2}}\Pi_z=0$ and $\Pi_z\rho=(\rho\Pi_z)^*=0$.
	In both cases, we obtain $[\rho, \Pi_z]=0$.
	Here we have proved that for any eigenvalue $z$ of $F(t)$, $[\rho, \Pi_z]=0$.
	This implies that $[\rho, F(t)]=0$. 
	
	Next, we prove that $[\rho, F(r)]=0$.
	If $z=z^*$, since $[\Pi_z, F(r)]=0$, $\Pi_z F(r)$ commutes with $F(r)$.
	Since $[\Pi_z, F(r)]=0$ and $F(t)\Pi_z=\Pi_z F(t)=\Pi_z$, we get
	\begin{align}
		\Pi_z F(r)F(t)=F(r)\Pi_z F(t)=F(r) \Pi_z=\Pi_z F(r)=F(t)\Pi_z F(r). \label{eq:trt}
	\end{align}
	This implies that $[\Pi_z F(r), F(t)]=0$, and thus $\Pi_z F(r)$ commutes with $F(G)$. 
	Since $\Pi_z$ and $F(r)$ commute with $H$, $\Pi_z F(r)$ commutes with $H$.
	Since $\rho$ satisfies Condition~\ref{conditionA1}, $[\rho, \Pi_z F(r)]=0$.
	Since $\rho$ commutes with $\Pi_z$, this implies that $\Pi_z[\rho, F(r)]=0$.
	If $z\neq z^*$, in the same way as Eq.~\eqref{eq:rtr4}, we obtain $[\Pi_z+\Pi_{z^*}, F(r)]=0$ and thus $(\Pi_z+\Pi_{z^*})F(r)$ commutes with $F(r)$.
	Since $[\Pi_z+\Pi_{z^*}, F(r)]=0$ and $F(t)(\Pi_z+\Pi_{z^*})=(\Pi_z+\Pi_{z^*})F(t)=\Pi_z+\Pi_{z^*}$, we can prove that $(\Pi_z+\Pi_{z^*})F(r)$ commutes with $F(t)$ in the same way as Eq.~\eqref{eq:trt}, and thus $(\Pi_z+\Pi_{z^*})F(r)$ commutes with $F(G)$.
	Since $\Pi_z$, $\Pi_{z^*}$ and $F(r)$ commute with $H$, $(\Pi_z+\Pi_{z^*})F(r)$ commutes with $H$.
	Since $\rho$ satisfies Condition~\ref{conditionA1}, $[\rho, (\Pi_z+\Pi_{z^*})F(r)]=0$.
	Since $\rho$ also commutes with $\Pi_z$ and $\Pi_{z^*}$, we get
	\begin{align}
		\Pi_z[\rho, F(r)]=\Pi_z(\Pi_z+\Pi_{z^*})[\rho, F(r)]=\Pi_z[\rho, (\Pi_z+\Pi_{z^*})F(r)]=0.
	\end{align}
	Here we have proved that for any eigenvalue $z$ of $F(t)$, $\Pi_z[\rho, F(r)]=0$.
	This implies that $[\rho, F(r)]=\sum_{z: z^n=1} \Pi_z[\rho, F(r)]=0$.
	Since $\rho$ commutes with $F(t)$ and $F(r)$, $\rho$ commutes with $F(G)$.
\end{proof}

Proposition~\ref{prop:finiteGEderivation} states that if a state $\rho$ satisfies Conditions~\ref{conditionA1} and \ref{conditionA2}, $\rho$ is the Gibbs ensemble at positive temperature. 
This corresponds to Proposition~\ref{prop:GGEderivation} in the case of Lie group symmetry.
In the proof of Proposition~\ref{prop:finiteGEderivation}, we use the finiteness of $G$ instead of the properties of Lie algebras.
Note that if we can prove Proposition~\ref{prop:prefiniteGEderivation} for general finite groups, we can also prove Proposition~\ref{prop:finiteGEderivation} for general finite groups in the same manner as the proof below.

\begin{proposition} \label{prop:finiteGEderivation}
	Let $G$ be a finite cyclic group or a dihedral group with a unitary representation $F$ acting on $\mathcal{H}$, $H\in\mathcal{B}^\mathrm{H}(\mathcal{H})$ commute with $F(G)$, and $\rho\in\mathcal{B}^{++}(\mathcal{H})$ satisfy $\mathrm{tr}(\rho)=1$.
	If $\rho$ satisfies Conditions~\ref{conditionA1} and \ref{conditionA2}, then $\rho$ can be written as $\rho=\frac{1}{Z}e^{-\beta H}$ with some $\beta\in[0, \infty)$, where $Z:=\mathrm{tr}(e^{-\beta H})$.
\end{proposition}

\begin{proof}
	Suppose that $\rho$ satisfies Conditions~\ref{conditionA1} and \ref{conditionA2}. 
	From Proposition~\ref{prop:prefiniteGEderivation}, $\rho$ commutes with $F(G)$.
	In the same way as $H$ in Eq.~\eqref{eq:Hamiltoniandecomp}, $\rho$ can be written as $\rho=\sum_{\lambda\in\Lambda_F} \iota_\lambda (I_{\mathcal{R}_\lambda} \otimes \tau_\lambda) \iota_\lambda^\dag$ with some $\tau_\lambda\in\mathcal{B}^{++}(\mathcal{M}_\lambda)$. 
	In the same way as Proposition~\ref{prop:GGEderivation}, $\tau_\lambda$ can be diagonalized as $\tau_\lambda=\sum_{j=1}^{m_\lambda} p_{\lambda j} \ket{\psi_{\lambda j}}\bra{\psi_{\lambda j}}$ with some $p_{\lambda j}\in(0, 1]$, where $\ket{\psi_{\lambda j}}$ is defined in Eq.~\eqref{eq:Hamiltoniandecomp2}.
	Therefore, $\rho$ can be written as 
	\begin{align}
	\rho=\sum_{\lambda\in\Lambda_F} \iota_\lambda \left(I_{\mathcal{R}_\lambda}\otimes \sum_{j=1}^{m_\lambda} p_{\lambda j}\ket{\psi_{\lambda j}}\bra{\psi_{\lambda j}}\right) \iota_\lambda^\dag.
	\end{align}
	
	Take arbitrary $\bm{n}'_0=(n'_{0k})_{k\in K}, \bm{n}'_1=(n'_{1k})_{k\in K}\in\mathbb{N}^K$ satisfying $\bm{n}'_0\cdot\bm{w}_0=\bm{n}'_1\cdot\bm{w}_0$, where $K:=\{(\lambda, j)\ |\ \lambda\in\Lambda_F, j=1, \cdots, m_\lambda\}$ and $\bm{w}_0:=(1)_{k\in K}$.
	We define $\ket{\Psi'_0}, \ket{\Psi'_1}\in\mathcal{H}^{\otimes L'}$ by $\ket{\Psi'_i}:=\ket{\Phi(|G|\bm{n}'_i)}$ with $\ket{\Phi(\bm{n})}$ defined by Eq.~\eqref{eq:determinantstate}, where $L':=|G|D\bm{n}'_0\cdot\bm{w}_0=|G|D\bm{n}'_1\cdot\bm{w}_0$.
	For any $g\in G$, $F(g)$ can be regarded as $\Omega$ in Eq.~\eqref{eq:Omega_decomp} in the case where $\zeta_k=F_\lambda(g)$.
	From Lemma~\ref{lem:eigenvectorcomposition}, $\ket{\Psi'_i}$ is an eigenstate of $F^{\otimes L'}(g)$ with the eigenvalue
	\begin{align}
		\prod_{k\in K}(\mathrm{det}(F_\lambda(g)))^{|G|n'_{ik}\frac{D}{r_\lambda}}
		=\prod_{k\in K}(\mathrm{det}(F_\lambda(g^{|G|})))^{n'_{ik}\frac{D}{r_\lambda}}
		=\prod_{k\in K}(\mathrm{det}(F_\lambda(1)))^{n'_{ik}\frac{D}{r_\lambda}}
		=\prod_{k\in K}(\mathrm{det}(I))^{n'_{ik}\frac{D}{r_\lambda}}
		=1.
	\end{align}
	In the same way as Eq.~\eqref{eq:H_L'_eigenvalue}, we obtain
	\begin{align}
		H^{(L')}\ket{\Psi'_l}=|G|D\bm{n}'_l\cdot\bm{E}\ket{\Psi'_l}, \label{eq:finite_H_L'_eigenvalue}
	\end{align}	
	where $\bm{E}:=(E_k)_{k\in K}$, $E_k:=E_{\lambda j}$ for $k=(\lambda, j)\in K$, and $E_{\lambda j}$ is defined in Eq.~\eqref{eq:Hamiltoniandecomp2}.
	Since $\rho$ satisfies Condition~\ref{conditionA2}, there exists $\beta\in[0, \infty)$ such that
	\begin{align}
		-\log(\braket{\Psi'_1|\rho^{\otimes L'}|\Psi'_1})+\log(\braket{\Psi'_0|\rho^{\otimes L'}|\Psi'_0})=\beta(\braket{\Psi'_1|H^{(L')}|\Psi'_1}-\braket{\Psi'_0|H^{(L')}|\Psi'_0}). \label{eq:finite_virtualtemperature}
	\end{align}
	In the same way as Eqs.~\eqref{eq:rho_L'_eigenvalue} and \eqref{eq:log_rho_L'_eigenvalue}, we get
	\begin{align}
		&\braket{\Psi'_i|\rho^{\otimes L'}|\Psi'_i}=\prod_{k\in K} p_k^{|G|Dn'_{ik}},\\
		&-\log(\braket{\Psi'_i|\rho^{\otimes L'}|\Psi'_i})=D|G| \bm{n}'_i\cdot\bm{s}, \label{eq:finite_log_rho_L'_eigenvalue}
	\end{align}
	where $p_k:=p_{\lambda j}$, $\bm{s}:=(s_k)_{k\in K}$, $s_k:=s_{\lambda j}$, and $s_{\lambda j}:=-\log(p_{\lambda j})$. 
	By substituting Eqs.~\eqref{eq:finite_H_L'_eigenvalue} and \eqref{eq:finite_log_rho_L'_eigenvalue} into Eq.~\eqref{eq:finite_virtualtemperature}, we get
	\begin{align}
		|G|D\bm{n}'_0\cdot\bm{s}-|G|D\bm{n}'_1\cdot\bm{s}=\beta(|G|D\bm{n}'_0\cdot\bm{E}-|G|D\bm{n}'_1\cdot\bm{E}).
	\end{align}
	This implies that $\bm{n}'_0\cdot\bm{t}=\bm{n}'_1\cdot\bm{t}$, where $\bm{t}:=\bm{s}-\beta\bm{E}$.
	Since this holds for any $\bm{n}'_0,\ \bm{n}'_1\in\mathbb{N}^K$ satisfying $\bm{n}'_0\cdot\bm{w}_0=\bm{n}'_1\cdot\bm{w}_0$, from Lemma~\ref{lem:latticeduality}, $\bm{t}$ can be written as $\bm{t}=\mu\bm{w}_0$ with some $\mu\in\mathbb{R}$.
	Then, for any $\lambda\in\Lambda_F$ and $j=1, \cdots, m_\lambda$, we obtain $s_{\lambda j}-\beta E_{\lambda j}=\mu$.
	Therefore, we get
	\begin{align}
		-\log(\rho)-\beta H=&\sum_{\lambda\in\Lambda_F}\sum_{j=1}^{m_\lambda} (s_{\lambda j}-\beta E_{\lambda j})\iota_\lambda(I_{\mathcal{R}_\lambda}\otimes\ket{\psi_{\lambda j}}\bra{\psi_{\lambda j}})\iota_\lambda^\dag \nonumber\\
		=&\mu\sum_{\lambda\in\Lambda_F} \iota_\lambda\left(I_{\mathcal{R}_\lambda}\otimes\sum_{j=1}^{m_\lambda} \ket{\psi_{\lambda j}}\bra{\psi_{\lambda j}}\right)\iota_\lambda^\dag \nonumber\\
		=&\mu\sum_{\lambda\in\Lambda_F} \iota_\lambda\iota_\lambda^\dag \nonumber\\
		=&\mu I.
	\end{align}	
	This implies that 
	\begin{align}
		\rho=e^{-\mu I-\beta H}=\frac{1}{e^\mu}e^{-\beta H}.
	\end{align}
	Since $\mathrm{tr}(\rho)=1$, we obtain $e^\mu=\mathrm{tr}(e^{-\beta H})=Z$ and thus $\rho=\frac{1}{Z}e^{-\beta H}$.
\end{proof}

	Proposition~\ref{prop:finiteGFworkextop} states that we can extract positive work from multiple copies of a state other than the Gibbs ensemble by some operation that respects symmetry described by a finite cyclic group or a dihedral group.
	The proof is parallel to that of Proposition~\ref{prop:GFworkextop}, except that $\ket{\Psi_i}$ exists for a Hamiltonian that is not trivial.

\begin{proposition} \label{prop:finiteGFworkextop}
	Let $G$ be a finite cyclic group or a dihedral group with a unitary representation $F$ acting on $\mathcal{H}$, $H\in\mathcal{B}^\mathrm{H}(\mathcal{H})$ commute with $F(G)$ and be not trivial, and $\rho\in\mathcal{B}^{++}(\mathcal{H})$ satisfy $\mathrm{tr}(\rho)=1$.
	If $\rho$ cannot be written as $\rho=\frac{1}{Z}e^{-\beta H}$ with $\beta\in[0, \infty)$ and $Z:=\mathrm{tr}(e^{-\beta H})$, there exist $N\in\mathbb{N}$ and $U\in\mathcal{U}_{G, F^{\otimes N}}(\mathcal{H}^{\otimes N})$ that satisfy $[U^\dag H^{(N)}U, H^{(N)}]=0$ and $W(\rho^{\otimes N}, H^{(N)}, U)>0$.
\end{proposition}

\begin{proof}
	Since $H$ is not trivial, from Lemma~\ref{lem:Psiexistenceforfinite}, we can take $L\in\mathbb{N}$ and a pair of simultaneous eigenstates $\ket{\Psi_0}, \ket{\Psi_1}\in\mathcal{H}^{\otimes L}$ of $F^{\otimes L}(G)$ and $H^{(L)}$ that satisfies Eq.~\eqref{eq:Psi_condition}.
	Suppose that $\rho$ cannot be written as $\rho=\frac{1}{Z}e^{-\beta H}$ with $\beta\in[0, \infty)$ and $Z:=\mathrm{tr}(e^{-\beta H})$.
	From Proposition~\ref{prop:finiteGEderivation}, Conditions~\ref{conditionA1} and \ref{conditionA2} cannot be simultaneously satisfied for $\rho$.
	We can construct a symmetry-respecting unitary operator $U\in\mathcal{U}_{G, F^{\otimes N}}(\mathcal{H}^{\otimes N})$ that satisfy $[U^\dag H^{(N)}U, H^{(N)}]=0$ and $W(\rho^{\otimes N}, H^{(N)}, U)>0$ for some $N\in\mathbb{N}$ in the same way as in Proposition~\ref{prop:GFworkextop}.
\end{proof}

We prove Theorem~\ref{thm:finiteGFHCP} from Proposition~\ref{prop:finiteGFworkextop}.
It is obvious that we cannot extract positive work from the Gibbs ensemble by any symmetry-respecting operations.\\\\
\textit{Proof of Theorem~\ref{thm:finiteGFHCP}.}
	First, suppose that $\rho$ cannot be written as $\rho=\frac{1}{Z}e^{-\beta H}$ with $\beta\in[0, \infty)$ and $Z:=\mathrm{tr}(e^{-\beta H})$.
	From Proposition~\ref{prop:finiteGFworkextop}, there exist $N\in\mathbb{N}$ and $U\in\mathcal{U}_{G, F^{\otimes N}}(\mathcal{H}^{\otimes N})$ such that $W(\rho^{\otimes N}, H^{(N)}, U)>0$.
	This implies that $\rho$ is not $(G, F)$-completely passive w.r.t. $H$.
	
	Next, suppose that $\rho$ can be written as $\rho=\frac{1}{Z}e^{-\beta H}$ with some $\beta\in[0, \infty)$.
	Then, $\rho$ is completely passive in the ordinary sense.
	Thus for any $N\in\mathbb{N}$, $\mathcal{W}_{H^{(N)}}(\rho^{\otimes N})=0$. 
	Since $\mathcal{U}_{G, F^{\otimes N}}(\mathcal{H}^{\otimes N})\subset\mathcal{U}(\mathcal{H}^{\otimes N})$, we get
	\begin{align}
		\mathcal{W}_{G, F^{\otimes N}, H^{(N)}}(\rho^{\otimes N})=&\max_{U\in\mathcal{U}_{G, F^{\otimes N}}(\mathcal{H}^{\otimes N})}W(\rho^{\otimes N}, H^{(N)}, U) \nonumber\\
		\leq&\max_{U\in\mathcal{U}(\mathcal{H}^{\otimes N})}W(\rho^{\otimes N}, H^{(N)}, U) \nonumber\\
		=&\mathcal{W}_{H^{(N)}}(\rho^{\otimes N}) \nonumber\\
		=&0.
	\end{align}
	This implies that $\rho$ is $(G, F)$-completely passive w.r.t. $H$. \hspace{\fill} $\Box$\\

\subsection{Proof of Theorem~\ref{thm:THCP}}

We can prove Theorem~\ref{thm:THCP} almost in parallel to the case of $G=\{1\}$ in Theorem~\ref{thm:GFHCP}. 
In order to prove Theorem~\ref{thm:THCP}, we introduce the following two conditions instead of Conditions~\ref{conditionA1} and \ref{conditionA2}.

\begin{enumerate}[1]
\renewcommand{\theenumi}{B\arabic{enumi}}
 \item $\rho$ commutes with $\Omega$, for any $\Omega\in\mathcal{B}^\mathrm{H}(\mathcal{H})$ that commutes with $\mathcal{T}$ and $H$. \label{conditionB1}
 \item There exists $\beta\in[0, \infty)$ that satisfies the following: $-\log(\braket{\Psi'_1|\rho^{\otimes L'}|\Psi'_1})+\log(\braket{\Psi'_0|\rho^{\otimes L'}|\Psi'_0})=\beta(\braket{\Psi'_1|H^{(L')}|\Psi'_1}-\braket{\Psi'_0|H^{(L')}|\Psi'_0})$ holds, for any $L'\in\mathbb{N}$ and any pair of simultaneous eigenstates $\ket{\Psi'_0}, \ket{\Psi'_1}\in\mathcal{H}^{\otimes L'}$ of $H^{(L')}$ with $\mathcal{T}\ket{\Psi'_i}=\ket{\Psi'_i}$. \label{conditionB2}
\end{enumerate}

	Proposition~\ref{prop:TGEderivation} states that if a state $\rho$ satisfies Conditions~\ref{conditionB1} and \ref{conditionB2}, $\rho$ is the Gibbs ensemble at positive temperature.
	This corresponds to Propositions~\ref{prop:preGGEderivation} and \ref{prop:GGEderivation} in the case of Lie group symmetry.

\begin{proposition} \label{prop:TGEderivation}
	Let $H\in\mathcal{B}^\mathrm{H}(\mathcal{H})$ commute with $\mathcal{T}$ and $\rho\in\mathcal{B}^{++}(\mathcal{H})$ satisfy $\mathrm{tr}(\rho)=1$.
	If $\rho$ satisfies Conditions~\ref{conditionB1} and \ref{conditionB2}, $\rho$ can be written as $\rho=\frac{1}{Z}e^{-\beta H}$ with some $\beta\in[0, \infty)$, where $Z:=\mathrm{tr}(e^{-\beta H})$.
\end{proposition}

\begin{proof}
	Suppose that $\rho$ satisfies Conditions~\ref{conditionB1} and \ref{conditionB2}. 
	Since $H$ commutes with $\mathcal{T}$, $H$ can be diagonalized as 
	\begin{align}
		H=\sum_{j=1}^d E_j\ket{\psi_j}\bra{\psi_j} \label{eq:H_diagonalization}
	\end{align}
	with some $E_j\in\mathbb{R}$ and $\ket{\psi_j}\in\mathcal{H}$ satisfying $\mathcal{T}\ket{\psi_j}=\ket{\psi_j}$.
	Since $\ket{\psi_j}\bra{\psi_j}$ commutes with $\mathcal{T}$ and $H$, from Condition~\ref{conditionB1}, $\rho$ commutes with $\ket{\psi_j}\bra{\psi_j}$.
	Thus $\rho$ can be diagonalized as $\rho=\sum_{j=1}^d p_j\ket{\psi_j}\bra{\psi_j}$ with some $p_j\in(0, 1)$.

	Take arbitrary $\bm{n}'_0=(n'_{0j})_{j=1}^d, \bm{n}'_1=(n'_{1j})_{j=1}^d\in\mathbb{N}^d$ satisfying $\bm{n}'_0\cdot\bm{w}_0=\bm{n}'_1\cdot\bm{w}_0$, where $\bm{w}_0:=(1)_{j=1}^d$.
	We define $\ket{\Psi'_0}, \ket{\Psi'_1}\in\mathcal{H}^{\otimes L'}$ by 
	\begin{align}
		\ket{\Psi'_i}:=\bigotimes_{j=1}^d \ket{\psi_j}^{\otimes n'_{ij}}, \label{eq:Psi_for_T}
	 \end{align}
	where $L':=\bm{n}'_0\cdot\bm{w}_0=\bm{n}'_1\cdot\bm{w}_0$.
	Then $\ket{\Psi'_i}$ satisfies 
	\begin{align}
	&\mathcal{T}\ket{\Psi'_i}=\ket{\Psi'_i}, \\
	&H^{(L')}\ket{\Psi'_i}=\bm{n}'_i\cdot\bm{E}\ket{\Psi'_i}, \label{eq:T_H_L'_eigenvalue}
	\end{align}
	where $\bm{E}:=(E_j)_{j=1}^d$.  
	Since $\rho$ satisfies Condition~\ref{conditionB2}, there exists $\beta\in[0, \infty)$ such that
	\begin{align}
		-\log(\braket{\Psi'_1|\rho^{\otimes L'}|\Psi'_1})+\log(\braket{\Psi'_0|\rho^{\otimes L'}|\Psi'_0})=\beta(\braket{\Psi'_1|H^{(L')}|\Psi'_1}-\braket{\Psi'_0|H^{(L')}|\Psi'_0}). \label{eq:virtual_temperature_for_T}
	\end{align}
	From Eq.~\eqref{eq:Psi_for_T},
	\begin{align}
		-\log(\braket{\Psi'_i|\rho^{\otimes L'}|\Psi'_i})=-\log\left(\prod_{j=1}^d \braket{\psi_j|\rho|\psi_j}^{n'_{ij}}\right)=-\sum_{j=1}^d \log\left(p_j^{n'_{ij}}\right)=\sum_{j=1}^d n'_{ij}(-\log(p_j))=\bm{n}'_i\cdot\bm{s}, \label{eq:T_log_rho_L'_eigenvalue}
	\end{align}
	where $\bm{s}:=(s_j)_{j=1}^d$ and $s_j:=-\log(p_j)$. 
	By substituting Eqs.~\eqref{eq:T_H_L'_eigenvalue}~and~\eqref{eq:T_log_rho_L'_eigenvalue} into Eq.~\eqref{eq:virtual_temperature_for_T}, we get  
	\begin{align}
		\bm{n}'_0\cdot\bm{s}-\bm{n}'_1\cdot\bm{s}=\beta(\bm{n}'_0\cdot\bm{E}-\bm{n}'_1\cdot\bm{E}).
	\end{align}
	This implies that $\bm{n}'_0\cdot\bm{t}=\bm{n}'_1\cdot\bm{t}$, where $\bm{t}:=\bm{s}-\beta\bm{E}$.
	Since this holds for all $\bm{n}'_0,\ \bm{n}'_1\in\mathbb{N}^d$ satisfying $\bm{n}'_0\cdot\bm{w}_0=\bm{n}'_1\cdot\bm{w}_0$,  $\bm{t}$ can be written as $\bm{t}=\mu\bm{w}_0$ with some $\mu\in\mathbb{R}$ from Lemma~\ref{lem:latticeduality}.
	Then, for any $j=1, \cdots, d$, we obtain $s_j-\beta E_j=\mu$.
	Therefore, we get
	\begin{align}
		-\log(\rho)-\beta H=\sum_{j=1}^d (s_j-\beta E_j)\ket{\psi_j}\bra{\psi_j}=\sum_{j=1}^d \mu \ket{\psi_j}\bra{\psi_j}=\mu I.
	\end{align}
	This implies that 
	\begin{align}
		\rho=e^{-\mu I-\beta H}=\frac{1}{e^\mu}e^{-\beta H}.
	\end{align}
	Since $\mathrm{tr}(\rho)=1$, we obtain $e^\mu=\mathrm{tr}(e^{-\beta H})=Z$ and thus $\rho=\frac{1}{Z}e^{-\beta H}$.
\end{proof}

	Proposition~\ref{prop:Tworkextop} states that we can extract positive work from multiple copies of a state other than the Gibbs ensemble by some time-reversal symmetry-respecting operation.
	The proof is parallel to that of Proposition~\ref{prop:GFworkextop} except for the construction of $\ket{\Psi_i}$.

\begin{proposition} \label{prop:Tworkextop}
	Let $H\in\mathcal{B}^\mathrm{H}(\mathcal{H})$ commute with $\mathcal{T}$ and be not trivial, and $\rho\in\mathcal{B}^{++}(\mathcal{H})$ satisfy $\mathrm{tr}(\rho)=1$. 
	If $\rho$ cannot be written as $\rho=\frac{1}{Z}e^{-\beta H}$ with $\beta\in[0, \infty)$ and $Z:=\mathrm{tr}(e^{-\beta H})$, there exist $N\in\mathbb{N}$ and $U\in\mathcal{U}_\mathcal{T}(\mathcal{H}^{\otimes N})$ that satisfy $[U^\dag H^{(N)}U, H^{(N)}]=0$ and $W(\rho^{\otimes N}, H^{(N)}, U)>0$.
\end{proposition}

\begin{proof}
	Let $H$ be diagonalized as Eq.~\eqref{eq:H_diagonalization}.
	Since $H$ is not trivial, we can take $j_0, j_1\in\{1, \cdots, d\}$ satisfying $E_{j_0}<E_{j_1}$.
	We define $\ket{\Psi_0}, \ket{\Psi_1}\in\mathcal{H}$ as $\ket{\Psi_i}:=\ket{\psi_{j_i}}$.
	Then, $\ket{\Psi_i}$ satisfies $\mathcal{T}\ket{\Psi_i}=\ket{\Psi_i}$ and $H\ket{\Psi_i}=E_{j_i}\ket{\Psi_i}$.
	Suppose that $\rho$ cannot be written as $\rho=\frac{1}{Z}e^{-\beta H}$ with $\beta\in[0, \infty)$ and $Z:=\mathrm{tr}(e^{-\beta H})$.
	From Proposition~\ref{prop:TGEderivation}, Conditions~\ref{conditionB1} and \ref{conditionB2} cannot be simultaneously satisfied for $\rho$.
	We can construct a symmetry-respecting unitary operator $U\in\mathcal{U}_\mathcal{T}(\mathcal{H}^{\otimes N})$ that satisfies $[U^\dag H^{(N)}U, H^{(N)}]=0$ and $W(\rho^{\otimes N}, H^{(N)}, U)>0$ for some $N\in\mathbb{N}$ in the same way as in Proposition~\ref{prop:GFworkextop}.	
\end{proof}

We prove Theorem~\ref{thm:THCP} from Proposition~\ref{prop:Tworkextop}.\\\\
\textit{Proof of Theorem~\ref{thm:THCP}.}
	First, suppose that $\rho$ cannot be written as $\rho=\frac{1}{Z}e^{-\beta H}$ with $\beta\in[0, \infty)$ and $Z:=\mathrm{tr}(e^{-\beta H})$.
	From Proposition~\ref{prop:Tworkextop}, there exist $N\in\mathbb{N}$ and $U\in\mathcal{U}_\mathcal{T}(\mathcal{H}^{\otimes N})$ such that $W(\rho^{\otimes N}, H^{(N)}, U)>0$.
	This implies that $\rho$ is not $\mathcal{T}$-completely passive w.r.t. $H$.
	
	Next, suppose that $\rho$ can be written as $\rho=\frac{1}{Z}e^{-\beta H}$ with some $\beta\in[0, \infty)$.
	Then $\rho$ is completely passive in the ordinary sense.
	Thus for any $N\in\mathbb{N}$, $\mathcal{W}_{H^{(N)}}(\rho^{\otimes N})=0$. 
	Since $\mathcal{U}_\mathcal{T}(\mathcal{H}^{\otimes N})\subset\mathcal{U}(\mathcal{H}^{\otimes N})$, we get
	\begin{align}
		\mathcal{W}_{\mathcal{T}, H^{(N)}}(\rho^{\otimes N})=&\max_{U\in\mathcal{U}_\mathcal{T}(\mathcal{H}^{\otimes N})}W(\rho^{\otimes N}, H^{(N)}, U) \nonumber\\
		\leq&\max_{U\in\mathcal{U}(\mathcal{H}^{\otimes N})}W(\rho^{\otimes N}, H^{(N)}, U) \nonumber\\
		=&\mathcal{W}_{H^{(N)}}(\rho^{\otimes N}) \nonumber\\
		=&0.
	\end{align}
	This implies that $\rho$ is $\mathcal{T}$-completely passive w.r.t. $H$. \hspace{\fill} $\Box$\\

\section{Setup with a quantum work storage}
\label{sec:work_storage}

In this section, we consider a setup that explicitly includes a quantum work storage.
We adopt as a work storage a system whose Hamiltonian is the position operator $x$ for convenience (see also Sec.~III of the main text~\cite{Main}). 
Such a work storage is introduced in Ref.~\cite{Malabarba2015}.
We call this setup the \textit{WS setup} in the rest of this paper.
We define symmetry-protected (complete) passivity in the WS setup and prove Theorems~\ref{thm:GFHP1} to \ref{thm:THCP1}, which are the counterparts of Theorems~\ref{thm:GFHP} to \ref{thm:THCP}.
We denote by $\mathcal{H}_\mathrm{W}$ the Hilbert space of the work storage.

\subsection{Formal definitions}

First, we formulate work extraction in the WS setup. 
In this setup, the extracted work is also equivalent to the difference between the energy expectation values of the states of the system of interest before and after an operation. 

\begin{definition} [WS-extracted work] \label{def:WSextwork}
	Let $\rho\in\mathcal{B}^+(\mathcal{H})$, $H\in\mathcal{B}^\mathrm{H}(\mathcal{H})$, $\rho_\mathrm{W}\in\mathcal{B}^+(\mathcal{H_\mathrm{W}})$ and $V\in\mathcal{U}(\mathcal{H}\otimes\mathcal{H}_\mathrm{W})$. 
	The WS-extracted work from a state $\rho$ of the system of interest under a Hamiltonian $H$ with a state $\rho_\mathrm{W}$ of the work storage by the action of $V\in\mathcal{U}(\mathcal{H}\otimes\mathcal{H}_\mathrm{W})$ is defined as 
	\begin{align}
		W^\mathrm{WS}(\rho, H, \rho_\mathrm{W}, V):=\mathrm{tr}((\rho\otimes\rho_\mathrm{W}) (H\otimes I_{\mathcal{H}_\mathrm{W}}))-\mathrm{tr}((\rho\otimes\rho_\mathrm{W}) V^\dag (H\otimes I_{\mathcal{H}_\mathrm{W}})V).
	\end{align}
\end{definition}

In the WS setup, we consider the setup where all energy-preserving unitary operators are allowed.
We call such unitary operators WS-operators, which are defined as follows.

\begin{definition} [WS-operator] \label{def:WSoperator}
	Let $H\in\mathcal{B}^\mathrm{H}(\mathcal{H})$ be the Hamiltonian of the system of interest, and $x$ and $p$ be the position and momentum operators of the work storage.
	An operator $V\in\mathcal{B}(\mathcal{H}\otimes\mathcal{H}_\mathrm{W})$ is a WS-operator, if $V$ satisfies $V\in\mathcal{U}(\mathcal{H}\otimes\mathcal{H}_\mathrm{W})$, $[V, H\otimes I_{\mathcal{H}_\mathrm{W}}+I_\mathcal{H}\otimes x]=0$ and $[V, I_{\mathcal{H}}\otimes p]=0$.
	We define $\mathcal{U}^\mathrm{WS}(\mathcal{H}, \mathcal{H}_\mathrm{W})$ as the set of all WS-operators on $\mathcal{H}\otimes\mathcal{H}_\mathrm{W}$.
\end{definition}

The condition $[V, H\otimes I_{\mathcal{H}_\mathrm{W}}+I_\mathcal{H}\otimes x]=0$ means the conservation of the total energy of the system of interest and the work storage.
This reflects the first law of thermodynamics in the quantum setup in a strong sense.
	When $V$ is a WS-operator, the WS-extracted work defined in Definition~\ref{def:WSextwork} is also written as
	\begin{align}
		W^\mathrm{WS}(\rho, H, \rho_\mathrm{W}, V)=\mathrm{tr}((\rho\otimes\rho_\mathrm{W}) V^\dag (I_\mathcal{H}\otimes x)V)-\mathrm{tr}((\rho\otimes\rho_\mathrm{W}) (I_\mathcal{H}\otimes x)).
	\end{align}	
The condition $[V, I_{\mathcal{H}_\mathrm{W}}\otimes p]=0$ reflects the invariance under energy translation of the work storage.

When we consider operators that respect group symmetry in the WS setup, we define them as unitary operators that commute with the representation of all elements of the group.
Here we denote by $G$ a group with a unitary representation $F$ acting on $\mathcal{H}$.
On the other hand, when we consider time-reversal symmetry, we suppose that the complex conjugation operator $\mathcal{T}$ acts on the total system including the work storage.
This is a proper way to construct an anti-unitary operator, because if we take the complex conjugation acting only on the system of interest, the  operation becomes neither unitary nor anti-unitary on the total system.

\begin{definition} [WS-$(G, F)$-respecting operator, WS-$\mathcal{T}$-respecting operator]
	An operator $V\in\mathcal{B}(\mathcal{H}\otimes \mathcal{H}_\mathrm{W})$ is a WS-$(G, F)$-respecting (resp. WS-$\mathcal{T}$-respecting) operator, if $V$ satisfies $V\in\mathcal{U}^\mathrm{WS}(\mathcal{H}, \mathcal{H}_\mathrm{W})$ and commutes with $F(G)\otimes I_{\mathcal{H}_\mathrm{W}}$ (resp. $\mathcal{T}$). 
	We define $\mathcal{U}_{G, F}^\mathrm{WS}(\mathcal{H}, \mathcal{H}_\mathrm{W})$ (resp. $\mathcal{U}_\mathcal{T}^\mathrm{WS}(\mathcal{H},\mathcal{H}_\mathrm{W})$) as the set of all WS-$(G, F)$-respecting (resp. WS-$\mathcal{T}$-respecting) operators on $\mathcal{H}\otimes\mathcal{H}_\mathrm{W}$.
\end{definition}

We define the (symmetry-respecting) WS-ergotropy of a state as the maximal extracted work from the state by the action of (symmetry-respecting) WS-operators. 
Note that this value is not necessarily independent of the state of the work storage.

\begin{definition} [WS-$\rho_\mathrm{W}$-ergotropy, WS-$(G, F, \rho_\mathrm{W})$-ergotropy, WS-$(\mathcal{T}, \rho_\mathrm{W})$-ergotropy] \label{def:WSergotropy}
	Let $H\in\mathcal{B}^\mathrm{H}(\mathcal{H})$, $\rho\in\mathcal{B}^+(\mathcal{H})$ and $\rho_\mathrm{W}\in\mathcal{B}^+(\mathcal{H}_\mathrm{W})$. 
	WS-$\rho_\mathrm{W}$-ergotropy $\mathcal{W}_{H, \rho_\mathrm{W}}^\mathrm{WS}(\rho)$ (resp. WS-$(G, F, \rho_\mathrm{W})$-ergotropy $\mathcal{W}_{G, F, H, \rho_\mathrm{W}}^\mathrm{WS}(\rho)$ or WS-$(\mathcal{T}, \rho_\mathrm{W})$-ergotropy $\mathcal{W}_{\mathcal{T}, H, \rho_\mathrm{W}}^\mathrm{WS}(\rho)$) of a state $\rho$ of the system of interest under a Hamiltonian $H$ with a state $\rho_\mathrm{W}$ of the work storage is defined as the maximal extracted work from $\rho$ under the Hamiltonian $H$ with the work storage state $\rho_\mathrm{W}$ by the action of operators in $\mathcal{U}^\mathrm{WS}(\mathcal{H}, \mathcal{H}_\mathrm{W})$ (resp. $\mathcal{U}_{G, F}^\mathrm{WS}(\mathcal{H}, \mathcal{H}_\mathrm{W})$ or $\mathcal{U}_\mathcal{T}^\mathrm{WS}(\mathcal{H}, \mathcal{H}_\mathrm{W})$), i.e.,
	\begin{align}
		&\mathcal{W}_{H, \rho_\mathrm{W}}^\mathrm{WS}(\rho):=\max_{V\in\mathcal{U}^\mathrm{WS}(\mathcal{H}, \mathcal{H}_\mathrm{W})} W^\mathrm{WS}(\rho, H, \rho_\mathrm{W}, V),\\
		&\mathcal{W}_{G, F, H, \rho_\mathrm{W}}^\mathrm{WS}(\rho):=\max_{V\in\mathcal{U}_{G, F}^\mathrm{WS}(\mathcal{H}, \mathcal{H}_\mathrm{W})} W^\mathrm{WS}(\rho, H, \rho_\mathrm{W}, V),\\
		&\mathcal{W}_{\mathcal{T}, H, \rho_\mathrm{W}}^\mathrm{WS}(\rho):=\max_{V\in\mathcal{U}_\mathcal{T}^\mathrm{WS}(\mathcal{H}, \mathcal{H}_\mathrm{W})} W^\mathrm{WS}(\rho, H, \rho_\mathrm{W}, V).
	\end{align}
\end{definition}

We define WS-(symmetry-protected) passive states as the state from which no positive work can be extracted by the action of WS-(symmetry-respecting) operators.

\begin{definition} [WS-$\rho_\mathrm{W}$-passivity, WS-$(G, F, \rho_\mathrm{W})$-passivity, WS-$(\mathcal{T}, \rho_\mathrm{W})$-passivity]
	Let $H\in\mathcal{B}^\mathrm{H}(\mathcal{H})$, $\rho\in\mathcal{B}^+(\mathcal{H})$ and $\rho_\mathrm{W}\in\mathcal{B}^+(\mathcal{H}_\mathrm{W})$. 
	A state $\rho$ is WS-$\rho_\mathrm{W}$-passive (resp. WS-$(G, F, \rho_\mathrm{W})$-passive or WS-$(\mathcal{T}, \rho_\mathrm{W})$-passive) w.r.t. a Hamiltonian $H$ if $\mathcal{W}_H^\mathrm{WS}(\rho, \rho_\mathrm{W})=0$ (resp. $\mathcal{W}_{G, F, H}^\mathrm{WS}(\rho, \rho_\mathrm{W})=0$ or $\mathcal{W}_{\mathcal{T}, H}^\mathrm{WS}(\rho, \rho_\mathrm{W})=0$).
\end{definition}

WS-(symmetry-protected) completely passive states are defined as the states such that no positive work can be extracted from any number of copies of them.
Here we consider the situation where we have a single work storage and multiple copies of the system of interest.

\begin{definition} [WS-$\rho_\mathrm{W}$-complete passivity, WS-$(G, F, \rho_\mathrm{W})$-complete passivity, WS-$(\mathcal{T}, \rho_\mathrm{W})$-complete passivity]
	Let $H\in\mathcal{B}^\mathrm{H}(\mathcal{H})$, $\rho\in\mathcal{B}^+(\mathcal{H})$ and $\rho_\mathrm{W}\in\mathcal{B}^+(\mathcal{H}_\mathrm{W})$.  
	A state $\rho$ is WS-$\rho_\mathrm{W}$-completely passive (resp. WS-$(G, F, \rho_\mathrm{W})$-completely passive or WS-$(\mathcal{T}, \rho_\mathrm{W})$-completely passive) w.r.t. a Hamiltonian $H$ if $\rho^{\otimes N}$ is WS-$\rho_\mathrm{W}$-passive (resp. WS-$(G, F, \rho_\mathrm{W})$-passive or WS-$(\mathcal{T}, \rho_\mathrm{W})$-passive) w.r.t. $H^{(N)}$ for all $N\in\mathbb{N}$.
\end{definition}

\subsection{Main theorems}
In this subsection, we present the main theorems of this section, which are the counterparts of the theorems presented in Sec.~\ref{subsec:main_theorems}.

First, Theorem~\ref{thm:GFHP1} gives the necessary and sufficient condition for group symmetry-protected passivity in the WS setup.

\begin{theorem} \label{thm:GFHP1}
	Let $G$ be a group with a unitary representation $F$ acting on $\mathcal{H}$, $H\in\mathcal{B}^\mathrm{H}(\mathcal{H})$ commute with $F(G)$, $\rho\in\mathcal{B}^+(\mathcal{H})$ and $\rho_\mathrm{W}\in\mathcal{B}^+(\mathcal{H}_\mathrm{W})$.
	Then, a state $\rho$ is WS-$(G, F, \rho_\mathrm{W})$-passive w.r.t. a Hamiltonian $H$, if and only if $\tilde{\rho}_\lambda:=\mathrm{tr}_{\mathcal{R}_\lambda}(\iota_\lambda^\dag \mathcal{D}_{H, \rho_\mathrm{W}}(\rho) \iota_\lambda)$ is passive w.r.t. $H_\lambda$ for all $\lambda\in\Lambda_F$, where $\iota_\lambda$ is defined as Eq.~\eqref{eq:irreducibledecomposition}, the mapping $\mathcal{D}_{H, \rho_\mathrm{W}}: \mathcal{B}^+(\mathcal{H})\to\mathcal{B}^+(\mathcal{H})$ is defined as 
	\begin{align}
		\mathcal{D}_{H, \rho_\mathrm{W}}(\rho):=\int_{-\infty}^\infty dq\ {}_{p}\braket{q|\rho_\mathrm{W}|q}_p e^{-iqH}\rho e^{iqH}, \label{eq:mapDdef}
	\end{align}
	and $\ket{q}_p$ is a momentum eigenstate of the work storage with eigenvalue $q$ satisfying the normalizing condition ${}_p\braket{q|r}_p=\delta(q-r)$.
\end{theorem} 

	If $\rho_\mathrm{W}$ is the position eigenstate with eigenvalue 0, $\mathcal{D}_{H, \rho_\mathrm{W}}$ is the dephasing mapping in the energy eigenbasis, i.e., $\mathcal{D}_{H, \rho_\mathrm{W}}(\rho)=\sum_{E} \Pi_E \rho \Pi_E$, where $\Pi_E$ is the projection operator onto the energy eigenspace of $E$.
	On the other hand, if $\rho_\mathrm{W}$ is the momentum eigenstate with eigenvalue 0, $\mathcal{D}_{H, \rho_\mathrm{W}}$ is the identity mapping.

	Theorem~\ref{thm:THP1} gives the necessary and sufficient condition for time-reversal symmetry-protected passivity in the WS setup.

\begin{theorem} \label{thm:THP1}
	Let $H\in\mathcal{B}^\mathrm{H}(\mathcal{H})$ commute with $\mathcal{T}$, $\rho\in\mathcal{B}^+(\mathcal{H})$ and $\rho_\mathrm{W}\in\mathcal{B}^+(\mathcal{H}_\mathrm{W})$.
	Then, a state $\rho$ is WS-$(\mathcal{T}, \rho_\mathrm{W})$-passive w.r.t. a Hamiltonian $H$, if and only if $\mathcal{S}_\mathcal{T}(\mathcal{D}_{H, \rho_\mathrm{W}}(\rho))$ is passive w.r.t. $H$, where $\mathcal{S}_\mathcal{T}$ and $\mathcal{D}_{H, \rho_\mathrm{W}}$ are respectively defined as Eqs.~\eqref{eq:Tsymmetrizer} and \eqref{eq:mapDdef}. 
\end{theorem}

	In contrast to passivity, the condition for complete passivity in the WS setup is the same as that in the setup without the work storage.
	We stress that it is independent of the state of the work storage.
	Theorem~\ref{thm:GFHCP1} states that completely passive states in the WS setup protected by Lie group symmetry are only generalized Gibbs ensembles.

\begin{theorem} \label{thm:GFHCP1}
	Let $G$ be a connected compact Lie group with a faithful unitary representation $F$ acting on $\mathcal{H}$, $\mathfrak{g}$ be the Lie algebra of $G$ with the associated representation $f$, $H\in\mathcal{B}^\mathrm{H}(\mathcal{H})$ commute with $F(G)$ and be not $(G, F)$-trivial, and $\rho\in\mathcal{B}^{++}(\mathcal{H})$, $\rho_\mathrm{W}\in\mathcal{B}^+(\mathcal{H}_\mathrm{W})$ satisfy $\mathrm{tr}(\rho)=\mathrm{tr}(\rho_\mathrm{W})=1$.
	Then, for any initial state $\rho_W$ of the work storage, a state $\rho$ is WS-$(G, F, \rho_\mathrm{W})$-completely passive w.r.t. a Hamiltonian $H$, if and only if $\rho=\frac{1}{Z}e^{-\beta H-f(X)}$ with some $\beta\in[0, \infty)$ and $X\in\mathfrak{g}$, where $Z:=\mathrm{tr}(e^{-\beta H-f(X)})$.	
\end{theorem}

Theorem~\ref{thm:finiteGFHCP1} states that completely passive states in the WS setup protected by finite cyclic group symmetry or dihedral group symmetry are only Gibbs ensembles.

\begin{theorem} \label{thm:finiteGFHCP1}
	Let $G$ be a finite cyclic group or a dihedral group with a unitary representation $F$ acting on $\mathcal{H}$, $H\in\mathcal{B}^\mathrm{H}(\mathcal{H})$ commute with $F(G)$ and be not trivial, and $\rho\in\mathcal{B}^{++}(\mathcal{H})$, $\rho_\mathrm{W}\in\mathcal{B}^+(\mathcal{H}_\mathrm{W})$ satisfy $\mathrm{tr}(\rho)=\mathrm{tr}(\rho_\mathrm{W})=1$.
	Then, for any initial state $\rho_W$ of the work storage, a state $\rho$ is WS-$(G, F, \rho_\mathrm{W})$-completely passive w.r.t. a Hamiltonian $H$, if and only if $\rho$ can be written as $\rho=\frac{1}{Z}e^{-\beta H}$ with some $\beta\in[0, \infty)$, where $Z:=\mathrm{tr}(e^{-\beta H})$.
\end{theorem} 
	
Theorem~\ref{thm:THCP1} states that completely passive states in the WS setup protected by time-reversal symmetry are also only Gibbs ensembles.

\begin{theorem} \label{thm:THCP1}
	Let $H\in\mathcal{B}^\mathrm{H}(\mathcal{H})$ commute with $\mathcal{T}$ and be not trivial, and $\rho\in\mathcal{B}^{++}(\mathcal{H})$, $\rho_\mathrm{W}\in\mathcal{B}^+(\mathcal{H}_\mathrm{W})$ satisfy $\mathrm{tr}(\rho)=\mathrm{tr}(\rho_\mathrm{W})=1$.
	Then, for any initial state $\rho_W$ of the work storage, a state $\rho$ is WS-$(\mathcal{T}, \rho_\mathrm{W})$-completely passive w.r.t. a Hamiltonian $H$, if and only if $\rho$ can be written as $\rho=\frac{1}{Z}e^{-\beta H}$ with some $\beta\in[0, \infty)$, where $Z:=\mathrm{tr}(e^{-\beta H})$.
\end{theorem}

\subsection{Proofs of Theorems~\ref{thm:GFHP1}~and~\ref{thm:THP1}}

	On the basis of Ref.~\cite{Kitaev2004}, we introduce a mapping that is useful for the investigation of work extraction in the WS setup, which is the Kitaev construction discussed in Sec.~IIIB of the main text~\cite{Main} (Eq.~(11)).
	From Lemma~\ref{lem:SCFQcorrespondence}, we can define a mapping $\mathcal{C}$: $\mathcal{U}(\mathcal{H})\to\mathcal{U}^\mathrm{WS}(\mathcal{H}, \mathcal{H}_\mathrm{W})$ as
	\begin{align}
		\mathcal{C}(U):=\int_{-\infty}^\infty dq\ e^{\mathrm{i}qH}Ue^{-\mathrm{i}qH}\otimes\ket{q}_{p\ p}\bra{q}. \label{eq:SCFQcorrespondence}
	\end{align}
	We note that $\mathcal{C}(U)$ is equivalent to the one that appears in Ref.~\cite{Malabarba2015, Skrzypczyk2013}.
	Ref.~\cite{Skrzypczyk2013} states that $\mathcal{C}$ is bijective, but we prove it in Lemma~\ref{lem:SCFQbijection} for the sake of self-containedness of the proof.
	Therefore, it is sufficient to consider operators written as $\mathcal{C}(U)$ with some $U\in\mathcal{U}(\mathcal{H})$ in order to consider all the WS-operators.

	Proposition~\ref{prop:SCFQextractedwork} shows the correspondence given by $\mathcal{C}$ between the extracted work in the setups with and without the work storage.  
	Concretely, the extracted work from a state $\rho$ of the system of interest by the action of $\mathcal{C}(U)$ in the WS setup equals the extracted work from the state $\mathcal{D}_{H, \rho_\mathrm{W}}(\rho)$ by the action of $U$ in the setup without the work storage.

\begin{proposition} \label{prop:SCFQextractedwork}
	Let $H\in\mathcal{B}^\mathrm{H}(\mathcal{H})$, $U\in\mathcal{U}(\mathcal{H})$, $\rho\in\mathcal{B}^+(\mathcal{H})$ and $\rho_\mathrm{W}\in\mathcal{B}^+(\mathcal{H}_\mathrm{W})$.
	Then, $W^\mathrm{WS}(\rho, H, \rho_\mathrm{W}, \mathcal{C}(U))=W(\mathcal{D}_{H, \rho_\mathrm{W}}(\rho), H, U)$, where $\mathcal{C}$ and $\mathcal{D}_{H, \rho_\mathrm{W}}$ are respectively defined as Eqs.~\eqref{eq:SCFQcorrespondence} and \eqref{eq:mapDdef}.
\end{proposition}

\begin{proof}
	Since
	\begin{align}
		&\mathcal{C}(U)^\dagger(H\otimes I)\mathcal{C}(U) \nonumber\\
		=&\left(\int_{-\infty}^\infty dq\ e^{\mathrm{i}qH}U e^{-\mathrm{i}qH}\otimes\ket{q}_{p\ p}\bra{q}\right)^\dag(H\otimes I) \left(\int_{-\infty}^\infty dr\ e^{\mathrm{i}rH}Ue^{-\mathrm{i}rH}\ket{r}_{p\ p}\bra{r}\right) \nonumber\\
		=&\int_{-\infty}^\infty dq\ \int_{-\infty}^\infty dr\ e^{\mathrm{i}qH}U^\dagger e^{-\mathrm{i}qH}He^{\mathrm{i}rH}Ue^{-\mathrm{i}rH}\otimes\ket{q}_{p\ p}\braket{q|r}_{p\ p}\bra{r} \nonumber\\
		=&\int_{-\infty}^\infty dq\ e^{\mathrm{i}qH}U^\dagger HUe^{-\mathrm{i}qH}\otimes\ket{q}_{p\ p}\bra{q}, \label{eq:FQHeisenbergHamiltonian}
	\end{align}
	we get
	\begin{align}
		W^\mathrm{WS}(\rho, H, \rho_\mathrm{W}, \mathcal{C}(U))=&\mathrm{tr}((\rho\otimes\rho_\mathrm{W})(H\otimes I-\mathcal{C}(U)^\dagger(H\otimes I)\mathcal{C}(U)))
 \nonumber\\
		=&\mathrm{tr}\left(\int_{-\infty}^\infty dq\ \rho e^{\mathrm{i}qH}(H-U^\dagger HU)e^{-\mathrm{i}qH}\otimes\rho_\mathrm{W}\ket{q}_{p\ p}\bra{q}\right) \nonumber\\
		=&\int_{-\infty}^\infty dq\ {}_p\braket{q|\rho_\mathrm{W}|q}_p \mathrm{tr}(e^{-\mathrm{i}qH}\rho e^{-\mathrm{i}qH}(H-U^\dagger HU)) \nonumber\\
		=&\mathrm{tr}(\mathcal{D}_{H, \rho_\mathrm{W}}(\rho)(H-U^\dagger HU)) \nonumber\\
		=&W(\mathcal{D}_{H, \rho_\mathrm{W}}(\rho), H, U).
	\end{align}
\end{proof}

	From Proposition~\ref{prop:SCFQextractedwork}, we can prove that the symmetry-respecting ergotropy of a state $\rho$ in the WS setup equals the symmetry-respecting ergotropy of the state $\mathcal{D}_{H, \rho_\mathrm{W}}(\rho)$ in the setup without the work storage.
	Corollary~\ref{cor:GFSCFQergotropy} deals with group-symmetry-respecting ergotropy.

\begin{corollary} \label{cor:GFSCFQergotropy}
	Let $G$ be a group with a unitary representation $F$ acting on $\mathcal{H}$, $H\in\mathcal{B}^\mathrm{H}(\mathcal{H})$ commute with $F(G)$, $\rho\in\mathcal{B}^+(\mathcal{H})$ and $\rho_\mathrm{W}\in\mathcal{B}^+(\mathcal{H}_\mathrm{W})$.
	Then, $\mathcal{W}_{G, F, H, \rho_\mathrm{W}}^\mathrm{WS}(\rho)=\mathcal{W}_{G, F, H}(\mathcal{D}_{H, \rho_\mathrm{W}}(\rho))$, where $\mathcal{D}_{H, \rho_\mathrm{W}}$ is defined as Eq.~\eqref{eq:mapDdef}.
\end{corollary}

\begin{proof}
	From Lemma~\ref{lem:GFSCFQcorrespondence}, $\mathcal{C}(\mathcal{U}_{G, F}(\mathcal{H}))=\mathcal{U}_{G, F}^\mathrm{WS}(\mathcal{H}, \mathcal{H}_\mathrm{W})$.
	Thus, 
	\begin{align}
		\mathcal{W}_{G, F, H, \rho_\mathrm{W}}^\mathrm{WS}(\rho)=&\max_{V\in\mathcal{U}_{G, F}^\mathrm{WS}(\mathcal{H}, \mathcal{H}_\mathrm{W})} W^\mathrm{WS}(\rho, H, \rho_\mathrm{W}, V) \nonumber\\
		=&\max_{U\in\mathcal{U}_{G, F}(\mathcal{H})} W^\mathrm{WS}(\rho, H, \rho_\mathrm{W}, \mathcal{C}(U)) \nonumber\\
		=&\max_{U\in\mathcal{U}_{G, F}(\mathcal{H})} W(\mathcal{D}_{H, \rho_\mathrm{W}}(\rho), H, U) \nonumber\\
		=&\mathcal{W}_{G, F, H}(\mathcal{D}_{H, \rho_\mathrm{W}}(\rho)). \label{eq:WSGFergotropy}
	\end{align}
\end{proof}

	Since we know the correspondence between symmetry-respecting ergotropy in the setups with and without the work storage from Corollary~\ref{cor:GFSCFQergotropy}, we can derive the condition for symmetry-protected passivity in the WS setup.
	The proof of Theorem~\ref{thm:GFHP1} is now given as follows:\\\\
\textit{Proof of Theorem~\ref{thm:GFHP1}.}
	From Corollary~\ref{cor:GFSCFQergotropy}, $\mathcal{W}_{G, F, H, \rho_\mathrm{W}}^\mathrm{WS}(\rho)=0$ if and only if $\mathcal{W}_{G, F, H}(\mathcal{D}_{H, \rho_\mathrm{W}}(\rho))=0$.
	This means that $\rho$ is WS-$(G, F, \rho_\mathrm{W})$-passive w.r.t. $H$ if and only if $\mathcal{D}_{H, \rho_\mathrm{W}}(\rho)$ is $(G, F)$-passive w.r.t. $H$.
	From Theorem~\ref{thm:GFHP}, $\mathcal{D}_{H, \rho_\mathrm{W}}(\rho)$ is $(G, F)$-passive, if and only if $\tilde{\rho}_\lambda=\mathrm{tr}_{\mathcal{R}_\lambda}(\iota_\lambda^\dag \mathcal{D}_{H, \rho_\mathrm{W}}(\rho) \iota_\lambda)$ is passive w.r.t. $H_\lambda$ for all $\lambda\in\Lambda_F$. \hspace{\fill} $\Box$\\

	Corollary~\ref{cor:TSCFQergotropy} deals with time-reversal symmetry-respecting ergotropy.

\begin{corollary} \label{cor:TSCFQergotropy}
	Let $H\in\mathcal{B}^\mathrm{H}(\mathcal{H})$ commute with $\mathcal{T}$, $\rho\in\mathcal{B}^+(\mathcal{H})$ and $\rho_\mathrm{W}\in\mathcal{B}^+(\mathcal{H}_\mathrm{W})$.
	Then, $\mathcal{W}_{\mathcal{T}, H, \rho_\mathrm{W}}^\mathrm{WS}(\rho)=\mathcal{W}_{\mathcal{T}, H}(\mathcal{D}_{H, \rho_\mathrm{W}}(\rho))$, where $\mathcal{D}_{H, \rho_\mathrm{W}}$ is defined as Eq.~\eqref{eq:mapDdef}.
\end{corollary}

\begin{proof}
	From Lemma~\ref{lem:TSCFQcorrespondence}, $\mathcal{C}(\mathcal{U}_\mathcal{T}(\mathcal{H}))=\mathcal{U}_\mathcal{T}^\mathrm{WS}(\mathcal{H}, \mathcal{H}_\mathrm{W})$.
	Thus, in the same way as Eq.~\eqref{eq:WSGFergotropy}, we can prove that $\mathcal{W}_{\mathcal{T}, H, \rho_\mathrm{W}}^\mathrm{WS}(\rho)=\mathcal{W}_{\mathcal{T}, H}(\mathcal{D}_{H, \rho_\mathrm{W}}(\rho))$.
\end{proof}

	Theorem~\ref{thm:THP1} immediately follows from Corollary~\ref{cor:TSCFQergotropy}:\\\\
\textit{Proof of Theorem~\ref{thm:THP1}.}
    From Corollary~\ref{cor:TSCFQergotropy}, $\mathcal{W}_{\mathcal{T}, H, \rho_\mathrm{W}}^\mathrm{WS}(\rho)=0$ if and only if $\mathcal{W}_{\mathcal{T}, H}(\mathcal{D}_{H, \rho_\mathrm{W}}(\rho))=0$.
	This means that $\rho$ is WS-$(\mathcal{T}, \rho_\mathrm{W})$-passive w.r.t. $H$, if and only if $\mathcal{D}_{H, \rho_\mathrm{W}}(\rho)$ is $\mathcal{T}$-passive w.r.t. $H$.
	From Theorem~\ref{thm:THP}, $\mathcal{D}_{H, \rho_\mathrm{W}}(\rho)$ is $\mathcal{T}$-passive, if and only if $\mathcal{S}_\mathcal{T}(\mathcal{D}_{H, \rho_\mathrm{W}}(\rho))$ is passive w.r.t. $H$. \hspace{\fill} $\Box$\\

\subsection{Proofs of Theorems~\ref{thm:GFHCP1}, \ref{thm:finiteGFHCP1} and \ref{thm:THCP1}}

	Theorems~\ref{thm:GFHCP1}, \ref{thm:finiteGFHCP1} and \ref{thm:THCP1} are respectively the counterparts of Theorems~\ref{thm:GFHCP},~\ref{thm:finiteGFHCP}~and~\ref{thm:THCP}.
	In Propositions~\ref{prop:GFworkextop}, \ref{prop:finiteGFworkextop} and \ref{prop:Tworkextop}, we proved that we can extract positive work from $N$ copies of a state other than the (generalized) Gibbs ensemble by the action of a certain unitary operator $U$ satisfying $[U^\dag H^{(N)}U, H^{(N)}]=0$ for some $N\in\mathbb{N}$.
	In Proposition~\ref{prop:stateswapunitary}, we prove that for such a unitary operator $U$, the extracted work by the action of $\mathcal{C}(U)$ in the WS setup is independent of the state of the work storage, and equals the extracted work by the action of $U$ in the setup without the work storage.
	Therefore, we can construct a unitary operator that extracts positive work from multiple copies of a state other than the (generalized) Gibbs ensemble in the WS setup only by replacing $U$ with $\mathcal{C}(U)$ in Propositions~\ref{prop:GFworkextop}, \ref{prop:finiteGFworkextop} and \ref{prop:Tworkextop}.
	On the other hand, in Propositions~\ref{prop:GFSCFQergotropyinequality} and \ref{prop:TSCFQergotropyinequality}, we prove that the symmetry respecting ergotropy of a state is in general smaller in the WS setup than in the setup without the work storage. 
	This fact guarantees that the (generalized) Gibbs ensemble is also symmetry-protected completely passive in the WS setup.

	Proposition~\ref{prop:stateswapunitary}, which is referred to as Proposition~2 of the main text \cite{Main}, states that for a unitary operator $U$ that satisfies $[U^\dag HU, H]=0$, the same amount of work is extracted by the action of $\mathcal{C}(U)$ and $U$ in the setups with and without the work storage.

\begin{proposition} \label{prop:stateswapunitary}
	Let $H\in\mathcal{B}^\mathrm{H}(\mathcal{H})$, $\rho\in\mathcal{B}^+(\mathcal{H})$ and $U\in\mathcal{U}(\mathcal{H})$ satisfy $[U^\dag HU,H]=0$.
	Then, $W^\mathrm{WS}(\rho, H, \rho_\mathrm{W}, \mathcal{C}(U))=W(\rho, H, U)$ holds for all $\rho_\mathrm{W}\in\mathcal{B}^+(\mathcal{H}_\mathrm{W})$ satisfying $\mathrm{tr}(\rho_\mathrm{W})=1$.
\end{proposition}

	It is proved in Ref.~\cite{Malabarba2015} that if a unitary operator $U$ only permutes energy eigenstates, the extracted work by the action of $\mathcal{C}(U)$ and $U$ in the setups with and without the work storage are the same.
	Since we find that $[U^\dag HU, H]=0$ is equivalent to the fact that $U$ only permutes energy eigenstates, we can prove Proposition~\ref{prop:stateswapunitary}.
    However, we can also prove it without using the permutation property of $U$ as follows.

\begin{proof}
	Since $U$ satisfies $[U^\dagger HU,H]=0$, from Eq.~\eqref{eq:FQHeisenbergHamiltonian},
	\begin{align}
		\mathcal{C}(U)^\dagger(H\otimes I)\mathcal{C}(U)=U^\dag HU\otimes\int_{-\infty}^\infty dq\ \ket{q}_{p\ p}\bra{q}=U^\dag HU\otimes I.
	\end{align}
	Therefore, for any $\rho_\mathrm{W}\in\mathcal{B}^+(\mathcal{H}_\mathrm{W})$ satisfying $\mathrm{tr}(\rho_\mathrm{W})=1$, 
	\begin{align}
		W^\mathrm{WS}(\rho, H, \rho_\mathrm{W}, \mathcal{C}(U))=&\mathrm{tr}((\rho\otimes\rho_\mathrm{W})[H\otimes I-\mathcal{C}(U)^\dag (H\otimes I)\mathcal{C}(U)]) \nonumber\\
		=&\mathrm{tr}((\rho\otimes\rho_\mathrm{W})[(H-U^\dag HU)\otimes I]) \nonumber\\
		=&\mathrm{tr}(\rho(H-U^\dagger HU))\mathrm{tr}(\rho_\mathrm{W}) \nonumber\\
		=&W(\rho,H,U).
	\end{align}
\end{proof}

	 Proposition~\ref{prop:GFSCFQergotropyinequality} states that the group-symmetry respecting ergotropy of a state is no greater in the WS setup than in the setup without the work storage.

\begin{proposition} \label{prop:GFSCFQergotropyinequality}
	Let $G$ be a group with a unitary representation $F$ acting on $\mathcal{H}$, $H\in\mathcal{B}^\mathrm{H}(\mathcal{H})$ commute with $F(G)$ and $\rho\in\mathcal{B}^+(\mathcal{H})$.
	Then, $\mathcal{W}_{G, F, H, \rho_\mathrm{W}}^\mathrm{WS}(\rho)\leq\mathcal{W}_{G, F, H}(\rho)$ holds for all $\rho_\mathrm{W}\in\mathcal{B}^+(\mathcal{H}_\mathrm{W})$ satisfying $\mathrm{tr}(\rho_\mathrm{W})=1$.
\end{proposition}

\begin{proof}
	For any $\rho_\mathrm{W}\in\mathcal{B}^+(\mathcal{H}_\mathrm{W})$ satifying $\mathrm{tr}(\rho_\mathrm{W})=1$, take $U_0\in\mathcal{U}_{G, F}(\mathcal{H})$ that satisfies $W(\mathcal{D}_{H, \rho_\mathrm{W}}(\rho), H, U_0)=\mathcal{W}_{G, F, H}(\mathcal{D}_{H, \rho_\mathrm{W}}(\rho))$. 
	From Corollary~\ref{cor:GFSCFQergotropy}, we obtain
	\begin{align}
		\mathcal{W}_{G, F, H, \rho_\mathrm{W}}^\mathrm{WS}(\rho)=&\mathcal{W}_{G, F, H}(\mathcal{D}_{H, \rho_\mathrm{W}}(\rho)) \nonumber\\
		=&W(\mathcal{D}_{H, \rho_\mathrm{W}}(\rho), H, U_0) \nonumber\\
		=&W\left(\int_{-\infty}^\infty dq\ {}_p\braket{q|\rho_\mathrm{W}|q}_p e^{-\mathrm{i}qH}\rho e^{\mathrm{i}qH}, H, U_0\right) \nonumber\\
		=&\int_{-\infty}^\infty dq\ {}_p\braket{q|\rho_\mathrm{W}|q}_p W(e^{-\mathrm{i}qH}\rho e^{\mathrm{i}qH}, H, U_0) \nonumber\\
		=&\int_{-\infty}^\infty dq\ {}_p\braket{q|\rho_\mathrm{W}|q}_p W(\rho , H, e^{\mathrm{i}qH}U_0 e^{-\mathrm{i}qH}) \nonumber\\
		\leq&\int_{-\infty}^\infty dq\ {}_p\braket{q|\rho_\mathrm{W}|q}_p \mathcal{W}_{G, F, H}(\rho) \nonumber\\
		=&\mathcal{W}_{G, F, H}(\rho). \label{eq:GF_work_locking}
	\end{align}
\end{proof}

	Proposition~\ref{prop:TSCFQergotropyinequality} states that the time-reversal-symmetry respecting ergotropy of a state is also no greater in the WS setup than in the setup without the work storage.

\begin{proposition} \label{prop:TSCFQergotropyinequality}
	Let $H\in\mathcal{B}^\mathrm{H}(\mathcal{H})$ commute with $\mathcal{T}$ and $\rho\in\mathcal{B}^+(\mathcal{H})$.
	Then, $\mathcal{W}_{\mathcal{T}, H, \rho_\mathrm{W}}^\mathrm{WS}(\rho)\leq\mathcal{W}_{\mathcal{T}, H}(\rho)$ holds for all $\rho_\mathrm{W}\in\mathcal{B}^+(\mathcal{H}_\mathrm{W})$ satifying $\mathrm{tr}(\rho_\mathrm{W})=1$.
\end{proposition}

\begin{proof}
	For any $\rho_\mathrm{W}\in\mathcal{B}^+(\mathcal{H}_\mathrm{W})$ satifying $\mathrm{tr}(\rho_\mathrm{W})=1$, take $U_0\in\mathcal{U}_\mathcal{T}(\mathcal{H})$ that satisfies $W(\mathcal{D}_{H, \rho_\mathrm{WS}}(\rho), H, U_0)=\mathcal{W}_{\mathcal{T}, H}(\mathcal{D}_{H, \rho_\mathrm{WS}}(\rho))$. 
	From Corollary~\ref{cor:TSCFQergotropy}, we can prove that $\mathcal{W}_{\mathcal{T}, H, \rho_\mathrm{W}}^\mathrm{WS}(\rho)\leq\mathcal{W}_{\mathcal{T}, H}(\rho)$ in the same way as Eq.~\eqref{eq:GF_work_locking}.
\end{proof}

	The proof of Theorem~\ref{thm:GFHCP1} is as follows.
	The necessity of the generalized Gibbs ensemble for complete passivity protected by Lie group symmetry in the WS setup is derived from Propositions~\ref{prop:GFworkextop} and \ref{prop:stateswapunitary}, while the sufficiency is derived from Theorem~\ref{thm:GFHCP} and Proposition~\ref{prop:GFSCFQergotropyinequality}.\\\\
\textit{Proof of Theorem~\ref{thm:GFHCP1}.}
	First, suppose that $\rho$ cannot be written as $\rho=\frac{1}{Z}e^{-\beta H-f(X)}$ with $\beta\in[0, \infty)$, $X\in\mathfrak{g}$ and $Z:=\mathrm{tr}(e^{-\beta H-f(X)})$.
	From Proposition~\ref{prop:GFworkextop}, there exist $N\in\mathbb{N}$ and $U\in\mathcal{U}_{G, F^{\otimes N}}(\mathcal{H}^{\otimes N})$ such that $[U^\dag H^{(N)}U, H^{(N)}]=0$ and $W(\rho^{\otimes N}, H^{(N)}, U)>0$.
	From Lemma~\ref{lem:GFSCFQcorrespondence}, $\mathcal{C}(U)\in\mathcal{U}_{G, F^{\otimes N}}^\mathrm{WS}(\mathcal{H}^{\otimes N}, \mathcal{H}_\mathrm{W})$.
	From Proposition~\ref{prop:stateswapunitary}, for any $\rho_\mathrm{W}\in\mathcal{B}^+(\mathcal{H}_\mathrm{W})$ satisfying $\mathrm{tr}(\rho_\mathrm{W})=1$, we obtain $W^\mathrm{WS}(\rho^{\otimes N}, H^{(N)}, \rho_\mathrm{W}, \mathcal{C}(U))=W(\rho^{\otimes N}, H^{(N)}, U)>0$.
	This implies that $\rho$ is not WS-$(G, F, \rho_\mathrm{W})$-completely passive w.r.t. $H$.
	
	Next, suppose that $\rho$ can be written as $\rho=\frac{1}{Z}e^{-\beta H-f(X)}$ with $\beta\in[0, \infty)$, $X\in\mathfrak{g}$ and $Z:=\mathrm{tr}(e^{-\beta H-f(X)})$.
	Take arbitrary $N\in\mathbb{N}$ and $\rho_\mathrm{W}\in\mathcal{B}^+(\mathcal{H}_\mathrm{W})$ satisfying $\mathrm{tr}(\rho_\mathrm{W})=1$.
	From Theorem~\ref{thm:GFHCP}, $\rho$ is $(G, F)$-completely passive w.r.t. $H$. 
	From Proposition~\ref{prop:GFSCFQergotropyinequality}, we get
	\begin{align}
		\mathcal{W}_{G, F^{\otimes N}, H^{(N)}, \rho_\mathrm{W}}^\mathrm{WS}(\rho^{\otimes N})\leq\mathcal{W}_{G, F^{\otimes N}, H^{(N)}}(\rho^{\otimes N})=0.
	\end{align}
	This implies that $\rho$ is WS-$(G, F, \rho_\mathrm{W})$-completely passive w.r.t. $H$. \hspace{\fill} $\Box$\\

	The proof of Theorem~\ref{thm:finiteGFHCP1} is as follows, whose structure is parallel to that of Theorem~\ref{thm:GFHCP1}.
	The necessity of the Gibbs ensemble for complete passivity protected by finite cyclic group symmetry or dihedral group symmetry in the WS setup is derived from Propositions~\ref{prop:finiteGFworkextop} and \ref{prop:stateswapunitary}, while the sufficiency is derived from Theorem~\ref{thm:finiteGFHCP} and Proposition~\ref{prop:GFSCFQergotropyinequality}.\\\\
\textit{Proof of Theorem~\ref{thm:finiteGFHCP1}.}
	First, suppose that $\rho$ cannot be written as $\rho=\frac{1}{Z}e^{-\beta H}$ with $\beta\in[0, \infty)$ and $Z:=\mathrm{tr}(e^{-\beta H})$.
	From Proposition~\ref{prop:finiteGFworkextop}, there exist $N\in\mathbb{N}$ and $U\in\mathcal{U}_{G, F^{\otimes N}}(\mathcal{H}^{\otimes N})$ such that $[U^\dag H^{(N)}U, H^{(N)}]=0$ and $W(\rho^{\otimes N}, H^{(N)}, U)>0$.
	From Lemma~\ref{lem:GFSCFQcorrespondence}, $\mathcal{C}(U)\in\mathcal{U}_{G, F^{\otimes N}}^\mathrm{WS}(\mathcal{H}^{\otimes N}, \mathcal{H}_\mathrm{W})$.
	From Proposition~\ref{prop:stateswapunitary}, for any $\rho_\mathrm{W}\in\mathcal{B}^+(\mathcal{H}_\mathrm{W})$ satisfying $\mathrm{tr}(\rho_\mathrm{W})=1$, we obtain $W^\mathrm{WS}(\rho^{\otimes N}, H^{(N)}, \rho_\mathrm{W}, \mathcal{C}(U))=W(\rho^{\otimes N}, H^{(N)}, U)>0$.
	This implies that $\rho$ is not WS-$(G, F, \rho_\mathrm{W})$-completely passive w.r.t. $H$.

	Next, suppose that $\rho$ can be written as $\rho=\frac{1}{Z}e^{-\beta H}$ with $\beta\in[0, \infty)$ and $Z:=\mathrm{tr}(e^{-\beta H})$.
	Take arbitrary $N\in\mathbb{N}$ and $\rho_\mathrm{W}\in\mathcal{B}^+(\mathcal{H}_\mathrm{W})$ satisfying $\mathrm{tr}(\rho_\mathrm{W})=1$.
	From Theorem~\ref{thm:finiteGFHCP}, $\rho$ is $(G, F)$-completely passive w.r.t. $H$. 
	From Proposition~\ref{prop:GFSCFQergotropyinequality}, we get
	\begin{align}
		\mathcal{W}_{G, F^{\otimes N}, H^{(N)}, \rho_\mathrm{W}}^\mathrm{WS}(\rho^{\otimes N})\leq\mathcal{W}_{G, F^{\otimes N}, H^{(N)}}(\rho^{\otimes N})=0.
	\end{align}
	This implies that $\rho$ is WS-$(G, F, \rho_\mathrm{W})$-completely passive w.r.t. $H$. \hspace{\fill} $\Box$\\

	The proof of Theorem~\ref{thm:THCP1} is as follows, whose structure is parallel to that of Theorem~\ref{thm:GFHCP1}.
	The necessity of the Gibbs ensemble for complete passivity protected by time-reversal symmetry in the WS setup is derived from Propositions~\ref{prop:Tworkextop} and \ref{prop:stateswapunitary}, while the sufficiency is derived from Theorem~\ref{thm:THCP} and Proposition~\ref{prop:TSCFQergotropyinequality}.\\\\
\textit{Proof of Theorem~\ref{thm:THCP1}.}
	First, suppose that $\rho$ cannot be written as $\rho=\frac{1}{Z}e^{-\beta H}$ with $\beta\in[0, \infty)$ and $Z:=\mathrm{tr}(e^{-\beta H})$.
	From Proposition~\ref{prop:Tworkextop}, there exist $N\in\mathbb{N}$ and $U\in\mathcal{U}_\mathcal{T}(\mathcal{H}^{\otimes N})$ such that $[U^\dag H^{(N)}U, H^{(N)}]=0$ and $W(\rho^{\otimes N}, H^{(N)}, U)>0$.
	From Lemma~\ref{lem:TSCFQcorrespondence}, $\mathcal{C}(U)\in\mathcal{U}_\mathcal{T}^\mathrm{WS}(\mathcal{H}^{\otimes N}, \mathcal{H}_\mathrm{W})$.
	From Proposition~\ref{prop:stateswapunitary}, for any $\rho_\mathrm{W}\in\mathcal{B}^+(\mathcal{H}_\mathrm{W})$ satisfying $\mathrm{tr}(\rho_\mathrm{W})=1$, we obtain $W^\mathrm{WS}(\rho^{\otimes N}, H^{(N)}, \rho_\mathrm{W}, \mathcal{C}(U))=W(\rho^{\otimes N}, H^{(N)}, U)>0$.
	This implies that $\rho$ is not WS-$(\mathcal{T}, \rho_\mathrm{W})$-completely passive w.r.t. $H$.

	Next, suppose that $\rho$ can be written as $\rho=\frac{1}{Z}e^{-\beta H}$ with $\beta\in[0, \infty)$ and $Z:=\mathrm{tr}(e^{-\beta H})$.
	Take arbitrary $N\in\mathbb{N}$ and $\rho_\mathrm{W}\in\mathcal{B}^+(\mathcal{H}_\mathrm{W})$ satisfying $\mathrm{tr}(\rho_\mathrm{W})=1$.
	From Theorem~\ref{thm:THCP}, $\rho$ is $\mathcal{T}$-completely passive w.r.t. $H$. 
	From Proposition~\ref{prop:TSCFQergotropyinequality}, we get
	\begin{align}
		\mathcal{W}_{\mathcal{T}, H^{(N)}, \rho_\mathrm{W}}^\mathrm{WS}(\rho^{\otimes N})\leq\mathcal{W}_{\mathcal{T}, H^{(N)}}(\rho^{\otimes N})=0.
	\end{align}
	This implies that $\rho$ is WS-$(\mathcal{T}, \rho_\mathrm{W})$-completely passive w.r.t. $H$. \hspace{\fill} $\Box$\\

\section{Technical lemmas}
\label{sec:technical_lemmas}

	In this section, we prove technical lemmas used in the proofs of the foregoing theorems.

	Lemma~\ref{lem:latticeduality} states that the orthogonal complement of the subspace of $\mathbb{R}^K$ spanned by vectors in $\mathbb{Q}^K$ can be spanned by vectors in $\mathbb{N}^K$.
	We use the Gram-Schmidt orthogonalization in the proof.

\begin{lemma} \label{lem:latticeduality}
	Let $K, a\in\mathbb{N}$, $\bm{w}_1, \bm{w}_2, \cdots, \bm{w}_a\in\mathbb{Q}^K$ and $\bm{t}\in\mathbb{R}^K$. 
	Suppose that for any $\bm{n}_0,\ \bm{n}_1\in\mathbb{N}^K$ satisfying $\bm{n}_0\cdot\bm{w}_i=\bm{n}_1\cdot\bm{w}_i$ for all $i=1, \cdots, a$, $\bm{t}$ satisfies $\bm{n}_0\cdot\bm{t}=\bm{n}_1\cdot\bm{t}$. 
	Then, $\bm{t}$ is a linear combination of $\{\bm{w}_{i}\}_{i=1}^a$ over $\mathbb{R}$.
\end{lemma} 

\begin{proof}
	Let $\{\bm{e}_i\}_{i=1}^K$ be the standard basis of $\mathbb{R}^K$. 
	Take arbitrary $\bm{t}\in\mathbb{R}^K$ such that $\bm{n}_0\cdot\bm{t}=\bm{n}_1\cdot\bm{t}$ for all $\bm{n}_0, \bm{n}_1\in\mathbb{N}^K$ satisfying $\bm{n}_0\cdot\bm{w}_i=\bm{n}_1\cdot\bm{w}_i$ for all $i=1, \cdots, a$.
	By applying the Gram-Schmidt orthogonalization to the sequence $(\bm{w}_1, \bm{w}_2, \cdots, \bm{w}_a)$, we get a set of orthogonal vectors $\{\bm{x}_1, \bm{x}_2, \cdots, \bm{x}_{a'}\}$, and by applying the Gram-Schmidt orthogonalization to the sequence $(\bm{w}_1, \bm{w}_2, \cdots, \bm{w}_a, \bm{e}_1, \bm{e}_2, \cdots, \bm{e}_K)$, we get a set of orthogonal vectors $\{\bm{x}_1, \bm{x}_2, \cdots, \bm{x}_{a'}, \bm{y}_1, \bm{y}_2, \cdots, \bm{y}_{b}\}$.
	Since $\bm{w}_i, \bm{e}_j\in\mathbb{Q}^K$ for all $i=1, \cdots, a$ and $j=1, \cdots, K$, we obtain $\bm{x}_i, \bm{y}_j\in\mathbb{Q}^K$ for all $i=1, \cdots, a'$ and $ j=1, \cdots, b$.
	Since $\{\bm{w}_1, \bm{w}_2, \cdots, \bm{w}_a, \bm{e}_1, \bm{e}_2, \cdots, \bm{e}_K\}$ spans $\mathbb{R}^K$, $\{\bm{x}_1, \bm{x}_2, \cdots, \bm{x}_{a'}, \bm{y}_1, \bm{y}_2, \cdots, \bm{y}_{b}\}$ is a basis of $\mathbb{R}^K$.
	Therefore, $b=K-a'$ and $\mathrm{span}(\{\bm{y}_1, \bm{y}_2, \cdots, \bm{y}_{b}\})$ is the orthogonal complement of $\mathrm{span}(\{\bm{x}_1, \bm{x}_2, \cdots, \bm{x}_{a'}\})=\mathrm{span}(\{\bm{w}_1, \bm{w}_2, \cdots, \bm{w}_{a}\})$.
	For any $j=1, \cdots, b$, since $\bm{y}_j\in\mathbb{Q}^K$, $\bm{y}_j$ can be written as $\bm{y}_j=\frac{1}{m}(\bm{n}_1-\bm{n}_0)$ with some $m\in\mathbb{N}$ and $\bm{n}_0, \bm{n}_1\in\mathbb{N}^K$.
	For any $i=1, \cdots, a$, since $\bm{y}_j\cdot\bm{w}_i=0$, we get $\bm{n}_0\cdot\bm{w}_i=\bm{n}_1\cdot\bm{w}_i$.
	From the assumption about $\bm{t}$, $\bm{n}_0\cdot\bm{t}=\bm{n}_1\cdot\bm{t}$ and thus $\bm{y}_j\cdot\bm{t}=0$.
	Since this holds for all $j=1, \cdots, b$, $\bm{t}$ is a linear combination of $\{\bm{x}_i\}$, or equivalently, a linear combination of $\{\bm{w}_i\}$.
\end{proof}

	In Lemma~\ref{lem:determinantstate}, we prove the properties of the totally antisymmetric state $\mathcal{A}_\mathcal{R}(\{\ket{\phi_i}\})$ defined as Eq.~\eqref{eq:defdetstate}. 
	We can regard this state as a generalization of the spin singlet state.

\begin{lemma} \label{lem:determinantstate}
	Let $\mathcal{R}$ be a Hilbert space of dimension $r$ and $\{\ket{\phi_i}\}_{i=1}^r$ be an orthonormal basis of $\mathcal{R}$.
	Then, $\mathcal{A}_\mathcal{R}(\{\ket{\phi_i}\})$ defined as Eq.~\eqref{eq:defdetstate} is independent of the choice of the basis up to phase factor and satisfies 
	\begin{align}
		&\|\mathcal{A}_\mathcal{R}(\{\ket{\phi_i}\})\|=1,\\
		&\forall\Omega\in\mathcal{B}(\mathcal{R})\ \Omega^{\otimes r}\mathcal{A}_\mathcal{R}(\{\ket{\phi_i}\})=\mathrm{det}(\Omega)\mathcal{A}_\mathcal{R}(\{\ket{\phi_i}\}).
	\end{align} 
\end{lemma}

\begin{proof}
	For any orthonormal basis $\{\ket{\phi_i}\}$, since every element of $\{\mathrm{sgn}(\sigma)\ket{\phi_{\sigma(1)}}\otimes\cdots\otimes\ket{\phi_{\sigma(r)}}\}_{\sigma\in S_r}$ is normalized and orthogonal to each other, $\|\mathcal{A}_\mathcal{R}(\{\ket{\phi_i}\})\|=1$.
	Take an arbitrary orthonormal basis $\{\ket{\phi_i}\}$.
	For any $\Omega\in\mathcal{B}(\mathcal{R})$, we get
	\begin{align}
		\Omega^{\otimes r}\mathcal{A}_\mathcal{R}(\{\ket{\phi_i}\})=&\mathcal{A}_\mathcal{R}(\Omega\ket{\phi_1}, \cdots, \Omega\ket{\phi_r}) \nonumber\\
		=&\mathcal{A}_\mathcal{R}\left(\sum_{s_1=1}^r\omega_{s_1 1}\ket{\phi_{s_1}}, \cdots, \sum_{s_r=1}^r\omega_{s_r r}\ket{\phi_{s_r}}\right) \nonumber\\
		=&\sum_{\sigma\in T_r} \mathcal{A}_\mathcal{R}(\omega_{\sigma(1)1}\ket{\phi_{\sigma(1)}}, \cdots, \omega_{\sigma(r)r}\ket{\phi_{\sigma(r)}}) \nonumber\\
		=&\sum_{\sigma\in T_r} \prod_{i=1}^r \omega_{\sigma(i)i}\mathcal{A}_\mathcal{R}(\ket{\phi_{\sigma(1)}}, \cdots, \ket{\phi_{\sigma(r)}}) \nonumber\\
		=&\sum_{\sigma\in S_r} \prod_{i=1}^r \omega_{\sigma(i)i}\mathrm{sgn}(\sigma)\mathcal{A}_\mathcal{R}(\ket{\phi_1}, \cdots, \ket{\phi_r}) \nonumber\\
		=&\mathrm{det}(\Omega)\mathcal{A}_\mathcal{R}(\{\ket{\phi_i}\}), \label{eq:detstateproperty}
	\end{align}
	where $\omega_{ij}:=\braket{\phi_i|\Omega|\phi_j}$, $T_r$ is the set of all mappings from $\{1, 2, \cdots, r\}$ to itself, and we used the total antisymmetricity of $\mathcal{A}_\mathcal{R}(\{\ket{\phi_i}\})$ about $\{\ket{\phi_i}\}$.
	Take another orthonormal basis $\{\ket{\phi'_i}\}$.
	Then we can take $U\in\mathcal{U}(\mathcal{R})$ such that $U\ket{\phi_i}=\ket{\phi'_i}$ for all $i=1, \cdots, r$.
	Therefore, the relation between $\mathcal{A}_\mathcal{R}(\{\ket{\phi_i}\})$ and $\mathcal{A}_\mathcal{R}(\{\ket{\phi'_i}\})$ can be written as 
	\begin{align}
		\mathcal{A}_\mathcal{R}(\{\ket{\phi'_i}\})=&\mathcal{A}_\mathcal{R}(U\ket{\phi_1}, \cdots, U\ket{\phi_r}) \nonumber\\
		=&U^{\otimes r}\mathcal{A}_\mathcal{R}(\ket{\phi_1}, \cdots, \ket{\phi_r}) \nonumber\\
		=&\mathrm{det}(U)\mathcal{A}_\mathcal{R}(\{\ket{\phi_i}\}) \nonumber\\
		=&e^{\mathrm{i}\theta}\mathcal{A}_\mathcal{R}(\{\ket{\phi_i}\})
	\end{align}
	with some $\theta\in\mathbb{R}$.
\end{proof}

	In Lemma~\ref{lem:eigenvectorcomposition}, we prove the properties of the state $\ket{\Phi(\bm{n})}$ defined as Eq.~\eqref{eq:determinantstate}.
	The definition of $\ket{\Phi(\bm{n})}$ is based on the irreducible decomposition of a representation $F$ of a group $G$ and the diagonalization of a Hamiltonian $H$, and $\ket{\Phi(\bm{n})}$ is a simultaneous eigenstate of such operators as $F^{\otimes N}(g)$ with $g\in G$ and $H^{(N)}$ on $N$ copies of a system.

\begin{lemma} \label{lem:eigenvectorcomposition}
	Let $G$ be a group with a unitary representation $F$ acting on $\mathcal{H}$, $H\in\mathcal{B}^\mathrm{H}(\mathcal{H})$ commute with $F(G)$, $F$ be irreducibly decomposed  as Eq.~\eqref{eq:irreducibledecomposition}, and $\Omega\in\mathcal{B}(\mathcal{H})$ be written as
	\begin{align}
		\Omega=\sum_{\lambda\in\Lambda_F} \iota_\lambda \left(\sum_{j=1}^{m_\lambda} \zeta_{\lambda j}\otimes\ket{\psi_{\lambda j}}\bra{\psi_{\lambda j}}\right)\iota_\lambda^\dag. \label{eq:Omega_decomp}
	\end{align}
	with some $\zeta_{\lambda j}\in\mathcal{B}(\mathcal{R}_\lambda)$.
	Then, for any $\bm{n}=(n_k)_{k\in K}\in\mathbb{N}^K$, $\ket{\Phi(\bm{n})}$ defined as Eq.~\eqref{eq:determinantstate} satisfies
	\begin{align}
		\Omega^{\otimes D\bm{n}\cdot\bm{w}_0}\ket{\Phi(\bm{n})}=\left(\prod_{k\in K}\mathrm{det}(\zeta_k)^{n_k\frac{D}{r_\lambda}}\right)\ket{\Phi(\bm{n})},\ \Omega^{(D\bm{n}\cdot\bm{w}_0)}\ket{\Phi(\bm{n})}=D\bm{n}\cdot\bm{u}\ket{\Phi(\bm{n})}, \label{eq:Omega_eigenvalue}
	\end{align}
	where $D:=\prod_{\lambda\in\Lambda_F} r_\lambda$, $K:=\{(\lambda, j)\ |\ \lambda\in\Lambda_F, j=1, \cdots, m_\lambda\}$, $\bm{w}_0:=(1)_{k\in K}$, $\bm{u}:=(u_k)_{k\in K}$, $u_k:=\frac{1}{r_\lambda}\mathrm{tr}(\zeta_{\lambda j})$, $\zeta_k:=\zeta_{\lambda j}$ for $k=(\lambda, j)\in K$ and $\bm{a}\cdot\bm{b}:=\sum_{k\in K} a_kb_k$ for $\bm{a}=(a_k)_{k\in K}, \bm{b}=(b_k)_{k\in K}\in\mathbb{R}^K$. 
\end{lemma}

\begin{proof}
	From Lemma~\ref{lem:determinantstate}, for any $\Omega\in\mathcal{B}(\mathcal{H})$, $\lambda\in\Lambda_F$ and $j\in M_\lambda$,
	\begin{align}
		&\Omega^{\otimes r_\lambda}\iota_\lambda^{\otimes r_\lambda}(\mathcal{A}_{\mathcal{R}_\lambda}(\{\ket{\phi_{\lambda i}}\}_{i=1}^{r_\lambda})\otimes\ket{\psi_{\lambda j}}^{\otimes r_\lambda}) \nonumber\\
		=&\left[\sum_{\lambda'\in\Lambda_F} \iota_{\lambda'} \left(\sum_{j'=1}^{m_{\lambda'}} \zeta_{\lambda' j'}\otimes\ket{\psi_{\lambda' j'}}\bra{\psi_{\lambda' j'}}\right)\iota_{\lambda'}^\dag\right]^{\otimes r_\lambda}\iota_\lambda^{\otimes r_\lambda}(\mathcal{A}_{\mathcal{R}_\lambda}(\{\ket{\phi_{\lambda i}}\}_{i=1}^{r_\lambda})\otimes\ket{\psi_{\lambda j}}^{\otimes r_\lambda}) \nonumber\\
		=&\iota_\lambda^{\otimes r_\lambda}\left(\sum_{j'=1}^{m_\lambda} \zeta_{\lambda j'}\otimes\ket{\psi_{\lambda j'}}\bra{\psi_{\lambda j'}}\right)^{\otimes r_\lambda}(\mathcal{A}_{\mathcal{R}_\lambda}(\{\ket{\phi_{\lambda i}}\}_{i=1}^{r_\lambda})\otimes\ket{\psi_{\lambda j}}^{\otimes r_\lambda}) \nonumber\\
		=&\iota_\lambda^{\otimes r_\lambda}(\zeta_{\lambda j}^{\otimes r_\lambda}\mathcal{A}_{\mathcal{R}_\lambda}(\{\ket{\phi_{\lambda i}}\}_{i=1}^{r_\lambda})\otimes\ket{\psi_{\lambda j}}^{\otimes r_\lambda}) \nonumber\\
		=&\iota_\lambda^{\otimes r_\lambda}(\mathrm{det}(\zeta_{\lambda j})\mathcal{A}_{\mathcal{R}_\lambda}(\{\ket{\phi_{\lambda i}}\}_{i=1}^{r_\lambda})\otimes\ket{\psi_{\lambda j}}^{\otimes r_\lambda}) \nonumber\\
		=&\mathrm{det}(\zeta_{\lambda j})\iota_\lambda^{\otimes r_\lambda}(\mathcal{A}_{\mathcal{R}_\lambda}(\{\ket{\phi_{\lambda i}}\}_{i=1}^{r_\lambda})\otimes\ket{\psi_{\lambda j}}^{\otimes r_\lambda}).
	\end{align}
	By the definition of $\ket{\chi_k}$ (Eq.~\eqref{eq:chi_k_definition}),
	\begin{align}
		&\Omega^{\otimes D}\ket{\chi_k} \nonumber\\
		=&[\Omega^{\otimes r_\lambda}\iota_\lambda^{\otimes r_\lambda}(\mathcal{A}_{\mathcal{R}_\lambda}(\ket{\phi_1}, \cdots, \ket{\phi_{r_\lambda}})\otimes\ket{\psi_{\lambda j}}^{\otimes r_\lambda})]^{\otimes\frac{D}{r_\lambda}} \nonumber\\
		=&[\mathrm{det}(\zeta_{\lambda j})\iota_\lambda^{\otimes r_\lambda}(\mathcal{A}_{\mathcal{R}_\lambda}(\ket{\phi_1}, \cdots, \ket{\phi_{r_\lambda}})\otimes\ket{\psi_{\lambda j}}^{\otimes r_\lambda})]^{\otimes\frac{D}{r_\lambda}} \nonumber\\
		=&\mathrm{det}(\zeta_{\lambda j})^\frac{D}{r_\lambda}\ket{\chi_k}.
	\end{align}
	Furthermore, for any $\bm{n}\in\mathbb{N}^K$, by the definition of $\ket{\Phi(\bm{n})}$, 
	\begin{align}
		\Omega^{\otimes D\bm{n}\cdot\bm{w}_0}\ket{\Phi(\bm{n})}=\bigotimes_{k\in K}(\Omega^{\otimes D}\ket{\chi_k})^{\otimes n_k}
		=\bigotimes_{k\in K}\left(\mathrm{det}(\zeta_k)^\frac{D}{r_\lambda}\ket{\chi_k}\right)^{\otimes n_k}
		=\left(\prod_{k\in K}\mathrm{det}(\zeta_k)^{n_k \frac{D}{r_\lambda}}\right)\ket{\Phi(\bm{n})}. \label{eq:detstateresult}
	\end{align}
	Since for any $\theta\in\mathbb{R}$, 
	\begin{align}
		e^{\theta\Omega}=\sum_{\lambda\in\Lambda_F} \iota_\lambda \left(\sum_{j=1}^{m_\lambda} e^{\theta\zeta_{\lambda j}}\otimes\ket{\psi_{\lambda j}}\bra{\psi_{\lambda j}}\right)\iota_\lambda^\dag,
	\end{align}
	by substituting $e^{\theta\Omega}$ into $\Omega$ in Eq.~\eqref{eq:detstateresult}, we get 
	\begin{align}
		e^{\theta\Omega^{(D\bm{n}\cdot\bm{w}_0)}}\ket{\Phi(\bm{n})}=&{(e^{\theta\Omega})}^{\otimes D\bm{n}\cdot\bm{w}_0}\ket{\Phi(\bm{n})} \nonumber\\
		=&\left(\prod_{k\in K}\mathrm{det}(e^{\theta\zeta_k})^{n_k \frac{D}{r_\lambda}}\right)\ket{\Phi(\bm{n})} \nonumber\\
		=&\left(\prod_{k\in K}e^{\theta\mathrm{tr}(\zeta_k)\cdot n_k \frac{D}{r_\lambda}}\right)\ket{\Phi(\bm{n})} \nonumber\\
		=&e^{\theta D\sum_{k\in K}n_l\frac{1}{r_\lambda}\mathrm{tr}(\zeta_k)}\ket{\Phi(\bm{n})} \nonumber\\
		=&e^{\theta D\bm{n}\cdot\bm{u}}\ket{\Phi(\bm{n})}.
	\end{align}
	By taking the derivative at $\theta=0$, we get
	\begin{align}
		\Omega^{(D\bm{n}\cdot\bm{w}_0)}\ket{\Phi(\bm{n})}=D\bm{n}\cdot\bm{u}\ket{\Phi(\bm{n})}.
	\end{align}
\end{proof}

Lemma~\ref{lem:PsiexistenceforLie} guarantees that there exist states $\ket{\Psi_0}$ and $\ket{\Psi_1}$ in Propositions~\ref{prop:preGGEderivation} and \ref{prop:GGEderivation}, when the symmetry group $G$ is a connected compact Lie group and the Hamiltonian $H$ is not $(G, F)$-trivial.

\begin{lemma} \label{lem:PsiexistenceforLie}
	Let $G$ be a connected compact Lie group with a unitary representation $F$ acting on $\mathcal{H}$, $\mathfrak{g}$ be the Lie algebra with the associated representation $f$, $H\in\mathcal{B}^\mathrm{H}(\mathcal{H})$ commute with $F(G)$ and be not $(G, F)$-trivial.
	Then, there exist $L\in\mathbb{N}$ and a pair of simultaneous eigenstates $\ket{\Psi_0},\ \ket{\Psi_1}\in\mathcal{H}^{\otimes L}$ of $F^{\otimes L}(G)$ and $H^{(L)}$ such that the eigenvalues satisfy $\braket{\Psi_0|F^{\otimes L}(g)|\Psi_0}=\braket{\Psi_1|F^{\otimes L}(g)|\Psi_1}$ for all $g\in G$ but $\braket{\Psi_0|H^{(L)}|\Psi_0}\neq\braket{\Psi_1|H^{(L)}|\Psi_1}$.
\end{lemma}

\begin{proof}
	Suppose that for any $\bm{n}_0=(n_{0k})_{k\in K},\ \bm{n}_1=(n_{1k})_{k\in K}\in\mathbb{N}^K$ satisfying $\bm{n}_0\cdot\bm{w}_i=\bm{n}_1\cdot\bm{w}_i$ for all $i=0, \cdots, a$, we obtain $\bm{n}_0\cdot\bm{E}=\bm{n}_1\cdot\bm{E}$, where $K:=\{(\lambda, j)\ |\ \lambda\in\Lambda_F, j=1, \cdots, m_\lambda\}$, $\bm{w}_0:=(1)_{k\in K}$, $\bm{w}_i:=(w_{ik})_{k\in K}$, $w_{ik}:=w_{i \lambda}$ for $k=(\lambda, j)\in K$, $w_{i \lambda}$ is defined in Lemma~\ref{lem:Liealgebrabasis} for $i=1, \cdots, a$, $\bm{E}:=(E_k)_{k\in K}$, $E_k:=E_{\lambda j}$ for $k=(\lambda, j)\in K$, and $E_{\lambda j}$ is defined in Eq.~\eqref{eq:Hamiltoniandecomp2}.
	Since $\bm{w}_i\in\mathbb{N}^K$ for all $i=0, \cdots, a$, from Lemma~\ref{lem:latticeduality}, $\bm{E}$ can be written as $\bm{E}=\sum_{i=0}^a \alpha_i \bm{w}_i$ with some $\alpha_0, \cdots, \alpha_a\in\mathbb{R}$.
	Then, for any $\lambda\in\Lambda_F$ and $j=1, \cdots, m_\lambda$, we obtain $E_{\lambda j}=\sum_{i=0}^a \alpha_i w_{i\lambda}$, and thus 
	\begin{align}
		H
		=&\sum_{\lambda\in\Lambda_F} \sum_{j=1}^{m_\lambda} E_{\lambda j}\iota_\lambda(I_{\mathcal{R}_\lambda}\otimes \ket{\psi_{\lambda j}}\bra{\psi_{\lambda j}})\iota_\lambda^\dag \nonumber\\
		=&\sum_{\lambda\in\Lambda_F} \sum_{j=1}^{m_\lambda} \left(\sum_{i=0}^a \alpha_i w_{i\lambda}\right)\iota_\lambda(I_{\mathcal{R}_\lambda}\otimes \ket{\psi_{\lambda j}}\bra{\psi_{\lambda j}})\iota_\lambda^\dag \nonumber\\
		=&\sum_{\lambda\in\Lambda_F} \left(\alpha_0+\sum_{i=1}^a \alpha_i w_{i\lambda}\right)\iota_\lambda\left(I_{\mathcal{R}_\lambda}\otimes \sum_{j=1}^{m_\lambda} \ket{\psi_{\lambda j}}\bra{\psi_{\lambda j}}\right)\iota_\lambda^\dag \nonumber\\
		=&\sum_{\lambda\in\Lambda_F} \left(\alpha_0+\sum_{i=1}^a \alpha_i w_{i\lambda}\right)\iota_\lambda\iota_\lambda^\dag \nonumber\\
		=&\alpha_0\sum_{\lambda\in\Lambda_F} \iota_\lambda\iota_\lambda^\dag+\sum_{i=1}^a \left(\alpha_i \sum_{\lambda\in\Lambda_F} w_{i\lambda}\iota_\lambda\iota_\lambda^\dag\right) \nonumber\\
		=&\alpha_0 I+\sum_{i=1}^a \alpha_i f(W_i) \nonumber\\
		=&\alpha_0 I+f(W),
	\end{align}
	where $W:=\sum_{i=1}^a W_i\in\mathfrak{g}$.
	Since this contradicts the fact that $H$ is not $(G, F)$-trivial, we can take $\bm{n}_0,\ \bm{n}_1\in\mathbb{N}^K$ such that $\bm{n}_0\cdot\bm{w}_i=\bm{n}_1\cdot\bm{w}_i$ for all $i=0, \cdots, a$ but $\bm{n}_0\cdot\bm{E}\neq\bm{n}_1\cdot\bm{E}$.
	We define $\ket{\Psi_l}:=\ket{\Phi(\bm{n}_l)}$ with $\ket{\Phi(\bm{n})}$ defined as Eq.~\eqref{eq:determinantstate}, where $L:=D\bm{n}_0\cdot\bm{w}_0=D\bm{n}_1\cdot\bm{w}_0$.
	In the same way as $\ket{\Psi'_0}$ and $\ket{\Psi'_1}$ in Proposition~\ref{prop:GGEderivation}, $\ket{\Psi_0}$ and $\ket{\Psi_1}$ are simultaneous eigenstates of $F^{\otimes L}(G)$ and $H^{(L)}$ and the eigenvalues satisfy 
	\begin{align}
		&\forall g\in G\ \braket{\Psi_0|F^{\otimes L}(g)|\Psi_0}=\braket{\Psi_1|F^{\otimes L}(g)|\Psi_1},\\
		&\braket{\Psi_0|H^{(L)}|\Psi_0}=D\bm{n}_0\cdot\bm{E}\neq D\bm{n}_1\cdot\bm{E}=\braket{\Psi_1|H^{(L)}|\Psi_1}.
	\end{align}
\end{proof}

Lemma~\ref{lem:Psiexistenceforfinite} guarantees that there exist states $\ket{\Psi_0}$ and $\ket{\Psi_1}$ in Propositions~\ref{prop:prefiniteGEderivation} and \ref{prop:finiteGEderivation}, when the symmetry group $G$ is a finite group and the Hamiltonian $H$ is not trivial.

\begin{lemma} \label{lem:Psiexistenceforfinite}
	Let $G$ be a finite group with a unitary representation $F$ acting on $\mathcal{H}$, $H\in\mathcal{B}^\mathrm{H}(\mathcal{H})$ commute with $F(G)$ and be not trivial.
	Then, there exist $L\in\mathbb{N}$ and a pair of simultaneous eigenstates $\ket{\Psi_0},\ \ket{\Psi_1}\in\mathcal{H}^{\otimes L}$ of $F^{\otimes L}(G)$ and $H^{(L)}$ such that the eigenvalues satisfy $\braket{\Psi_0|F^{\otimes L}(g)|\Psi_0}=\braket{\Psi_1|F^{\otimes L}(g)|\Psi_1}=1$ for all $g\in G$ but $\braket{\Psi_0|H^{(L)}|\Psi_0}\neq\braket{\Psi_1|H^{(L)}|\Psi_1}$.
\end{lemma}

\begin{proof}
	Since $H$ is not trivial, we can take $k_0, k_1\in K$ satisfying $E_{k_0}\neq E_{k_1}$, where $K:=\{(\lambda, j)\ |\ \lambda\in\Lambda_F, j=1, \cdots, m_\lambda\}$, $E_k:=E_{\lambda j}$ for $k=(\lambda, j)\in K$ and $E_{\lambda j}$ is defined in Eq.~\eqref{eq:Hamiltoniandecomp2}.
	We define $\bm{n}_0=(n_{0k})_{k\in K}, \bm{n}_1=(n_{1k})_{k\in K}\in\mathbb{N}^K$ and $\ket{\Psi_0}, \ket{\Psi_1}\in\mathcal{H}^{\otimes L}$ as $n_{ik}:=\delta_{k k_i}$ and $\ket{\Psi_i}:=\ket{\Phi(|G|\bm{n}_i)}$ with $\ket{\Phi(\bm{n})}$ defined as Eq.~\eqref{eq:determinantstate}, where $L:=D|G|$. 
	In the same way as $\ket{\Psi'_0}$ and $\ket{\Psi'_1}$ in Proposition~\ref{prop:finiteGEderivation}, $\ket{\Psi_0}$ and $\ket{\Psi_1}$ are simultaneous eigenstates of $F^{\otimes L}(G)$ and $H^{(L)}$ and the eigenvalues satisfy
	\begin{align}
		&\forall g\in G\ \braket{\Psi_0|F^{\otimes L}(g)|\Psi_0}=\braket{\Psi_1|F^{\otimes L}(g)|\Psi_1}=1,\\
		&\braket{\Psi_0|H^{(L)}|\Psi_0}=D|G|E_{k_0}\neq D|G|E_{k_1}=\braket{\Psi_1|H^{(L)}|\Psi_1}.
	\end{align}
\end{proof}

In Lemma~\ref{lem:U1extractedwork}, we prove the properties of the mapping $\mathcal{O}$ defined as Eq.~\eqref{eq:unitary}.
$\mathcal{O}$ is useful for the construction of unitary operators that extract positive work in the proofs of Propositions~\ref{prop:workextop1} and \ref{prop:virtualtemperature}.

\begin{lemma} \label{lem:U1extractedwork}
	Let $N\in\mathbb{N}$, $\epsilon\in\mathbb{R}$, $C_{ij}\in\mathcal{B}(\mathcal{H}^{\otimes N})$ for $i, j\in\{0, 1\}$ and $\mathcal{O}$ be defined as Eq.~\eqref{eq:unitary} .
	If $\{C_{ij}\}_{i, j\in\{0, 1\}}$ satisfies
	\begin{align}
		\forall i, j, k, l\in\{0,1\}\ C_{ij}^\dagger=C_{ji},\ C_{ij}C_{kl}=\delta_{jk}C_{il}, \label{eq:Ccond1}
	\end{align}
	then $\mathcal{O}(\{C_{ij}\})$ is unitary.
	In addition, if $\{C_{ij}\}$ also satisfies
	\begin{align}
		\forall i, j\in\{0,1\}\ [H^{(N)}, C_{ij}]=\epsilon(i-j)C_{ij}, \label{eq:Ccond2}
	\end{align}
	then
	\begin{align}
		&W(\rho^{\otimes N}, H^{(N)}, \mathcal{O}(\{C_{ij}\}))=\epsilon(\mathrm{tr}(\rho^{\otimes N}C_{11})-\mathrm{tr}(\rho^{\otimes N}C_{00})), \label{eq:extractedwork1}\\
		&[\mathcal{O}(\{C_{ij}\})^\dag H^{(N)}\mathcal{O}(\{C_{ij}\}), H^{(N)}]=0.
	\end{align}
\end{lemma}

\begin{proof}
	We define 
	\begin{align}
		\Pi:=\frac{1}{2}\sum_{i,j\in\{0,1\}} (-1)^{i-j} C_{ij} \label{eq:projectionforunitary}
	\end{align} 
	and show that $\Pi$ is a projection operator.
	If $\{C_{ij}\}$ satisfies Eq.~\eqref{eq:Ccond1},
	\begin{align}
		\Pi^\dagger=\frac{1}{2}\sum_{i,j\in\{0,1\}} (-1)^{i-j} C_{ij}^\dag=\frac{1}{2}\sum_{i,j\in\{0,1\}} (-1)^{j-i} C_{ji}=\Pi,
	\end{align}
	and thus $\Pi$ is Hermitian.
	In addition, since 
	\begin{align}
		\Pi^2=&\left(\frac{1}{2}\sum_{i,j\in\{0,1\}} (-1)^{i-j} C_{ij}\right)\left(\frac{1}{2}\sum_{k,l\in\{0,1\}} (-1)^{k-l} C_{kl}\right) \nonumber\\
		=&\frac{1}{4}\sum_{i,j,k,l\in\{0,1\}} (-1)^{i-j+k-l} C_{ij}C_{kl} \nonumber\\
		=&\frac{1}{4}\sum_{i,j,k,l\in\{0,1\}} (-1)^{i-j+k-l} \delta_{jk}C_{il} \nonumber\\
		=&\frac{1}{4}\left(\sum_{j,k\in\{0,1\}} (-1)^{-j+k} \delta_{jk}\right)\left(\sum_{i,l\in\{0,1\}} (-1)^{i-l} C_{il}\right) \nonumber\\
		=&\frac{1}{4}\cdot 2\sum_{i,l\in\{0,1\}} (-1)^{i-l} C_{il} \nonumber\\
		=&\Pi,
	\end{align}
	$\Pi$ is a projection operator.
	Therefore, $\mathcal{O}(\{C_{ij}\})=I-2\Pi$ is unitary and Hermitian.

	If $\{C_{ij}\}$ satisfies Eq.~\eqref{eq:Ccond2}, 
	\begin{align}
		[H^{(N)},\mathcal{O}(\{C_{ij}\})]=-\epsilon\sum_{i,j\in\{0,1\}} (i-j)(-1)^{i-j} C_{ij}. \label{eq:HO_commutation}
	\end{align}
	Then we obtain
	\begin{align}
		&[H^{(N)},\mathcal{O}(\{C_{ij}\})]2\Pi \nonumber\\
		=&-\epsilon\left(\sum_{i,j\in\{0,1\}} (i-j)(-1)^{i-j} C_{ij}\right)\left(\sum_{k,l\in\{0,1\}} (-1)^{k-l} C_{kl}\right) \nonumber\\
		=&-\epsilon\sum_{i,j,k,l\in\{0,1\}} (i-j)(-1)^{i-j+k-l} C_{ij}C_{kl} \nonumber\\
		=&-\epsilon\sum_{i,j,k,l\in\{0,1\}} (i-j)(-1)^{i-j+k-l} \delta_{jk}C_{il} \nonumber\\
		=&-\epsilon\sum_{i,j,l\in\{0,1\}} (i-j)(-1)^{i-l} C_{il} \nonumber\\
		=&-\epsilon\sum_{i,l\in\{0,1\}} (2i-1)(-1)^{i-l} C_{il} \nonumber\\
		=&-\epsilon\sum_{i,j\in\{0,1\}} (2i-1)(-1)^{i-j} C_{ij}. \label{eq:HO_commutation_Pi}
	\end{align}
	From the Hermiticity of $\mathcal{O}(\{C_{ij}\})$ and Eqs.~\eqref{eq:HO_commutation} and \eqref{eq:HO_commutation_Pi},
	\begin{align}
		&H^{(N)}-\mathcal{O}(\{C_{ij}\})^\dagger H^{(N)} \mathcal{O}(\{C_{ij}\}) \nonumber\\
		=&[H^{(N)},\mathcal{O}(\{C_{ij}\})^\dagger]\mathcal{O}(\{C_{ij}\}) \nonumber\\
		=&[H^{(N)},\mathcal{O}(\{C_{ij}\})](I-2\Pi) \nonumber\\
		=&-\epsilon\sum_{i,j\in\{0,1\}} (1-i-j)(-1)^{i-j} C_{ij} \nonumber\\
		=&\epsilon(C_{11}-C_{00}). \label{eq:Hdifference}
	\end{align}
	Therefore, we get
	\begin{align}
		&W(\rho^{\otimes N}, H^{(N)}, \mathcal{O}(\{C_{ij}\})) \nonumber\\
		=&\mathrm{tr}\left(\rho^{\otimes N}\left(H^{(N)}-\mathcal{O}(\{C_{ij}\})^\dagger H^{(N)} \mathcal{O}(\{C_{ij}\})\right)\right) \nonumber\\
		=&\epsilon(\mathrm{tr}(\rho^{\otimes N}C_{11})-\mathrm{tr}(\rho^{\otimes N}C_{00})).
	\end{align}
	From Eq.~\eqref{eq:Ccond2}, $[C_{ii}, H^{(N)}]=0$.
	From this and Eq.~\eqref{eq:Hdifference}, we obtain $[\mathcal{O}(\{C_{ij}\})^\dagger H^{(N)}\mathcal{O}(\{C_{ij}\}), H^{(N)}]=[H^{(N)}-\epsilon(C_{11}-C_{00}), H^{(N)}]=0$.
\end{proof}

In Lemma~\ref{lem:SCFQcorrespondence}, we confirm that Eq.~\eqref{eq:SCFQcorrespondence} can be regarded as the definition of a mapping from $\mathcal{U}(\mathcal{H})$ to $\mathcal{U}^\mathrm{WS}(\mathcal{H}, \mathcal{H}_\mathrm{W})$. 
The commutativity of $V$ defined as Eq.~\eqref{eq:Vdef} with $H\otimes I_{\mathcal{H}_\mathrm{W}}+I_\mathcal{H}\otimes x$ is shown in Ref.~\cite{Kitaev2004}, but we also show it here for the sake of self-containedness.

\begin{lemma} \label{lem:SCFQcorrespondence}
	Let $H\in\mathcal{B}^\mathrm{H}(\mathcal{H})$, $U\in\mathcal{U}(\mathcal{H})$ and $V$ be defined as 
	\begin{align}
		V:=\int_{-\infty}^\infty dq\ e^{\mathrm{i}qH}Ue^{-\mathrm{i}qH}\otimes\ket{q}_{p\ p}\bra{q}, \label{eq:Vdef}
	\end{align}
	where $\ket{q}_p$ is the momentum eigenstate of the work storage with eigenvalue $q$.
	Then $V\in\mathcal{U}^\mathrm{WS}(\mathcal{H}, \mathcal{H}_\mathrm{W})$.
\end{lemma}

\begin{proof}
	For any $\theta\in\mathbb{R}$, 
	\begin{align}
		e^{\mathrm{i}\theta(H\otimes I+I\otimes x)}V e^{-\mathrm{i}\theta(H\otimes I+I\otimes x)}=&\int_{-\infty}^\infty dq\ e^{\mathrm{i}(q+\theta)H}Ue^{-\mathrm{i}(q+\theta)H}\otimes e^{\mathrm{i}\theta x}\ket{q}_{p\ p}\bra{q}e^{-\mathrm{i}\theta x} \nonumber\\
		=&\int_{-\infty}^\infty dq\ e^{\mathrm{i}(q+\theta)H}Ue^{-\mathrm{i}(q+\theta)H}\otimes \ket{q+\theta}_{p\ p}\bra{q+\theta} \nonumber\\
		=&\int_{-\infty}^\infty dq\ e^{\mathrm{i}qH}Ue^{-\mathrm{i}qH}\otimes\ket{q}_{p\ p}\bra{q} \nonumber\\
		=&V.
	\end{align}
	By taking the derivative at $\theta=0$, we get $[V, H\otimes I+I\otimes x]=0$.
	Since $[\ket{q}_{p\ p}\bra{q}, p]=0$ for all $q\in\mathbb{R}$, $V$ satisfies $[V, I\otimes p]=0$. 
	Since $U$ is unitary, 
	\begin{align}
		V^\dag V=&\int_{-\infty}^\infty dq\int_{-\infty}^\infty dr\ e^{\mathrm{i}qH}U^\dagger e^{-\mathrm{i}qH}e^{\mathrm{i}rH}Ue^{-\mathrm{i}rH}\otimes \ket{q}_{p\ p}\braket{q|r}_{p\ p}\bra{r} \nonumber\\
		=&\int_{-\infty}^\infty dq\ e^{\mathrm{i}qH}U^\dagger e^{-\mathrm{i}qH}e^{\mathrm{i}qH}Ue^{-\mathrm{i}qH}\otimes \ket{q}_{p\ p}\bra{q} \nonumber\\
		=&I,
	\end{align}
	and we can also prove that $VV^\dag=I$ in the same way.
	This implies that $V$ is unitary. 
	Therefore, $V\in\mathcal{U}^\mathrm{WS}(\mathcal{H}, \mathcal{H}_\mathrm{W})$.
\end{proof}

In Lemma~\ref{lem:SCFQbijection}, we prove that the mapping $\mathcal{C}$ defined as Eq.~\eqref{eq:SCFQcorrespondence} is bijective.

\begin{lemma} \label{lem:SCFQbijection}	
	The mapping $\mathcal{C}$: $\mathcal{U}(\mathcal{H})\to\mathcal{U}^\mathrm{WS}(\mathcal{H}, \mathcal{H}_\mathrm{W})$ defined as Eq.~\eqref{eq:SCFQcorrespondence} is bijective.
\end{lemma}

\begin{proof}
	First, we show that $\mathcal{C}$ is injective.
	Take arbitrary $U_1, U_2\in\mathcal{U}(\mathcal{H})$ that satisfy $\mathcal{C}(U_1)=\mathcal{C}(U_2)$, i.e., 
	\begin{align}
		\int_{-\infty}^\infty dq\ e^{\mathrm{i}qH}U_1 e^{-\mathrm{i}qH}\otimes\ket{q}_{p\ p}\bra{q}=\int_{-\infty}^\infty dq\ e^{\mathrm{i}qH}U_2 e^{-\mathrm{i}qH}\otimes\ket{q}_{p\ p}\bra{q}.
	\end{align}
	By comparing the both sides of this equation, we get $e^{\mathrm{i}qH}U_1 e^{-\mathrm{i}qH}=e^{\mathrm{i}qH}U_2 e^{-\mathrm{i}qH}$ and thus $U_1=U_2$.
	This implies that $\mathcal{C}$ is injective.
	
	Next, we show that $\mathcal{C}$ is surjective.
	Take arbitrary $V\in\mathcal{U}^\mathrm{WS}(\mathcal{H}, \mathcal{H}_\mathrm{W})$. 
	Since $V$ satisfies $[V, I\otimes p]=0$, $V$ can be written as 
	\begin{align}
		V=\int_{-\infty}^\infty dq\ U(q)\otimes\ket{q}_{p\ p}\bra{q}
	\end{align}
	with $U(q)\in\mathcal{B}(\mathcal{H})$. Since $V$ satisfies $[V, H\otimes I+I\otimes x]=0$, for any $\theta\in\mathbb{R}$, 
	\begin{align}
		\int_{-\infty}^\infty dq\ U(q)\otimes\ket{q}_{p\ p}\bra{q}=&V \nonumber\\
		=&e^{\mathrm{i}\theta(H\otimes I+I\otimes x)}Ve^{-\mathrm{i}\theta(H\otimes I+I\otimes x)} \nonumber\\
		=&\int_{-\infty}^\infty dq\ e^{\mathrm{i}(q+\theta)H}U(q)e^{-i(q+\theta)H}\otimes e^{\mathrm{i}\theta x}\ket{q}_{p\ p}\bra{q}e^{-i\theta x} \nonumber\\
		=&\int_{-\infty}^\infty dq\ e^{\mathrm{i}(q+\theta)H}U(q)e^{-\mathrm{i}(q+\theta)H}\otimes \ket{q+\theta}_{p\ p}\bra{q+\theta} \nonumber\\
		=&\int_{-\infty}^\infty dq\ e^{\mathrm{i}qH}U(q-\theta)e^{-\mathrm{i}qH}\otimes\ket{q}_{p\ p}\bra{q}.
	\end{align}
	By comparing the both sides of this equation, we get $U(q)=e^{\mathrm{i}qH}U(q-\theta)e^{-\mathrm{i}qH}$.
	By substituting $\theta=q$, we get $U(q)=e^{\mathrm{i}qH}U(0)e^{-\mathrm{i}qH}$.
	Therefore, $V$ can be written as 
	\begin{align}
		V=\int_{-\infty}^\infty dq\ e^{\mathrm{i}qH}U(0)e^{-\mathrm{i}qH}\otimes\ket{q}_{p\ p}\bra{q}.
	\end{align}
	Since $V$ is unitary,
	\begin{align}
		&\int_{-\infty}^\infty dq\ e^{\mathrm{i}qH}Ie^{-\mathrm{i}qH}\otimes\ket{q}_{p\ p}\bra{q} \nonumber\\
		=&I \nonumber\\
		=&V^\dag V \nonumber\\
		=&\int_{-\infty}^\infty dq\int_{-\infty}^\infty dr\ e^{\mathrm{i}qH}U(0)^\dag e^{-\mathrm{i}qH}e^{\mathrm{i}rH}U(0)e^{-\mathrm{i}rH}\otimes\ket{q}_{p\ p}\braket{q|r}_{p\ p}\bra{r} \nonumber\\
		=&\int_{-\infty}^\infty dq\ e^{\mathrm{i}qH}U(0)^\dagger U(0)e^{-\mathrm{i}qH}\otimes\ket{q}_{p\ p}\bra{q}.
	\end{align}
	By comparing both sides of this equation, we get $U(0)^\dag U(0)=I$. 
	We can prove that $U(0)U(0)^\dag=I$ in the same way.
	This implies that $U(0)$ is unitary. 
	Therefore, for any $V\in\mathcal{U}^\mathrm{WS}(\mathcal{H}, \mathcal{H}_\mathrm{W})$, $V$ can be written as $\mathcal{C}(U(0))$ with $U(0)\in\mathcal{U}(\mathcal{H})$.
	This implies that $\mathcal{C}$ is surjective.
	Therefore, $\mathcal{C}$ is bijective.
\end{proof}

Lemma~\ref{lem:GFSCFQcorrespondence} states that the mapping $\mathcal{C}$ gives a one-to-one correspondence from $\mathcal{U}_{G, F}(\mathcal{H})$ to $\mathcal{U}_{G, F}^\mathrm{WS}(\mathcal{H}, \mathcal{H}_\mathrm{W})$.

\begin{lemma} \label{lem:GFSCFQcorrespondence}
	Let $G$ be a group with a unitary representation $F$ acting on $\mathcal{H}$.
	The mapping $\mathcal{C}$: $\mathcal{U}(\mathcal{H})\to\mathcal{U}^\mathrm{WS}(\mathcal{H}, \mathcal{H}_\mathrm{W})$ defined as Eq.~\eqref{eq:SCFQcorrespondence} satisfies $\mathcal{C}(\mathcal{U}_{G, F}(\mathcal{H}))=\mathcal{U}_{G, F}^\mathrm{WS}(\mathcal{H}, \mathcal{H}_\mathrm{W})$.
\end{lemma}

\begin{proof}
	For any $U\in\mathcal{U}(\mathcal{H})$ and $g\in G$, we obtain
	\begin{align}
		(F(g)\otimes I)^\dag \mathcal{C}(U)(F(g)\otimes I)=&\int_{-\infty}^\infty dq\ F(g)^\dag e^{\mathrm{i}qH}Ue^{-\mathrm{i}qH}F(g)\otimes\ket{q}_{p\ p}\bra{q} \nonumber\\
		=&\int_{-\infty}^\infty dq\ e^{\mathrm{i}qH}F(g)^\dag UF(g)e^{-\mathrm{i}qH}\otimes\ket{q}_{p\ p}\bra{q} \nonumber\\
		=&\mathcal{C}(F(g)^\dag UF(g)). \label{eq:GFSCFQcorresp}
	\end{align}
	For any $U\in\mathcal{U}_{G, F}(\mathcal{H})$, from Eq.~\eqref{eq:GFSCFQcorresp}, $(F(g)\otimes I)^\dag\mathcal{C}(U)(F(g)\otimes I)=\mathcal{C}(F(g)^\dag UF(g))=\mathcal{C}(U)$ for all $g\in G$, and thus $\mathcal{C}(U)\in\mathcal{U}_{G, F}^\mathrm{WS}(\mathcal{H}, \mathcal{H}_\mathrm{W})$.
	On the other hand, for any $V\in\mathcal{U}_{G, F}^\mathrm{WS}(\mathcal{H}, \mathcal{H}_\mathrm{W})$, from Eq.~\eqref{eq:GFSCFQcorresp}, $F(g)^\dag\mathcal{C}^{-1}(V)F(g)=\mathcal{C}^{-1}((F(g)\otimes I)^\dag V(F(g)\otimes I))=\mathcal{C}^{-1}(V)$ for all $g\in G$, and thus $\mathcal{C}^{-1}(V)\in\mathcal{U}_{G, F}(\mathcal{H})$. 
	Therefore, $\mathcal{C}$ satisfies $\mathcal{C}(\mathcal{U}_{G, F}(\mathcal{H}))=\mathcal{U}_{G, F}^\mathrm{WS}(\mathcal{H}, \mathcal{H}_\mathrm{W})$.
\end{proof}

Lemma~\ref{lem:TSCFQcorrespondence} states that the mapping $\mathcal{C}$ gives a one-to-one correspondence from $\mathcal{U}_\mathcal{T}(\mathcal{H})$ to $\mathcal{U}_\mathcal{T}^\mathrm{WS}(\mathcal{H}, \mathcal{H}_\mathrm{W})$.

\begin{lemma} \label{lem:TSCFQcorrespondence}
	The mapping $\mathcal{C}: \mathcal{U}(\mathcal{H})\to\mathcal{U}^\mathrm{WS}(\mathcal{H}, \mathcal{H}_\mathrm{W})$ defined as Eq.~\eqref{eq:SCFQcorrespondence} satisfies $\mathcal{C}(\mathcal{U}_\mathcal{T}(\mathcal{H}))=\mathcal{U}_\mathcal{T}^\mathrm{WS}(\mathcal{H}, \mathcal{H}_\mathrm{W})$.
\end{lemma}
	
\begin{proof}
	For any $U\in\mathcal{U}(\mathcal{H})$, we obtain
	\begin{align}
		\mathcal{T}^{-1}\mathcal{C}(U)\mathcal{T}=&\int_{-\infty}^\infty dq\ \mathcal{T}^{-1}e^{\mathrm{i}qH}Ue^{-\mathrm{i}qH}\mathcal{T}\otimes\mathcal{T}^{-1}\ket{q}_{p\ p}\bra{q}\mathcal{T} \nonumber\\
		=&\int_{-\infty}^\infty dq\ e^{-\mathrm{i}qH}\mathcal{T}^{-1}U\mathcal{T}e^{\mathrm{i}qH}\otimes\ket{-q}_{p\ p}\bra{-q} \nonumber\\
		=&\int_{-\infty}^\infty dq\ e^{\mathrm{i}qH}\mathcal{T}^{-1}U\mathcal{T}e^{-\mathrm{i}qH}\otimes\ket{q}_{p\ p}\bra{q} \nonumber\\
		=&\mathcal{C}(\mathcal{T}^{-1}U\mathcal{T}). \label{eq:TSCFQcorresp}
	\end{align}
	For any $U\in\mathcal{U}_\mathcal{T}(\mathcal{H})$, from Eq.~\eqref{eq:TSCFQcorresp}, $\mathcal{T}^{-1}\mathcal{C}(U)\mathcal{T}=\mathcal{C}(\mathcal{T}^{-1}U\mathcal{T})=\mathcal{C}(U)$, and thus $\mathcal{C}(U)\in\mathcal{U}_\mathcal{T}^\mathrm{WS}(\mathcal{H}, \mathcal{H}_\mathrm{W})$.
	On the other hand, for any $V\in\mathcal{U}_\mathcal{T}^\mathrm{WS}(\mathcal{H}, \mathcal{H}_\mathrm{W})$, from Eq.~\eqref{eq:TSCFQcorresp}, $\mathcal{T}^{-1}\mathcal{C}^{-1}(V)\mathcal{T}=\mathcal{C}^{-1}(\mathcal{T}^{-1}V\mathcal{T})=\mathcal{C}^{-1}(V)$, and thus $\mathcal{C}^{-1}(V)\in\mathcal{U}_\mathcal{T}(\mathcal{H})$. 
	Therefore, $\mathcal{C}$ satisfies $\mathcal{C}(\mathcal{U}_\mathcal{T}(\mathcal{H}))=\mathcal{U}_\mathcal{T}^\mathrm{WS}(\mathcal{H}, \mathcal{H}_\mathrm{W})$.
\end{proof}


\end{document}